%% file: main.tex
\documentclass[manuscript,screen, nonacm]{acmart}

\AtBeginDocument{%
	\providecommand\BibTeX{{%
			\normalfont B\kern-0.5em{\scshape i\kern-0.25em b}\kern-0.8em\TeX}}}

\usepackage{amsmath,amsfonts}
\usepackage{algorithmic}
\usepackage{graphicx}
\usepackage{textcomp}
\usepackage{xcolor}
\usepackage{mathtools,dsfont}
\usepackage{psfrag,bbm,graphicx}
\usepackage{comment,balance}
\usepackage[title]{appendix}
\usepackage{tikz}
\usetikzlibrary{automata, positioning}
\usetikzlibrary{arrows}
\newcommand{\rom}[1]{\uppercase\expandafter{\romannumeral #1\relax}}
\usepackage{url}
\usepackage{epsfig,subfig}
\usepackage{multicol}
\usepackage{multirow}
\usepackage{amsfonts, color}
\usepackage{array}
\usepackage[linesnumbered,ruled,vlined]{algorithm2e}
\usepackage{nccmath}
\usepackage{caption}
\usepackage[belowskip=-10pt,aboveskip=0pt]{caption}
\usepackage[compact]{titlesec} 

\newtheorem{lemma}{Lemma}
\newtheorem{remark}{Remark}

\newtheorem{theorem}{Theorem}

\newtheorem{definition}{Definition}

\newtheorem{assumption}{Assumption}
\newtheorem{model}{Model}

\newcommand{\remove}[1]{}

\begin{document}
\title{Online Partial Service Hosting at the Edge
}
\author{V S Ch Lakshmi Narayana}
\affiliation{\institution{Department of Electrical Engineering, IIT Bombay}}
\email{lakshmi.n.borusu@gmail.com}

\author{Mohit Agarwala}
\affiliation{\institution{Department of Electrical Engineering, IIT Bombay}}
\email{mohit.496.ece@gmail.com}

\author{R Sri Prakash}
\affiliation{\institution{Department of Electrical Engineering, IIT Bombay}}
\email{prakash.14191@gmail.com}

\author{Nikhil Karamchandani}
\affiliation{\institution{Department of Electrical Engineering, IIT Bombay}}
\email{nikhil.karam@gmail.com}

\author{Sharayu Moharir}
\affiliation{\institution{Department of Electrical Engineering, IIT Bombay}}
\email{sharayu.moharir@gmail.com}

\begin{abstract}
We consider the problem of service hosting where an application provider can dynamically rent edge computing resources and serve user requests from the edge to deliver a better quality of service. A key novelty of this work is that we allow the service to be hosted partially at the edge which enables a fraction of the user query to be served by the edge. We model the total cost for (partially) hosting a service at the edge as a combination of the latency in serving requests, the  bandwidth  consumption,  and  the  time-varying  cost  for renting edge resources. We  propose  an  online  policy  called $\alpha$-RetroRenting ($\alpha$-RR) which  dynamically  determines  the  fraction  of the  service  to  be  hosted  at  the  edge  in  any  time-slot, based on the history of the request arrivals and the rent cost  sequence. As our main result, we derive  an upper bound on $\alpha$-RR's competitive ratio with respect to the  offline  optimal  policy  that  knows  the  entire request  arrival  and  rent  cost  sequence  in  advance.  In addition, we provide performance guarantees for our policy in the setting where the request arrival process is stochastic.  We conduct  extensive  numerical  evaluations  to compare  the  performance  of $\alpha$-RR  with  various  benchmarks for  synthetic  and  trace-based  request  arrival  and  rent cost processes, and find several parameter regimes where $\alpha$-RR's  ability  to  store  the  service  partially  greatly improves cost-efficiency.
\end{abstract}

\maketitle
\section{Introduction}
\input{introduction.tex}
\input{related_work.tex}

\section{System Setup}
\label{sec:setting}
\input{setting_retro.tex}
\section{Our Hosting Policy}
\input{optimal_RR.tex}

\section{Analytical Results: Adversarial Setting}
\label{sec:analyticalResults_adversarial}
\input{analytical_results_adversarial}

\section{Analytical Results: Stochastic Setting}
\label{sec:analyticalResults_stochastic}
\input{analytical_RR_Stochastic}

\section{Simulation Results: Model 1}
\input{rough2.tex}

\section{Simulation Results: Model 2}
\input{rough.tex}

\
\section{Conclusions}
\input{conclusions.tex}

\bibliographystyle{ACM-Reference-Format}
\bibliography{main.bbl}

 \section{Appendix A}
 \label{sec:proofs_adversarial}
 \input{appendix_proofs_adversarial.tex}

\section{Appendix B}
\label{sec:proofs_stochastic}
\input{appendix_proofs_stochastic.tex}


\end{document}

%% file: introduction.tex
 The emergence of services based on machine learning, computer vision, and augmented/virtual reality (AR/VR) for resource-constrained handheld devices is testing the limits of what traditional cloud-computing platforms can reliably support in terms of the required latency and bandwidth. This has led to the advent of edge computing wherein application providers can dynamically rent storage/computing resources much closer to the end-users via short-term rent contracts, and has also spurred a lot of academic research into the design and implementation of cost-efficient dynamic algorithms for offloading computation tasks \cite{zhao2018red, xia2020online, yang2015cost}. 

With respect to the prior work on service hosting, the key novelty of this work is the option to partially host the service at the edge. We say that a service is partially hosted when a fraction of the database and code of the service is hosted at the edge. The partially hosted service can be used to compute parts of the answers to user queries and deliver to the user with low latency. The rest of the answer is computed and delivered from the cloud at higher latency. Partial hosting requires lower edge resources than completely hosting the service at the edge and therefore can potentially lead to a reduction in the cost incurred for renting edge resources.  To understand the potential benefits of partial hosting for existing services, we use a GPS trajectory dataset \cite{zheng2008understanding, zheng2009mining, zheng2010geolife} collected as a part of the Geolife Project by Microsoft Research Asia to characterize the fraction of requests that can be served at the edge as a function of the fraction of the service hosted at the edge. The details of this dataset and our inferences from it are discussed in the simulations section.

We consider two models of serving requests when the service is partially hosted at the edge.  In the first model, the service can be partitioned in a way such that each partition can generate a partial response to a user's query/request and this partial response is of independent value to the user. For instance, for a navigation service like Google Maps, one part of the service could compute the possible routes between a source-destination pair, and the other part could compute the travel time for these routes based on current traffic patterns. Another example is online writing assistant services like Grammarly where one part of the service spellchecks and the other part checks grammar. One more example is a news website that uses the edge to deliver the text corresponding to news articles at low latency to the user and fetches the images/videos embedded in the article from the cloud at high latency. In such cases, under partial hosting, the user can start reading while the images/videos load from the cloud servers.  When the service is partially hosted, we assume that the fraction of the response to any user query served by the edge is a non-decreasing function of the fraction of the service stored. 

In the second model, if a service is partially hosted at the edge, some of the requests can be served by the edge, while the others have to be served by the cloud servers. The fraction of requests that can be served at the edge is an increasing function of the fraction of service hosted at the edge. For example, for a service like Google Translate, under partial hosting, the edge can be equipped to handle translation requests for some language pairs, while the rest have to be served by the cloud servers.

\subsection{Main Contributions}
In this work, we restrict the discussion to the setting which allows one partial hosting level $\alpha$, apart from the option of not hosting and complete hosting of the service. We model the total cost for (partially) hosting a service at the edge as a combination of the latency in serving requests, the bandwidth consumption, and the time-varying cost for renting edge resources which is assumed to scale linearly with the fraction of the service cached. The broad goal in this work is to design cost-efficient schemes which dynamically decide when and what fraction of the service to host at the edge. The main contributions of our work are as follows.
\begin{enumerate}
\item[--] We propose an online policy called $\alpha$-RetroRenting ($\alpha$-RR) which dynamically determines the fraction of the service to be hosted at the edge in any time-slot, based on the history of the request arrivals and the rent cost sequence. 
\item[--] We compare the performance of $\alpha$-RR with the offline optimal policy which knows the entire request arrival and rent cost sequence in advance. We characterize conditions under which $\alpha$-RR is optimal and show that in the worst-case, $\alpha$-RR is $6$-optimal.
\item[--]  In the setting where request arrivals are i.i.d. stochastic, we provide performance guarantees for $\alpha$-RR and compare its performance to that of the optimal online policy.
\item[--] We characterize conditions under which the offline optimal policy and $\alpha$-RR do not use partial hosting. This result can be used by service designers as a guideline on how to partition their service so that partial hosting can be used effectively to improve the performance of the system.
\item[--] We characterize a fundamental limit on the performance of any online deterministic policy, which helps benchmark the performance of the $\alpha$-RR scheme proposed in this work.
\item[--] Finally, we conduct extensive numerical evaluations to compare the performance of $\alpha$-RR with various policies including those that do not use partial hosting for synthetic and trace-based request arrival and rent cost processes. We find that there are several parameter regimes where $\alpha$-RR outperforms other baselines and the ability to store the service partially can greatly improve cost-efficiency.
\end{enumerate}

%% file: related_work.tex
\subsection{Related Work}
With applications based on the Internet of Things, AR/VR, and large-scale machine learning becoming more mainstream, various edge computing platforms and architectures have been developed \cite{Puliafito:2019, mao2017survey, mach2017mobile} which can reliably support the stringent latency and bandwidth requirements. There has also been a large amount of academic research on such systems, amongst which that focusing on the design and analysis of efficient task offloading algorithms is the most relevant to our work and we discuss some of it below. 

The broad setting where our work is placed is when there are one or more edge servers which assist clients in carrying out computation tasks and the goal is to determine which tasks to offload to the edge server(s) so that the overall cost is minimized. One approach towards the design of such schemes is to formulate the problem as a large-scale one-shot optimization problem \cite{pasteris2019service, bi2019joint, chen2017collaborative, tran2019costa, yang2015cost, xu2020collaborate, 9326369}. While solving the problem exactly turns out to be NP-hard in several instances, efficient heuristics are presented. Recently, \cite{9326369} studied the impact of decentralization on the performance of such task offloading schemes. Finally, \cite{9362261} used a game-theoretic approach to study the problem by considering a two-stage interactive game between an edge server and the users, wherein the server announces prices for hosting various services in the first stage and each user independently makes its offloading decision in the second round. Our work differs from this line of work in that we design online algorithms which adapt their service placement decisions dynamically over time depending on the varying number of requests and rental costs.

Several works consider stochastic models for incorporating time-varying requests \cite{xu2018joint, chen2019budget, wang2015dynamic, he2020bandit,xiong2101learning}, using frameworks such as Markov Decision Processes (MDPs), Reinforcement Learning and Multi-Armed Bandits to address scenarios where system parameters such as service popularity or task request rates are unknown. On the other hand, we focus on the case of arbitrary request arrival processes and provide `worst-case' guarantees on the performance of our proposed schemes instead of `average' performance guarantees. This can be vital in scenarios where the arrival patterns change frequently over time, making it difficult to predict demand or model it well as a stochastic process.

 The key distinguishing feature of our work with respect to almost all the literature is that we allow partial service hosting at the edge which allows part of the query response to be provided at low latency by the edge server and the rest at higher latency by the back-end cloud server. As we will see later, this can potentially lead to significant cost benefits for the application provider. While \cite{9155271} also considered partial service hosting, its focus was on stochastic request processes and expected cost, whereas we consider the `adversarial' setting with worst-case cost guarantees.   In \cite{narayana2021renting}, the possible hosting options are limited to either fully hosting the service at the edge or not hosting it at the edge at all. In this work, we generalize this setting to include a third intermediate/partial hosting level. The details are discussed in the next section. This key difference necessitates new algorithm design and performance analysis.


In \cite{zhao2018red}, the authors consider a system with multiple services and an edge server with limited capacity, and proposed an online scheme named ReD/LeD for deciding which services to host on the edge at each time whose performance is characterized in terms of the competitive ratio with respect to an oracle which knows the entire request sequence in advance.
The competitive ratio analysis was recently extended to a system with multiple connected caches in \cite{tan2021asymptotically}. Unlike \cite{zhao2018red,tan2021asymptotically} and other works mentioned above which optimize the system from the perspective of the access providers which rent out edge computing resources, we study the problem from the perspective of an application provider which provides a service to the end users. Other works which have studied service hosting at the edge from this perspective include \cite{9155228, 9155271,xia2020online} where the goal is to minimize the cost incurred by the application provider while serving the user requests. Although we study the problem from a perspective of a specific application provider, the effect of the presence of other application providers who might be simultaneously interested in using the potentially limited edge resources offered by the access providers is captured through the time-varying nature of the cost of renting edge resources. 

While the problem of service hosting does resemble the long-studied content caching problem \cite{borst2010distributed, tan2012optimal, Wolman99, breslau1999web, sleator1985amortized, belady1966study}, there are important differences. In particular, unlike the content caching problem, whenever a user request cannot be served at the edge there is an option to either simply forward that request or download the entire service at higher cost. This fundamentally changes the problem and in fact, prior work \cite{zhao2018red, 9155228} has demonstrated the sub-optimal performance for service hosting of several popular schemes which work well in the traditional content caching setting. Finally, we would like to point out that partial storage of files has been studied in the context of traditional content caching \cite{borst2010distributed,hefeeda2008traffic}. Here, the fraction of query served varies linearly with the fraction of the file stored, whereas the dependence can be more varied for the service caching problem studied in this work. We indeed observe such behavior for a shortest path query system we design using data from a GPS trajectory dataset \cite{zheng2008understanding, zheng2009mining, zheng2010geolife}. Details can be found in the simulations section.

%% file: setting_retro.tex
\subsection{Network Model}\label{syst:nwmodel}
We study a system consisting of one or more cloud servers and an edge server in proximity to the customers/users of a service. 
The cloud servers always host the service and can serve all requests that are routed to them. In addition, the service can also be hosted at the edge server to serve user requests by paying a rent cost for using edge resources. We allow for partial hosting at the edge, i.e., only a part of the service can be hosted at the edge. Details of how requests are served when only a part of the service is hosted at the edge are discussed in subsequent sections.

 
\subsection{Request Arrivals}
We consider a time-slotted system with  two different arrival patterns viz., adversarial and stochastic request arrivals.  Let $x_t$ denote the number of request arrivals in a time-slot $t$. For our analytical results for adversarial arrivals, we make the following assumptions on the arrival sequence. 
\begin{assumption}\label{assum_one_request}
At most one request arrives in each time-slot, i.e., $x_t \in \{0,1\}$.
\end{assumption}

Further, for our analytical results for stochastic arrivals, we make the following additional assumption. 
\begin{assumption}\label{assum_stochastic}
Request arrivals are i.i.d. across time-slots. 
\end{assumption} 

\begin{remark}
Our results can be extended to the setting considered in \cite{9155228} with potentially multiple request arrivals in a time-slot and an upper bound on the number of requests that can be served at the edge in a time-slot. In simulations, we  consider different cases where the request arrival process is Adversarial, Poisson, and Markovian.
\end{remark}

\subsection{Renting Edge Resources}
Edge resources can be rented from a third party edge resource provider by paying a rent cost.
The rent cost for a time-slot is determined and advertised by the third party provider. Let $c_t$ denote the cost of hosting the entire service at the edge in time-slot $t$. If the service is partially hosted at the edge in a time-slot, the rent cost is scaled proportional to the fraction of service hosted at the edge. For our analytical results, we make the following assumptions on the rent cost sequence. 
\begin{assumption}\label{bounded_rent}
$0<c_{\text{min}} \leq c_t \leq c_{\text{max}}$.
\end{assumption}

\begin{remark}
The time-varying nature of the rent cost captures the effect of the presence of multiple customers of the third-party edge resource provider and the potential fluctuation in the overall demand for the edge resources. 
\end{remark}

%

\subsection{Partial Service Hosting}\label{syst:partialHosting}

As discussed above, we allow the service to be partially hosted at the edge. Let $r_t$ denote the fraction of service hosted at the edge in time-slot $t$. Recall that in this work, we restrict the discussion to the setting which allows three hosting levels. Formally, we make the following assumption.

\begin{assumption}\label{assum_hosting_levels}
In addition to the option of hosting the entire service at the edge as in \cite{zhao2018red, 9155228}, in this work, $\alpha \in (0,1)$ fraction of the service can also be hosted at the edge. It follows that $r_t = \{0, \alpha, 1\}$, where $r_t = 0$ denotes that the service is not hosted at the edge in time-slot $t$.
\end{assumption}

We consider the following two service models when the service is partially hosted at the edge.

\begin{model}[Partial Service at the Edge]
\label{model:partial}
When the service is partially hosted at the edge, an incoming request can be partially served by the edge servers, i.e., a part of the answer to the user's query can be computed at the edge. We focus on services where this partial answer is of independent interest to the user. Thus, the part of the answer computed at the edge can be communicated to the user with low latency due to the proximity of the user to the edge servers. The answer to the rest of the query is computed at the cloud and delivered to the user at high latency. 
\end{model}

\begin{model}[i.i.d. Randomized Service at the Edge] 
\label{model:random}
When the service is partially hosted at the edge, each arriving request can be served at the edge with a probability which is a non-decreasing function of the fraction of service hosted; else the request is served by the cloud servers. 
\end{model}


\subsection{Sequence of Events in a Time-slot}

In each time-slot, we first have potential request arrivals. These requests are then served by the edge/cloud servers. The third-party edge resource provider then announces the rent cost the for next time-slot. Following this, our system determines the fraction of service to be hosted at the edge in the next time-slot.


\subsection{Cost Model}\label{syst:costmodel_adversarial}

We build on the models used in \cite{9155228, zhao2018red} when the requests are adversarial and consider three categories of costs. For a given hosting policy $\mathcal{P}$, the total cost incurred in time-slot $t$, denoted by $C_t^{\mathcal{P}}$, is the sum of the three costs.

\begin{enumerate}

	\item[--] \emph{Fetch cost $(C_{F,t}^\mathcal{P})$}: This is the cost incurred to fetch the service (code and databases/libraries) from the cloud server(s) to host on the edge server. On each fetch of $\Delta_t = (r_{t+1} - r_t)^+$ fraction the service from the cloud server(s) to host on the edge-server, a fetch cost of $\Delta_t M$ units is incurred, where $\Delta_t  \in \{1,\alpha, 1-\alpha\}$.
	\item[--] \emph{Rent cost $(C_{R,t}^\mathcal{P})$}: This is the cost incurred to rent edge resources to host the service. A rent cost of $c_t r_t$  units is incurred to host $r_t$ fraction of the service on the edge server in time-slot $t$.
	
		\item[--] \emph{Service cost $(C_{S,t}^\mathcal{P})$}: This is the cost incurred per request for using the cloud servers.
	
	Under Model \ref{model:partial}, since only that part of the request which cannot be served at the edge is forwarded to the cloud servers, this cost is a decreasing function of the fraction of service hosted at the edge in that time-slot. Let $g(r_t)$ denote the cost incurred per request in time-slot $t$. We assume that
	\begin{align*}
	    g(r_t) = \begin{cases}
	    1 & r_t = 0 \\
	    g(\alpha) \in (0,1) & r_t = \alpha \\
	    0 & r_t = 1. \\
	    \end{cases}
	\end{align*}
	
	Under Model \ref{model:random}, when the entire service is hosted at the edge, each incoming request can be served
	at no cost.
	When the service is not hosted at the edge, each incoming request is forwarded to the cloud server which serves it at a cost of one unit per request.
We now discuss the case when the service is partially hosted at the edge under Model \ref{model:random}. Let $X_t$ be the number of requests received in a time-slot. Recall that $r_t$ denotes the fraction of service hosted on the edge server during time-slot $t$. When $r_t=\alpha$, each incoming request can be served at the edge with zero cost with probability $(1-g(\alpha))$ and has to be forwarded to the cloud otherwise. It follows that each arriving request incurs a cost of one unit with probability $g(\alpha)$ and zero units with probability $1-g(\alpha)$. Formally,  for  $r_t=\alpha$ and $X_t> 0$, the service cost for each request is given by,

\begin{align*}
S_t^i & =
\begin{cases*}
1 & \text{ with probability $g(\alpha)$ } \nonumber \\
0 & \text{ with probability $1-g(\alpha)$. } \nonumber \\
\end{cases*} \nonumber
\end{align*}
for $i \in \{1, \ldots, X_t\}.$
Let $S_t$ denote the total service cost in a time-slot $t$ when $r_t=\alpha$ and $X_t> 0$. Thus $$S_t=\displaystyle\sum_{i=1}^{X_t} S_t^i.$$
\end{enumerate}

It follows that
\begin{align}
\label{equation:total_cost}
C_t^\mathcal{P} &=C_{F,t}^\mathcal{P}+C_{R,t}^\mathcal{P}+C_{S,t}^\mathcal{P}, \\
\text{where, } C_{F,t}^\mathcal{P}&=
\begin{cases*}
M & \text{ if $r_{t} = 0$ and $r_{t+1} = 1$ } \nonumber\\
\alpha M & \text{ if $r_{t} = 0$ and $r_{t+1} = \alpha$ } \nonumber\\
(1-\alpha)M & \text{ if $r_{t} = \alpha$ and $r_{t+1} = 1$ } \nonumber\\
0 &\text{ otherwise.} 
\end{cases*}\\
C_{R,t}^\mathcal{P} &=
\begin{cases*}
c_t & \text{ if $r_{t} = 1$ } \nonumber \\
c_t\alpha & \text{ if $r_{t} = \alpha$ } \nonumber \\
0 & \text{ otherwise.}
\end{cases*}\\
\text{For Model \ref{model:partial}: }
C_{S,t}^\mathcal{P} & =
\begin{cases*}
0 & \text{ if $r_t = 1$ } \nonumber \\
g(\alpha)x_t & \text{ if $r_t = \alpha$ } \nonumber \\
x_t & \text{ otherwise.} 
\end{cases*} \nonumber\\
\text{For Model \ref{model:random}: }
C_{S,t}^\mathcal{P} & =
\begin{cases*}
0 & \text{ if $r_t = 1$ } \nonumber \\
S_t & \text{  if $X_t>0$ and $r_t = \alpha$} \nonumber \\
X_t & \text{ otherwise.} 
\end{cases*} \nonumber
\end{align}

Typically, the amount of data (code and databases/libraries) needed to host the service at the edge is much larger than the amount of data delivered to a user in response to a request.  Motivated by this, we make the following assumption, also made in \cite{9155228, zhao2018red}. 
	\begin{assumption}
	\label{ass:largeM}
	The cost of fetching  the service is more than the  cost incurred to use the cloud servers to answer a request, i.e., $M> 1$. 
	\end{assumption}
%
\color{black}

\subsection{Algorithmic Challenge}
The algorithmic challenge is to design a policy determines the fraction of service hosted at the edge in each time-slot. 
Hosting policies can be divided into the following two classes. 
\begin{definition}(Types of Hosting Policies)
	\label{defn:typesOfPolicies}
	\begin{enumerate}
		\item[--] \emph{Offline Policies}: A policy in this class knows the entire request arrival sequence and rent cost sequence a priori.
		\item[--] \emph{Online Policies}: A policy in this class does not have knowledge of future arrivals and rent cost sequence. 
	\end{enumerate}
\end{definition}
We design an online policy which makes hosting decisions based on the request arrivals, rent costs thus far, intermediate hosting level ($\alpha$) and the various costs, i.e., the rent cost $(c_t)$ in a time-slot $t$, the fetch cost $(M)$, and the forwarding cost $g(\alpha)$. 
\subsection{Metric and Goal}
\begin{enumerate}
    \item [--] In the adversarial setting, the optimal offline policy ($\alpha$-OPT) serves as a benchmark to evaluate the performance any online policy $\mathcal{P}$. The goal is to design an online policy $\mathcal{P}$ which minimizes the competitive ratio $\rho^{\mathcal{P}}$ defined as 
\begin{equation}
\label{eq:competitiveRatio}
 \rho^{\mathcal{P}}=\sup_{a\in \mathcal{A}, d \in \mathcal{R}} \frac{C^{\mathcal{P}}(a,d)}{C^{\text{$\alpha$-OPT}}(a,d)},
\end{equation}
where $\mathcal{A}$, $\mathcal{R}$ are the set of all possible finite request arrival sequences and the set of all possible rent cost sequences respectively. $C^{\mathcal{P}}(a, d)$, $C^{\text{$\alpha$-OPT}}(a, d)$ are the overall costs of service for the request arrival sequence $a$, the rent cost sequence $d$ under online policy $\mathcal{P}$ and the optimal offline policy respectively.

\item[--] In the stochastic setting, we compare the performance of a policy $\mathcal{P}$ with the performance of the optimal online policy ($\alpha$-OPT-ON).

The goal is to minimize $\sigma^{\mathcal{P}}_T$, defined as the ratio of the expected cost incurred by policy $\mathcal{P}$ in $T$ time-slots to that of the optimal online policy in the same time interval. Formally,
\begin{equation}
\label{eq:efficiencyRatio}
\sigma^\mathcal{P}(T) =\frac{\mathbb{E}\bigg[\displaystyle\sum_{t=1}^T C_t^\mathcal{P}\bigg]}{\mathbb{E}\bigg[\displaystyle\sum_{t=1}^T C_t^{\alpha\text{-OPT-ON}}\bigg]},
\end{equation}
where $C_t^\mathcal{P}$ is as defined in \eqref{equation:total_cost}.
\end{enumerate}

\color{black}

%% file: optimal_RR.tex
In this section, we present our online edge hosting policy called $\alpha$-RetroRenting. The high-level idea behind the policy is to evaluate if the current hosting status under $\alpha$-RetroRenting is optimal in hindsight given the knowledge of the request arrival/rent cost process up to the current time. If not,  $\alpha$-RetroRenting changes the hosting status. A formal definition is given in Algorithm \ref{algo:alphaRR} and a detailed description of the policy is as follows. 

\begin{algorithm}
	\caption{$\alpha$-RetroRenting ($\alpha$-RR)}\label{algo:alphaRR}
	\SetAlgoLined
	\SetKwFunction{FtotalCost}{totalCost}
    \SetKwProg{Fn}{Function}{:}{end}
	Input: Fetch cost $M$, partial hosting level $\alpha$, latency cost under partial hosting $g(\alpha)$,  rent cost sequence $\{c_l\}_{l\geq0}^t$, request arrival sequence $\{x_l\}_{l\geq0}^t$\\
	Output:  service hosting strategy $r_{t+1}$, $t > 0$\\
	Initialize:  $r_1= t_{\text{recent}} = 0$\\
	\For {\textbf{each} time-slot $t$}{
	$I_t=(M$, $g(\alpha)$, $t$,  $t_{\text{recent}}$, $\{c_l\}_{l\geq t_{\text{recent}}}^t$, $\{x_l\}_{l\geq t_{\text{recent}}}^t )$\\
	$R_0^{(\tau_0)} = [\underbrace{r_t,r_t,\ldots,r_t}_{\tau_0-t_{\text{recent}}},\underbrace{0,0,\ldots, 0}_{t-\tau_0}]$ \\
    $R_\alpha^{(\tau_{\alpha})}= [\underbrace{r_t,r_t,\ldots,r_t}_{\tau_{\alpha}-t_{\text{recent}}},\underbrace{\alpha,\alpha,\ldots, \alpha}_{t-\tau_{\alpha}}]$ \\
    $R_1^{(\tau_1)} = [\underbrace{r_t,r_t,\ldots,r_t}_{\tau_1-t_{\text{recent}}},\underbrace{1,1,\ldots, 1}_{t-\tau_1}]$ \\
	$\text{minCost}(0)=\displaystyle\min_{\tau_0\in (t_{\text{recent}},t)} \FtotalCost(R_0^{(\tau_0)},I_t)$\\
	$\text{minCost}(\alpha)=\displaystyle\min_{\tau_{\alpha}\in (t_{\text{recent}},t)} \FtotalCost(R_\alpha^{(\tau_{\alpha})},I_t)$\\
	$\text{minCost}(1)=\displaystyle\min_{\tau_1\in (t_{\text{recent}},t)} \FtotalCost(R_1^{(\tau_1)},I_t)$\\
	$r_{t+1} = \displaystyle \arg\min_{i \in \{0, \alpha, 1\}} \text{minCost}(i)$ \\
		\If{$r_{t+1} \neq r_t$}
		{$t_{\text{recent}} =  t$
		}
	}
	\Fn{\FtotalCost{$R, I_t$}}{
	$g(0)=1$, $g(1)=0$\;
    cost = $R(1)\times c_1+x_1 \times g(R(1))$\;
    \For{$j\gets 2$ \KwTo $t-t_{\text{recent}}$ }{
    cost = cost $+ R(j)\times c_j+x_j \times g(R(j))$ \\
    \hspace{.67in}$+ M \times \left|R(j)-R(j-1)\right|$\;
    }
        \KwRet cost;
  }
\end{algorithm}

In each time-slot, $\alpha$-RetroRenting focuses on the time-frame starting from the most recent time-slot in which the hosting status was changed under $\alpha$-RetroRenting ($t_{\text{recent}}$) to the current time-slot ($t$). It follows that the hosting status under $\alpha$-RetroRenting is constant in this time-frame. 

$\alpha$-RetroRenting then considers two alternative hosting strategies in which the hosting status is changed to one of the other two hosting levels at some point in the time-frame and remains unchanged thereafter. Lines 6-8 in Algorithm \ref{algo:alphaRR} represent these two alternative hosting strategies in addition to the hosting status under $\alpha$-RetroRenting. 

In Lines 9-11, $\alpha$-RetroRenting computes the lowest possible total cost (fetch cost + rent cost + latency cost) incurred in the time-frame under each one of the three candidate hosting strategies by optimizing the time at which the hosting status is changed in the time-frame. The function totalCost (Lines 17-25) is used to compute the total cost incurred for specific hosting, request arrival, and rent cost sequences. 

Following this, in Line 12, $\alpha$-RetroRenting sets the value of the hosting level for the next time-slot ($r_{t+1}$) to the hosting level at time $t$ in the hosting strategy which has the least cost among the three candidates.

\begin{remark}
While the computation/storage complexity of $\alpha$-RetroRenting as presented in Algorithm \ref{algo:alphaRR} can scale with time, using techniques proposed in \cite{lu2012online} and used in \cite{narayana2021renting}, both the computation/storage complexity of $\alpha$-RetroRenting can be reduced to $O(1)$. We omit the details due to lack of space. 
\end{remark}

%% file: analytical_results_adversarial
In this section, we state and discuss our analytical results. The proofs of there results are discussed in Section \ref{sec:proofs_adversarial}.

Our first result characterizes sufficient conditions under which the optimal offline policy ($\alpha$-OPT) and our policy $\alpha$-RR do not partially host the service at the edge. 

\begin{theorem}\label{thm:opt_nointermediate}
	 Consider a system satisfying Assumptions \ref{assum_one_request}-\ref{ass:largeM} and   Model \ref{model:partial}. Let $r^*_{\tilde{t}}$ and $r^{\alpha\text{-RR}}_{\tilde{t}}$ be the fraction of service hosted on the edge in time-slot $t$ under $\alpha$-OPT and $\alpha\text{-RR}$ respectively. 
	 \begin{enumerate}
	     \item[(a)] If $\alpha + g(\alpha) \geq 1$ and $r^*_{\tilde{t}} \neq \alpha$, $r^*_t \neq \alpha$, for all $t>\tilde{t}$
	     \item[(b)] If $\alpha + g(\alpha) \geq 1$, $r^{\alpha\text{-RR}}_{t} \neq \alpha$, for all $t>0$.
	 \end{enumerate}
\end{theorem}

We thus conclude that if $\alpha + g(\alpha) \geq 1$, if the service is either fully hosted or not hosted at the edge in a time-slot, the offline optimal policy does not use partial hosting in all subsequent time-slots. In addition, for $\alpha + g(\alpha) \geq 1$, our policy $\alpha$-RR never uses partial hosting.

The condition $\alpha + g(\alpha) \geq 1$ imposes an upper bound on the difference between the service cost when the service is not hosted at the edge (one unit) and the service cost under partial hosting ($g(\alpha)$ units). The take-away from the result is that if the reduction in service cost due to partial hosting is less than the fraction of service hosted under partial hosting, the offline optimal policy and our policy make limited use of partial hosting. 

Our next result provides performance guarantees for $\alpha$-RR.

\begin{theorem}
	\label{thm:RR_adv}
	Let $\rho^{\alpha\text{-RR}}$ be the competitive ratio of $\alpha$-RR policy as defined in \eqref{eq:competitiveRatio}. Under Assumptions \ref{assum_one_request}-\ref{ass:largeM}  and under Model \ref{model:partial},
	\begin{enumerate}
	\item[(a)] if  $\alpha c_{\text{min}}+g(\alpha)\geq 1$ and $c_{\text{min}}\geq 1$ then $\rho^{\alpha\text{-RR}}=1,$ 
	     \item[(b)] if $c_{\text{min}}<1$ or $\alpha c_{\text{min}}+g(\alpha)<1$, then 
	     $$\rho^{\alpha\text{-RR}}\leq 4+\dfrac{1}{M}+\max\left\{\dfrac{1}{M},\dfrac{1-g(\alpha)}{M\alpha} \right\}.$$ 
	 \end{enumerate}
\end{theorem}

This result characterizes sufficient conditions under which $\alpha$-RR is optimal, i.e., its performance matches that of the offline optimal policy which knows the entire request arrivals/rent cost sequences a priori. In addition, it provides an upper bound on the cost incurred in the worst-case. 

Recall that the cost of fetching $x \in \{\alpha, 1\}$ fraction of the service is $Mx$. Also, the service cost of using the the cloud servers to serve the part of the request corresponding to $x \in \{\alpha, 1\}$ fraction of the service is $1-g(x)$, where, by definition, $g(1) = 0$.
Typically, the answers/responses to user queries are significantly shorter than the code and database/libraries used to compute these answers. Motivated by this, we consider the following. 
\begin{assumption}
	\label{ass:largeM_a}
	For a given value of $x \in \{\alpha, 1\}$, the cost of fetching that fraction of the service (given by $Mx$) is more than service cost of using the the cloud servers to serve the part of the request corresponding to $x$ (given by $1-g(x)$). It follows that,
	$$M > \max \left\{1 , \dfrac{1-g(\alpha)}{\alpha} \right\}.$$
	\end{assumption}

	The following result provides the universal upper bound on the worst case performance of $\alpha$-RR under Assumptions \ref{assum_one_request}-\ref{ass:largeM_a}.
\begin{corollary}
	Let $\rho^{\alpha\text{-RR}}$ be the competitive ratio of $\alpha$-RR policy as defined in \eqref{eq:competitiveRatio}. Under Assumptions \ref{assum_one_request}-\ref{ass:largeM_a}  and under Model \ref{model:partial},
	if $c_{\text{min}}<1$ or $\alpha c_{\text{min}}+g(\alpha)<1$, then $\rho^{\alpha\text{-RR}}\leq 6.$ 
\end{corollary}	

\color{black}

Our next result characterizes a fundamental limit on the performance of any deterministic online policy. 
\begin{theorem}
	\label{thm:anyonline}
	Consider a system satisfying Assumptions \ref{assum_one_request}-\ref{ass:largeM}  and  Model \ref{model:partial}. Let $\mathcal{P}$ be any deterministic online policy and let $$f(u, v)=1+  \frac{(uM+uc_{\text{min}}+g(u))(1-vc_{\text{min}}-g(v))}{vM}.$$
	(a) If  $c_{\text{min}}<1$ and  $\alpha c_{\text{min}}+g(\alpha)<1$,
\begin{align*}
	\rho^{\mathcal{P}} \geq & \min\left\{\min_{(u=v)\in\{\alpha, 1\}} f(u,v), \min_{u\in\{\alpha, 1\}}\frac{1}{u c_{\text{min}}+g(u)}\right\} > 1.  
	\end{align*}
	(b) If  $c_{\text{min}}<1$ and  $\alpha c_{\text{min}}+g(\alpha)\geq 1$, 
		\begin{align*}
	\rho^{\mathcal{P}} \geq  & \min\left\{\min_{u\in\{\alpha, 1\}} f(u,1), \frac{1}{c_{\text{min}}}\right\} > 1. 
	\end{align*}
	(c) If  $c_{\text{min}}\geq 1$ and  $\alpha c_{\text{min}}+g(\alpha)< 1$,
		\begin{align*}
	\rho^{\mathcal{P}} \geq  & \min\left\{\min_{u\in\{\alpha, 1\}} f(u,\alpha), \frac{1}{\alpha c_{\text{min}}+g(\alpha)}\right\} > 1. 
	\end{align*}
\end{theorem}
\color{black}

From Theorems \ref{thm:RR_adv} and \ref{thm:anyonline}, we conclude that if  $\alpha c_{\text{min}}+g(\alpha)\geq 1$ and $c_{\text{min}}\geq 1$, $\alpha$-RR performs as well as the offline optimal policy and if $c_{\text{min}} < 1$ and/or $\alpha c_{\text{min}}+g(\alpha) < 1$, no deterministic online policy can match the performance of the offline optimal policy.  That is, if $c_{\text{min}} < 1$ and/or $\alpha c_{\text{min}}+g(\alpha) < 1$ no deterministic online policy can have a competitive ratio equal to one whereas in this case the competitive ratio of $\alpha$-RR is is uniformly bounded by six independent of various system parameters.

%% file: analytical_RR_Stochastic
In this section, we characterize the performance of $\alpha$-RetroRenting ($\alpha$-RR) for stochastic request arrivals and rent cost process. Our analytical results hold under the following assumptions of the request arrival and rent cost processes.

\begin{assumption}[Request arrivals and rent cost]
\label{assum:arrivals_and_rent}
\begin{enumerate}
    \item[--] The request arrival process $X_t$, $t = 1, 2, \cdots$ is i.i.d. across time with $X_t \sim \text{Ber}(p)$. 
    \item[--] The rent cost process $Z_t$, $t = 1, 2, \cdots$ is negatively associated \cite{wajc2017negative} with mean $c$ where $c\in [c_{\text{min}},c_{\text{max}}]$. Note that the case where $Z_t$ is i.i.d. across time is a special case of negative association.  
\end{enumerate}

\end{assumption}
We first define three functions $f(), q(),$ and $h()$ which are used in our main result.  These functions  contain the terms which are obtained by using Hoeffding’s inequality to bound the probability of certain events. We introduce these functions to represent the final result in a compact form.

The function $f(\lambda,M,p,c,\alpha,g(\alpha))$ is defined as follows:
\begin{align*}
	f(\lambda,M,p,c,\alpha,g(\alpha)) =  & \max\{M+p, M+c\} 
	\Bigg[  \frac{ \lambda \widetilde{M}_f\delta_A^f\exp\left(-2 (\frac{M}{c_{\text{max}}}+1) \frac{(p(1-g(\alpha))-\alpha c)^2}{(1+\alpha c_{\text{max}}-\alpha c_{\text{min}})^2}\right)}{1-\exp\left(-2 \frac{(p(1-g(\alpha))-\alpha c)^2}{(1+\alpha c_{\text{max}}-\alpha c_{\text{min}})^2}\right)}+\\
	&  \frac{ \lambda \widetilde{M}_f\delta_B^f\exp\left(-2 (\frac{(1-\alpha) M}{1-(1-\alpha) c_{\text{min}}}+1) \frac{((1-\alpha) c-pg(\alpha))^2}{(1+(1-\alpha) (c_{\text{max}}- c_{\text{min}}))^2}\right)}{1-\exp\left(-2 \frac{((1-\alpha) c-pg(\alpha))^2}{(1+(1-\alpha) (c_{\text{max}}- c_{\text{min}}))^2}\right)}+\\
	&\exp\left(  \frac{-2 (\lambda-1)^2 M^2(1-\alpha)^2}{\lambda \widetilde{M}(1+(1-\alpha) (c_{\text{max}}-c_{\text{min}}))^2}\right)+\exp\left(  \frac{-2 (\lambda-1)^2 M^2\alpha^2}{\lambda \widetilde{M}(1+\alpha (c_{\text{max}}-c_{\text{min}}))^2}\right)\Bigg],
	\end{align*}
where $$\widetilde{M}_f=\max\Bigg\{\left\lceil\frac{ M\alpha}{p(1-g(\alpha))-\alpha c}\right\rceil, \left\lceil\frac{ M(1-\alpha)}{(1-\alpha) c-pg(\alpha)}\right\rceil\Bigg\},$$ $$\delta_A^f=\exp\left(\frac{-4  (p(1-g(\alpha))-\alpha c)) \alpha M}{(1+\alpha c_{\text{max}}-\alpha c_{\text{min}})^2}\right),$$ and $$\delta_B^f=\exp\left(\frac{-4  ((1-\alpha) c-pg(\alpha))) (1-\alpha) M}{(1+(1-\alpha) (c_{\text{max}}- c_{\text{min}}))^2}\right).$$

The function $q(\lambda,M,p,c,\alpha,g(\alpha))$ is defined as follows:
\begin{align*}
   q(\lambda,M,p,c,\alpha,g(\alpha)) = &  \max\{\alpha M+\alpha c+g(\alpha)p, M+c\}   \Bigg[\frac{\delta_A^q \lambda \widetilde{M}_q\exp\left(-2 (\frac{M}{c_{\text{max}}}+1) \frac{(p- c)^2}{(1+ c_{\text{max}}-\alpha c_{\text{min}})^2}\right)}{1-\exp\left(-2 \frac{(p- c)^2}{(1+ c_{\text{max}}- c_{\text{min}})^2}\right)}+\\
   &\frac{\delta_B^q \lambda \widetilde{M}_q\exp\left(-2 (\frac{M}{c_{\text{max}}}+1) \frac{(pg(\alpha)-(1-\alpha) c)^2}{(1+(1-\alpha) (c_{\text{max}}- c_{\text{min}}))^2}\right)}{1-\exp\left(-2 \frac{(pg(\alpha)-(1-\alpha) c)^2}{(1+(1-\alpha) (c_{\text{max}}- c_{\text{min}}))^2}\right)}+\\
   &\exp\left(  \frac{-2 (\lambda-1)^2 M^2(1-\alpha)^2}{\lambda \widetilde{M}_q(1+(1-\alpha) (c_{\text{max}}-c_{\text{min}}))^2}\right)+\exp\left(\frac{-2 (\lambda-1)^2 M^2\alpha^2}{\lambda \widetilde{M}_q(1+\alpha (c_{\text{max}}-c_{\text{min}}))^2}\right)\Bigg],
\end{align*}
where $$\widetilde{M}_q=\max\Bigg\{\frac{M}{p- c},  \left\lceil\frac{ M(1-\alpha)}{pg(\alpha)-(1-\alpha) c}\right\rceil\Bigg\},$$
$$\delta_A^q=\exp\left(\frac{-4  (p- c) \alpha M}{(1+ c_{\text{max}}-c_{\text{min}})^2}\right),$$
$$\delta_B^q=\exp\left(\frac{-4  (pg(\alpha)-(1-\alpha) c)) (1-\alpha) M}{(1+(1-\alpha) (c_{\text{max}}- c_{\text{min}}))^2}\right),$$

The function $h(\lambda,M,p,c,\alpha,g(\alpha))$ is defined as follows:
\begin{align*}
 h(\lambda,M,p,c,\alpha,g(\alpha))=&\max\{\alpha M+\alpha c+g(\alpha)p, M+p\} \Bigg[  \frac{2\lambda \widetilde{M}_h\delta_A^h\exp\left(-2 (\frac{M}{1-c_{\text{min}}}+1) \frac{(c-p)^2}{(1+ c_{\text{max}}-\alpha c_{\text{min}})^2}\right)}{1-\exp\left(-2 \frac{(c-p)^2}{(1+ c_{\text{max}}- c_{\text{min}})^2}\right)}+\\
 &\frac{2\lambda \widetilde{M}_h\delta_B^h\exp\left(-2 (\frac{\alpha M}{1-g(\alpha)-\alpha c_{\text{min}}}+1) \frac{((\alpha c-p(1-g(\alpha)))^2}{(1+\alpha (c_{\text{max}}- c_{\text{min}}))^2}\right)}
  {1-\exp\left(-2 \frac{(\alpha c-p(1-g(\alpha)))^2}{(1+\alpha (c_{\text{max}}- c_{\text{min}}))^2}\right)}+\\
  &\exp\left(  \frac{-2 (\lambda-1)^2 M^2\alpha^2}{\lambda \widetilde{M}_h(1+\alpha (c_{\text{max}}-c_{\text{min}}))^2}\right)+\exp\left(  \frac{-2 (\lambda-1)^2 M^2}{\lambda \widetilde{M}_h(1+ (c_{\text{max}}-c_{\text{min}}))^2}\right)\Bigg],
\end{align*}
where $$\widetilde{M}_h=\max\Bigg\{\frac{M}{c-p},  \left\lceil\frac{ M\alpha}{\alpha c-p(1-g(\alpha))}\right\rceil\Bigg\},$$
$$\delta_A^h=\exp\left(\frac{-4  (c-p) \alpha M}{(1+ c_{\text{max}}-c_{\text{min}})^2}\right),$$
$$\delta_B^h=\exp\left(\frac{-4  (\alpha c-p(1-g(\alpha)) \alpha M}{(1+\alpha (c_{\text{max}}- c_{\text{min}}))^2}\right).$$


\begin{remark}
Note that the functions $f(), q(),$ and $h()$ contain terms which are products of linearly increasing functions of $M$ and exponentially decaying functions of $M$. It follows that as $M$ increases, $f(), q(),$ and $h()$ tend to zero.
\end{remark}



Our next theorem characterizes the performance of $\alpha$-RR in terms of our performance metric defined in \eqref{eq:efficiencyRatio}.
\begin{theorem}
	\label{thm:RR_stochastic_theorem}
	Under Assumption \ref{assum:arrivals_and_rent}  and under Model \ref{model:random}, 
	\begin{itemize}
		\item[--] Case $\frac{\alpha c}{1-g(\alpha)}<p<\frac{(1-\alpha)c}{g(\alpha)}$: 
		\begin{align*}
		\sigma^{\alpha-\text{RR}}(T) \leq 1+\min_{\lambda>1} \frac{f(\lambda,M,p,c,\alpha,g(\alpha))}{\alpha c+g(\alpha)p}  
		\end{align*}
		\item[--] Case $p>\max\{c,\frac{(1-\alpha) c}{g(\alpha)}\}$: 
		\begin{align*}
		\sigma^{\alpha-\text{RR}}(T) \leq 1+\min_{\lambda>1} \frac{q(\lambda,M,p,c,\alpha,g(\alpha))}{M+c}
		\end{align*}
		
		\item[--] Case $p<\min\{c,\frac{\alpha c}{1-g(\alpha)}\}$: 
		\begin{align*}
		\sigma^{\alpha-\text{RR}}(T) \leq 1+\min_{\lambda>1} \frac{h(\lambda,M,p,c,\alpha,g(\alpha))}{c}
		\end{align*}

	\end{itemize}
\end{theorem}

\begin{remark}\label{remark:theorem1}
We thus conclude that the the cost incurred by $\alpha$-RR approaches the cost incurred by the online optimal policy $\alpha$-OPT-ON as $M$ increases. 
\end{remark}


%% file: rough2.tex
\label{sec:simulation}
In this section, we present our simulation results  for Model \ref{model:partial}.  In addition to plotting the performance of $\alpha$-OPT, $\alpha$-RetroRenting, and the lower bound on the performance of any  online policy ($\alpha$-LB), we also simulate two policies which do not use partial hosting. We refer to the the offline optimal policy without partial hosting as OPT. The other policy called RetroRenting (RR) was proposed in \cite{narayana2021renting}. RR works on the same principle as that of  $\alpha$-RetroRenting without using partial hosting. We also plot the lower bound on the performance of any online policy which is not allowed to partially host the service (LB). The parameters used for each data point in a plot are given in the figure caption.  Although our analysis for Model 1 holds under Assumption \ref{assum_one_request} and for adversarial arrival processes, in this section, we consider more general arrival processes with more  than one request per time-slot. We also consider the setting where request arrivals are stochastic. 
\subsection{Synthetic Request Arrivals and Rent Costs}

The first set of results use synthetic request arrival and rent cost sequences. Unless stated otherwise, the request arrival process is i.i.d. Bernoulli with parameter $p$. We model the time-varying rent cost sequence using the Autoregressive moving-average (ARMA) model \cite{box2011time}, specifically, ARMA(4,2). To choose the parameters of the model, we fit the model to real world price data obtained from \cite{awsprices}, which provides region-wise prices of unused EC2 capacity in the Amazon Web Services (AWS) cloud.


\begin{figure}
\centering
\parbox{7.3cm}{
\includegraphics[width=7.3cm]{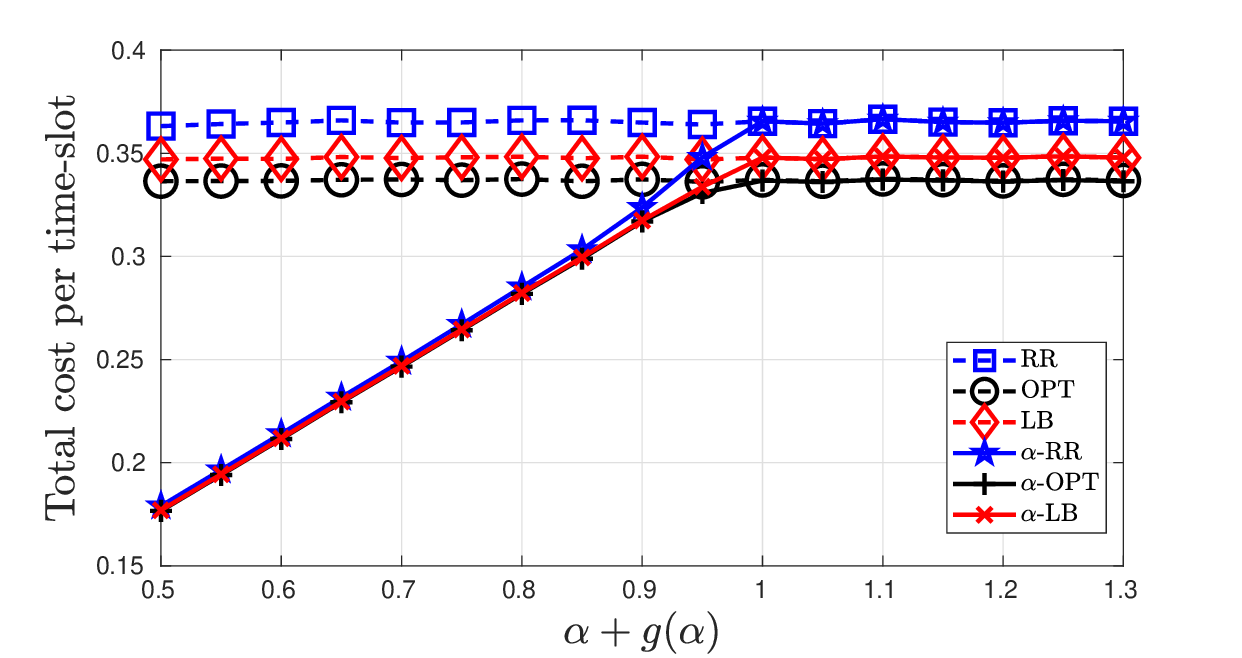}
\caption{Total cost per time slot as a function of $\alpha+g(\alpha)$ for $M = 10$, $c = 0.35$, $p=0.35$, and $\alpha=0.4$}
\label{fig:1}}
\quad
\begin{minipage}{7.3cm}
\includegraphics[width=7.3cm]{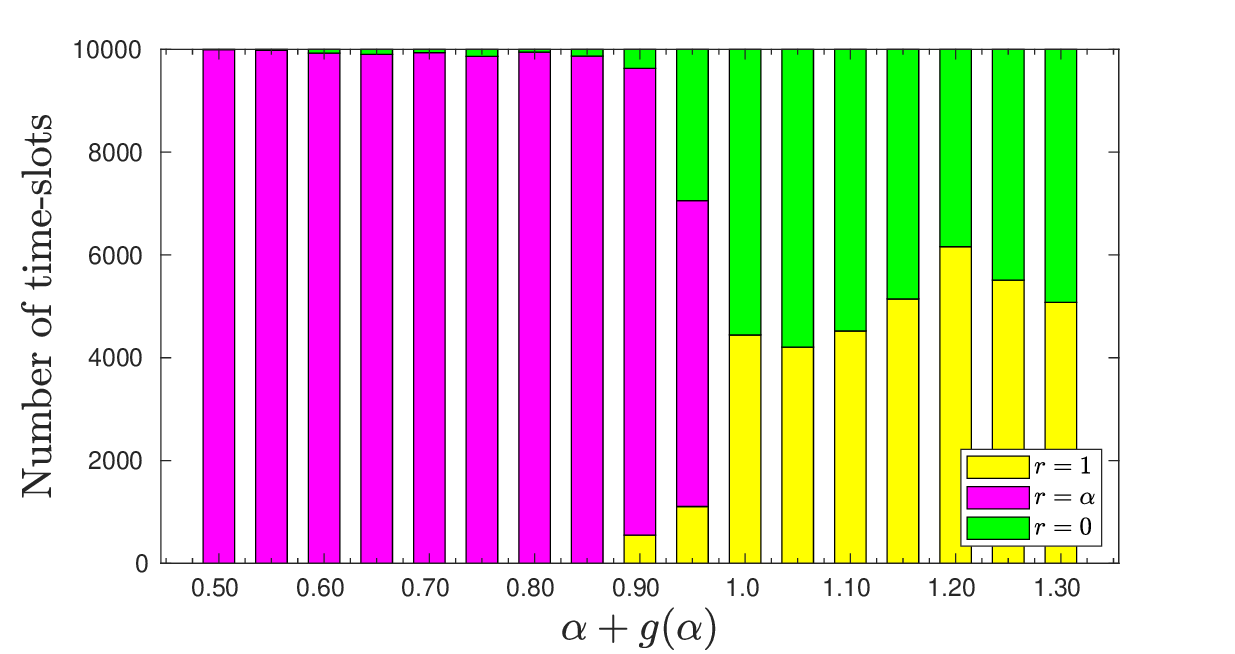}
\caption{Number of time-slots per hosting state under $\alpha$-RR as function of $\alpha+g(\alpha)$ for $M = 10$, $c = 0.35$, $p=0.35$, and $\alpha=0.4$}
\label{fig:middle0}
\end{minipage}
\end{figure}

	




\begin{figure}
\centering
\parbox{7.3cm}{
\includegraphics[width=7.3cm]{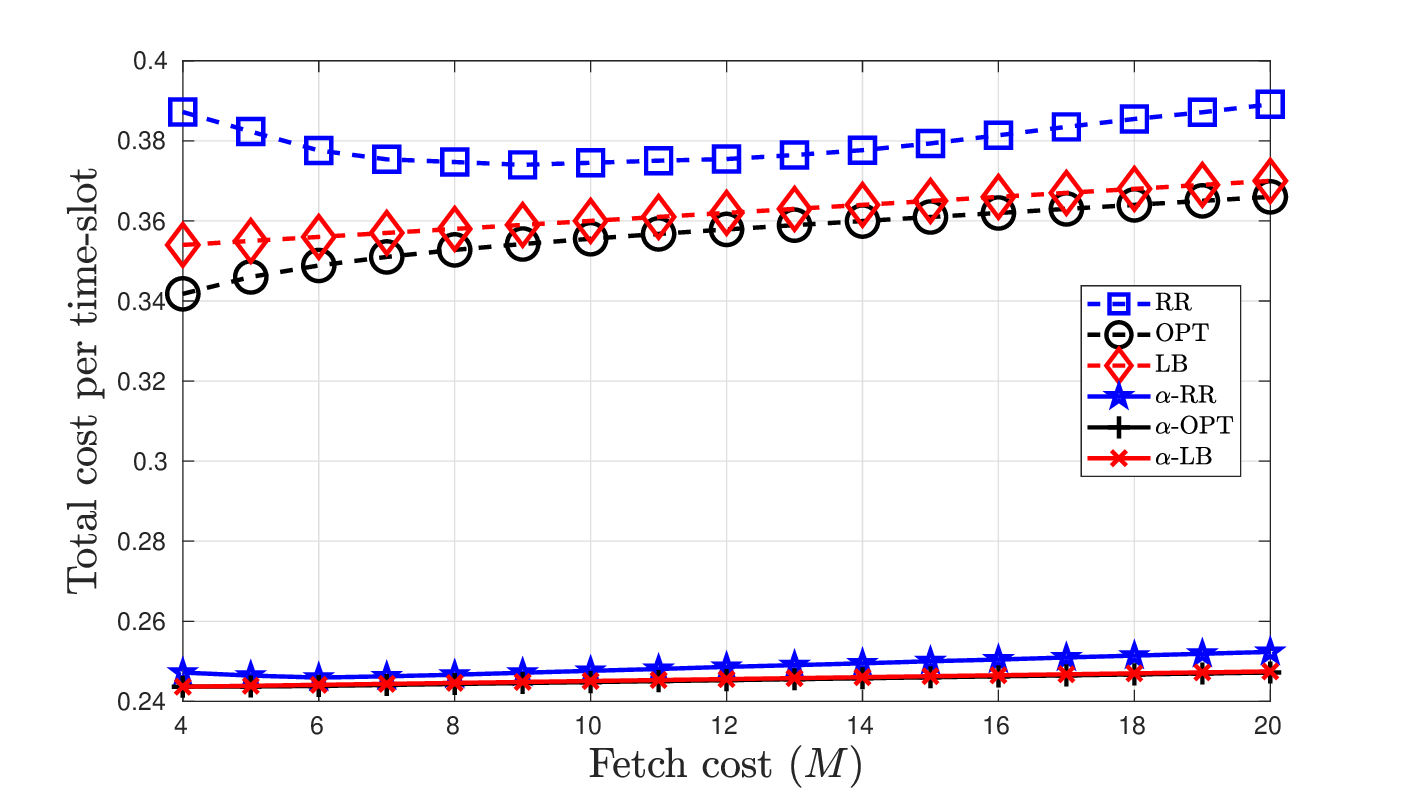}
\caption{Total cost per time slot as a function of fetch cost ($M$) for the case when $\alpha+g(\alpha)<1$. Here $c=0.35$, $\alpha=0.239$, $g(\alpha)=0.380$ and $p=0.42$}
\label{fig:4}}
\quad
\begin{minipage}{7.3cm}
\includegraphics[width=7.3cm]{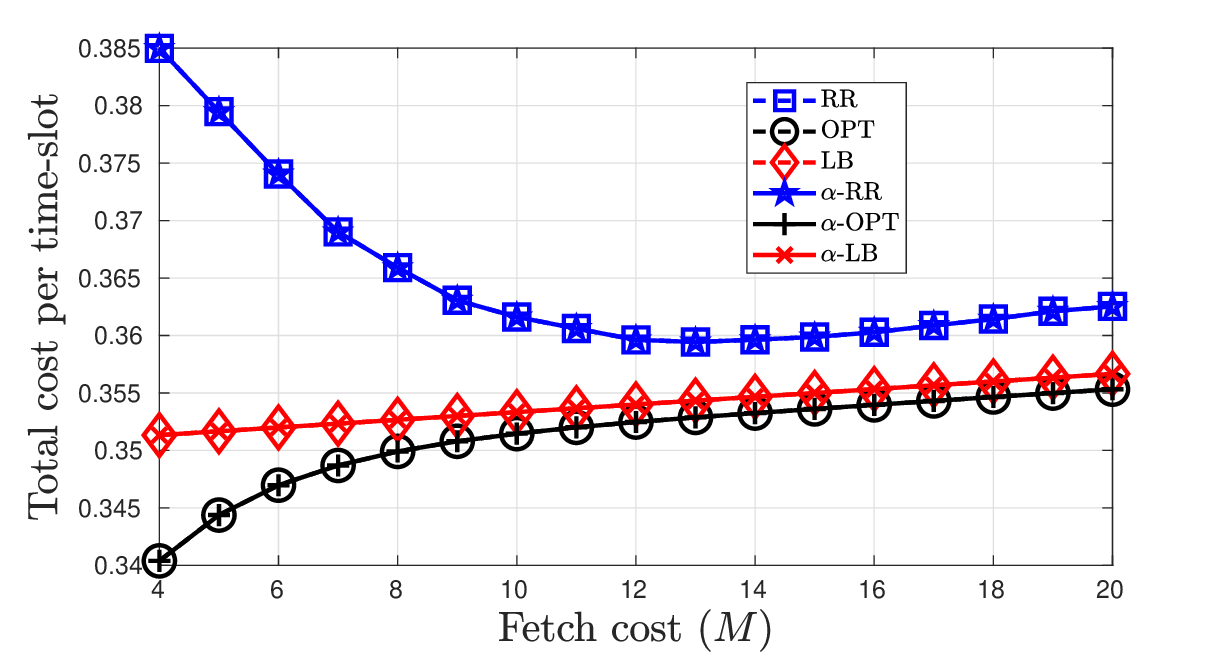}
\caption{Total cost per time slot as a function of fetch cost ($M$) for the case when $\alpha+g(\alpha) \geq 1$. Here $c=0.35$, $\alpha=0.5$, $g(\alpha)=0.7$ and $p =0.42 $}
\label{fig:7}
\end{minipage}
\end{figure}



\begin{figure}
\centering
\parbox{7.3cm}{
\includegraphics[width=7.3cm]{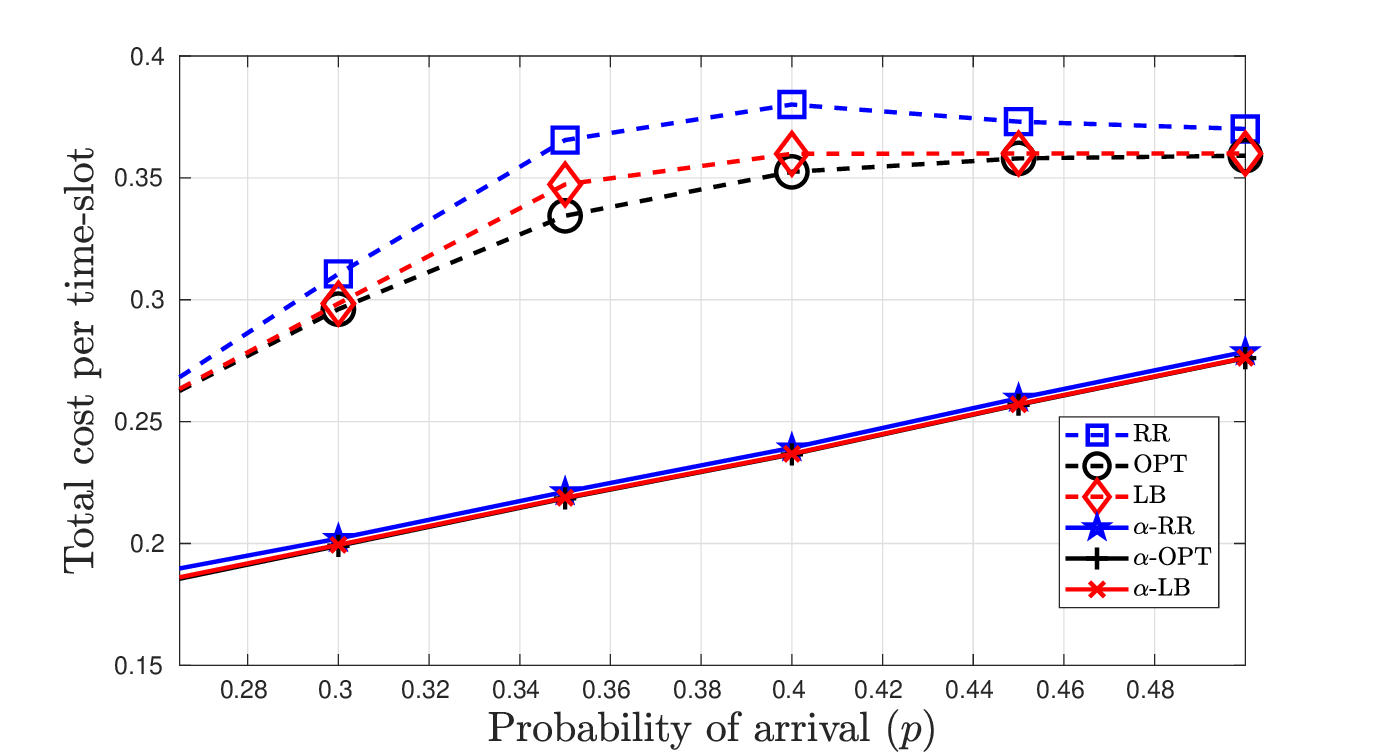}
\caption{Total cost per time slot as a function of request arrival probability ($p$) for the case when $\alpha +g(\alpha)<1$. Here $c=0.35$, $M = 10$, $\alpha=0.239$, $g(\alpha)=0.38$}
\label{fig:16}}
\quad
\begin{minipage}{7.3cm}
\includegraphics[width=7.3cm]{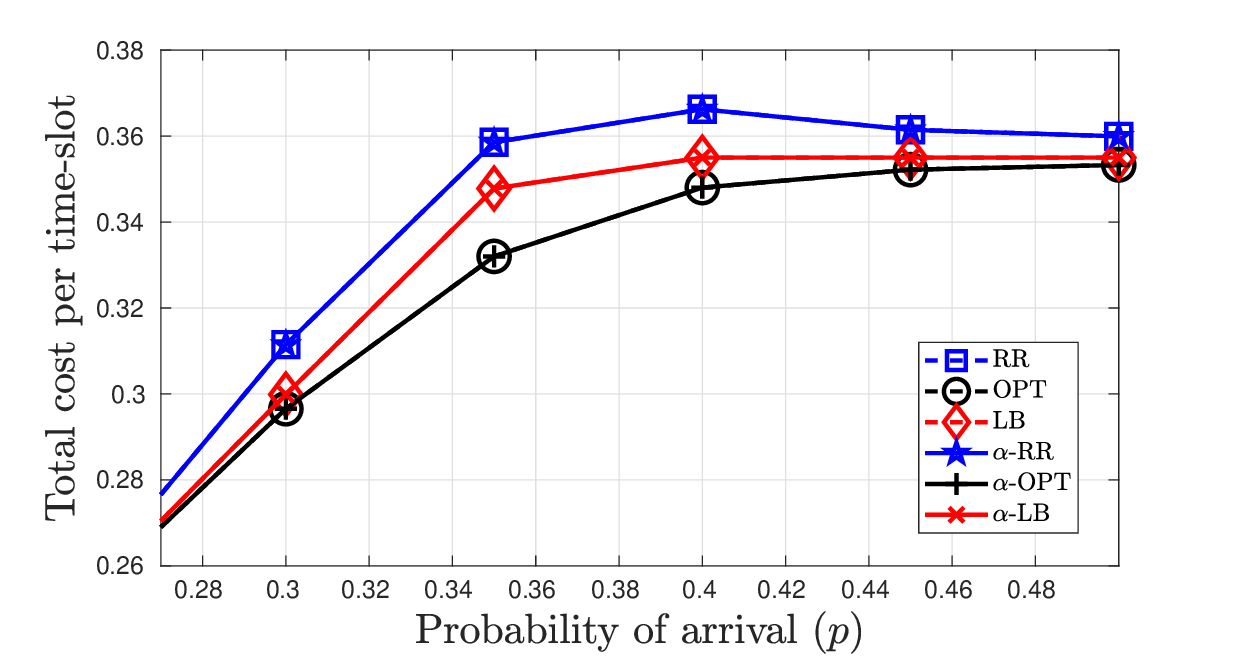}
\caption{Total cost per time slot as a function of request arrival probability ($p$) for the case when $\alpha +g(\alpha)\geq1$. Here $c=0.35$, $M = 10$, $\alpha=0.5$, $g(\alpha)=0.7$}
\label{fig:17}
\end{minipage}
\end{figure}

 

 \begin{figure}
\centering
\parbox{7.3cm}{
\includegraphics[width=7.3cm]{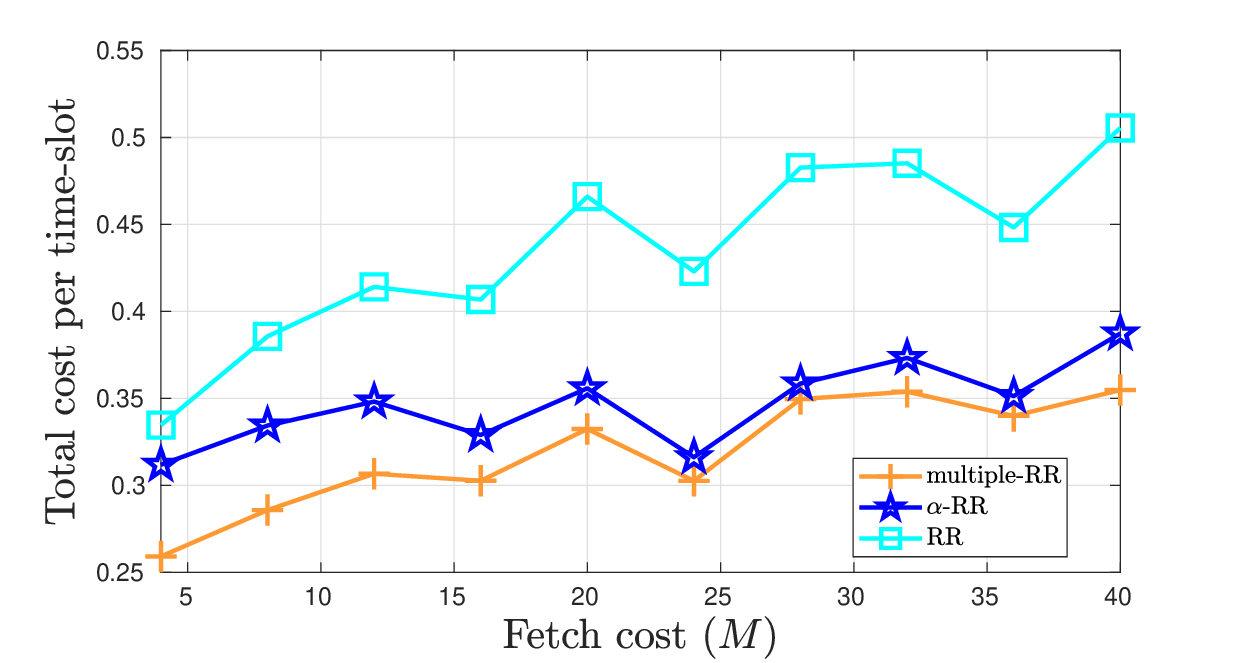}
\caption{Total cost per time slot as a function of fetch cost ($M$) for the case when $\alpha=0.3$, $g(\alpha)=0.4$, $\alpha_1 = 0.4$, $g(\alpha_1) = 0.3$, $\alpha_2 = 0.5$, $g(\alpha_2) = 0.15$ and $c = 0.5$}
\label{fig:GE1}}
\quad
\begin{minipage}{7.3cm}
\includegraphics[width=7.3cm]{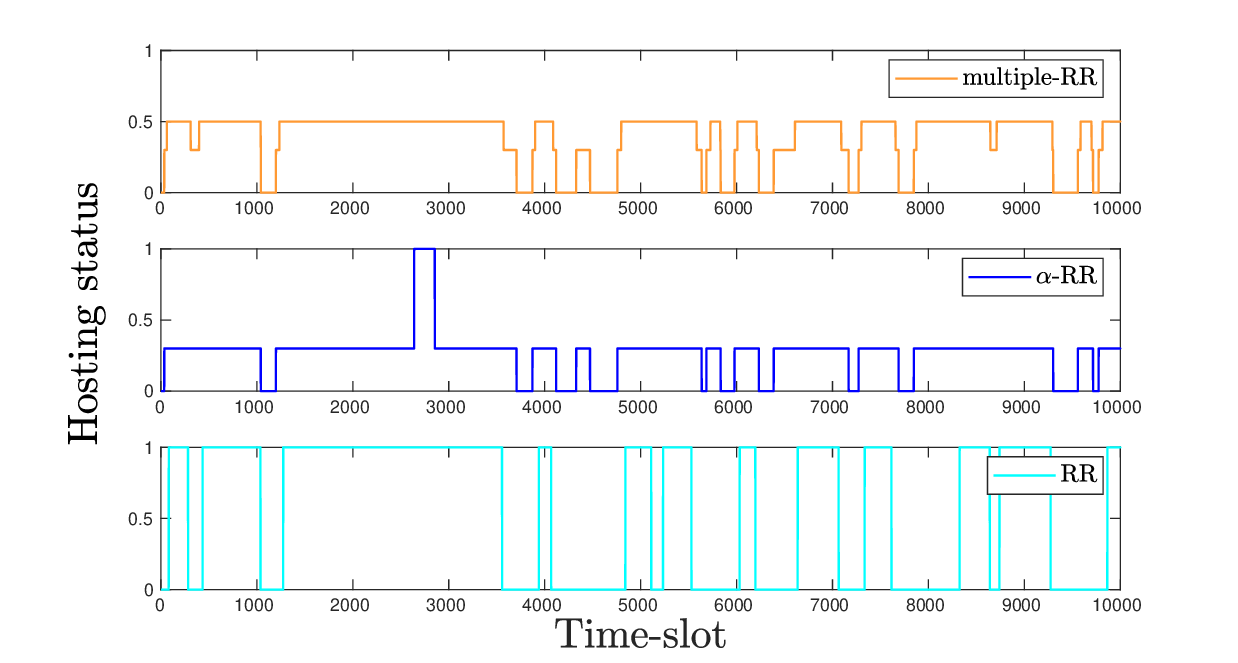}
\caption{Hosting Status over time for different policies for the case when $\alpha=0.3$, $g(\alpha)=0.4$, $\alpha_1 = 0.4$, $g(\alpha_1) = 0.3$, $\alpha_2 = 0.5$, $g(\alpha_2) = 0.15$ and $c = 0.5$}
\label{fig:GE2}
\end{minipage}
\end{figure}
 
	


In Figure \ref{fig:1}, we plot the total cost per time-slot incurred by various policies as a function of $\alpha + g(\alpha)$. We see that the performance gap between policies that are allowed to use partial hosting ($\alpha$-RetroRenting/$\alpha$-OPT/$\alpha$-LB) and their counterparts that are not allowed to use partial hosting (RetroRenting/OPT/LB) is significant for $\alpha + g(\alpha) < 1$ and vanishes for $\alpha + g(\alpha)\geq 1$. 

In Figure \ref{fig:middle0}, we illustrate the number of time-slots for which the $\alpha$-RetroRenting uses the three possible hosting levels as a function of $\alpha + g(\alpha)$ over a time-horizon of 10000 time-slots. As expected, $\alpha$-RetroRenting does not use partial hosting for $\alpha + g(\alpha) \geq 1$.

In Figures \ref{fig:4}, \ref{fig:7}, \ref{fig:16}, and \ref{fig:17}, we plot the total cost per time-slot incurred by various policies as a function of the fetch cost $M$ and request arrival probability $p$, for two cases, namely, $\alpha + g(\alpha) < 1$ and $\alpha + g(\alpha) \geq 1$. We observe a performance gap between policies that are allowed to use partial hosting ($\alpha$-RetroRenting/$\alpha$-OPT/$\alpha$-LB) and their counterparts that are not allowed to use partial hosting (RetroRenting/OPT/LB) only if $\alpha + g(\alpha) < 1$.

In the next set of results presented in Figures \ref{fig:GE1} and \ref{fig:GE2}, we compare the performance of RR and $\alpha$-RetroRenting with a third policy call multiple--RR which works on the same principle as $\alpha$-RetroRenting/RetroRenting and is allowed to use two additional intermediate hosting levels ($\alpha$, $\alpha_1$, $\alpha_2$). For this set of results, we use the Gilbert-Elliot model \cite{liu2010indexability} \color{black} for the request arrival process. The parameters of the model are as shown in Figure \ref{fig:gilbert}. The request arrival process is Bernoulli(0.9) when the Markov chain is in state 0, and Bernoulli(0.1) otherwise.  \color{black}

 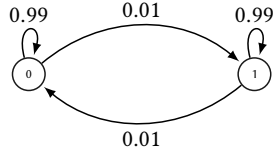
\begin{figure}[h] 
		\begin{center}
			\begin{tikzpicture}[font=\sffamily]
			\usetikzlibrary{calc,arrows.meta,positioning}
			\tikzset{node style/.style={circle,draw=black!90,thick, 
					minimum width=.4cm,
					line width=.2mm,
					fill=white!20}}
			\node[node style] at (6, 1)     (o1)     {\tiny{0}};
			\node[node style] at (9, 1)     (o2)     {\tiny{1}};
			
			\draw[every loop,
			auto=right,
			line width=0.2mm,
			>=latex,
			draw=black,
			fill=black]
		
			(o1)     edge[ bend left=40, auto=left] node {$0.01$} (o2)
			(o1)     edge[loop above] node {$0.99$} (o1)
			(o2)     edge[loop above] node {$0.99$} (o2)
			(o2)     edge[bend left=40, auto=left] node {$0.01$} (o1);

			\end{tikzpicture}
			\caption{The Gilbert-Elliot model}
			\vspace{15pt}
			\label{fig:gilbert}
		\end{center}
\end{figure}

In Figure \ref{fig:GE1}, we see that the additional intermediate hosting options available to multiple--RR brings down the cost incurred. 
In Figure \ref{fig:GE2}, we plot the hosting status $r_t$ for the three different policies as a function of time.

\subsection{Trace-driven Simulations}

\begin{figure}
\centering
\parbox{7.3cm}{
\includegraphics[width=7.3cm]{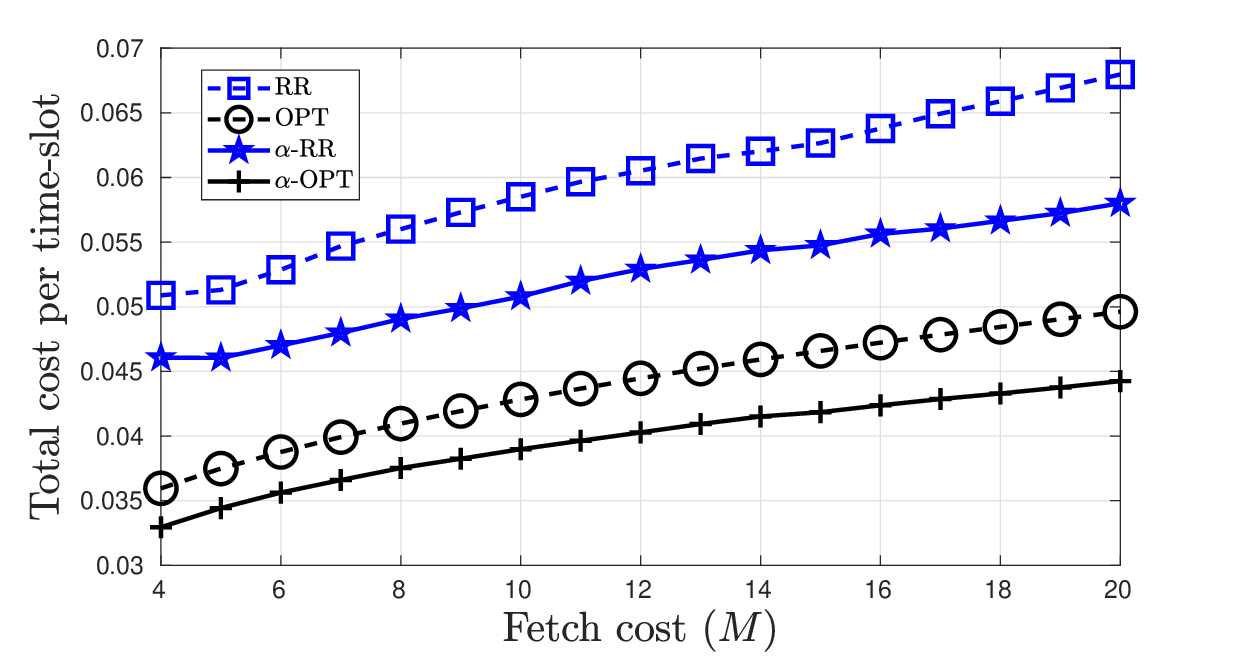}
\caption{Total cost per time slot for as a function of fetch cost ($M$) for the case when $\alpha+g(\alpha)<1$. Here $c=0.135$, $\alpha=0.239$, $g(\alpha)=0.38$}
\label{fig:j2_4}}
\quad
\begin{minipage}{7.3cm}
\includegraphics[width=7.3cm]{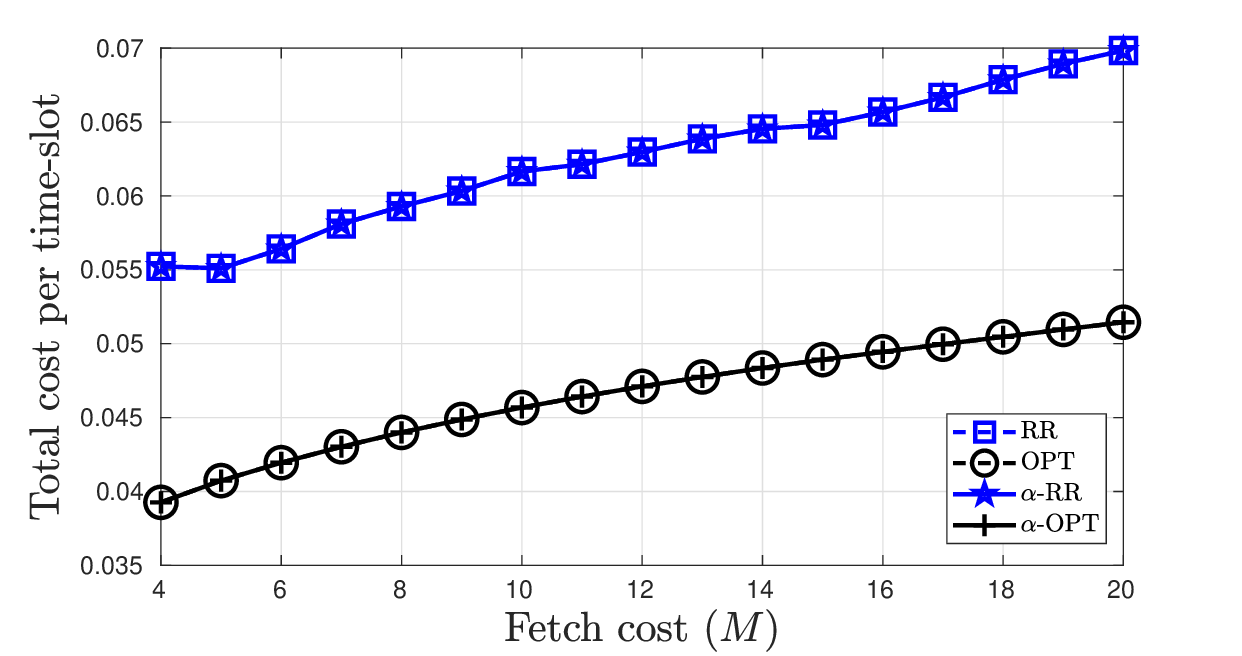}
\caption{Total cost per time slot for as a function of fetch cost ($M$) for the case when $\alpha+g(\alpha) \geq 1$. Here $c=0.135$,  $\alpha=0.5$, $g(\alpha)=0.7$}
\label{fig:j2_7}
\end{minipage}
\end{figure}


For the next set of results, 
For the next set of results, we use trace-data obtained from a Google Cluster \cite{googlecluster} for the arrival process and time-varying spot prices of spare server capacity in the  AWS cloud as given in \cite{awsprices} as the rent cost.


In Figures \ref{fig:j2_4} and \ref{fig:j2_7}, we plot the total cost per time-slot incurred by various policies as a function of the fetch cost $M$ for the case when $\alpha + g(\alpha) < 1$ and $\alpha + g(\alpha) \geq 1$ respectively. As expected, we observe a performance gap between policies that are allowed to use partial hosting ($\alpha$-RetroRenting/$\alpha$-OPT) and their counterparts that are not allowed to use partial hosting (RetroRenting/OPT) only in the first case.

%% file: rough.tex


\label{sec:simulation}
\color{black}
In this section, we present our simulation results for Model $2$. In addition to $\alpha$-RR, we also simulate a policy called RetroRenting (RR) which was proposed in \cite{narayana2021renting} and does not use partial hosting. RR works on the same principles as that of  $\alpha$-RR without using partial hosting.  Although our analysis for Model 2 holds under Assumption \ref{assum_one_request} and stochastic arrivals, in this section, we consider more general arrival
processes with more than one request per time-slot. We also use trace-based arrivals to compare the performance of various policies. \color{black}

\color{black}
\subsection{Synthetic Request Arrivals and Rent Costs}
We consider two types of synthetic request arrivals, namely, Poisson and Markovian. We model the time-varying rent cost sequence for both kinds of arrivals using the Autoregressive moving-average (ARMA) model \cite{box2011time}, specifically ARMA(4,2) with expected rent cost in a slot given as $c$. To choose the parameters of the model, we fit the model using the process described in \cite{Narayana2021OnlinePS} to real world price data obtained from \cite{awsprices}, which provides region-wise prices of unused EC2 capacity in the Amazon Web Services (AWS) cloud.
\subsubsection{Poisson Request Arrivals}
The first set of results are for i.i.d. Poisson request arrivals with parameter $\lambda$ over 10,000 time slots. We use synthetic values for the available partial storage size $\alpha$ and corresponding forwarding cost $g(\alpha)$ for this first set of simulations, as given in the figure captions. In addition to $\alpha$-RR and RR, we also plot the lower bound on the performance of any deterministic online policy ($\alpha$-LB) and a lower bound on the performance of any deterministic online policy which is not allowed to partially host the service (LB). The expressions for these lower bounds can be found in \cite[Lemma 14]{RRtompecs} and Lemma \ref{lemma:optimal_causal} in the appendix respectively. 

\begin{figure}
\centering
\parbox{7.3cm}{
\includegraphics[width=7.3cm]{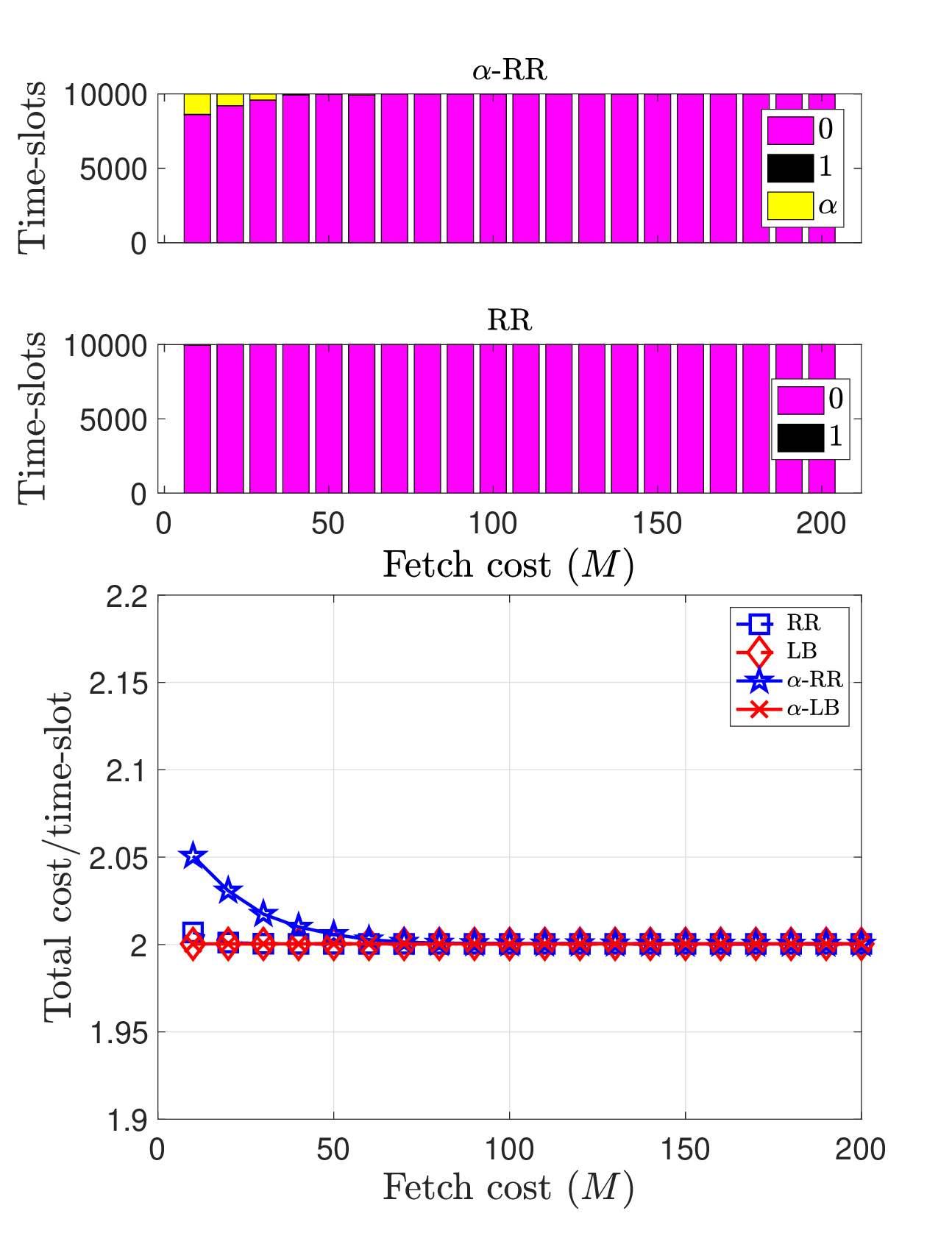}
\caption{Histogram of hosting status \& total cost per time slot as a function of fetch cost ($M$) for the case when $c=4.5$, $\alpha=0.30$, $g(\alpha)=0.50$, and $\lambda=2$}
\label{fig:b1}}
\quad
\begin{minipage}{7.3cm}
\includegraphics[width=7.3cm]{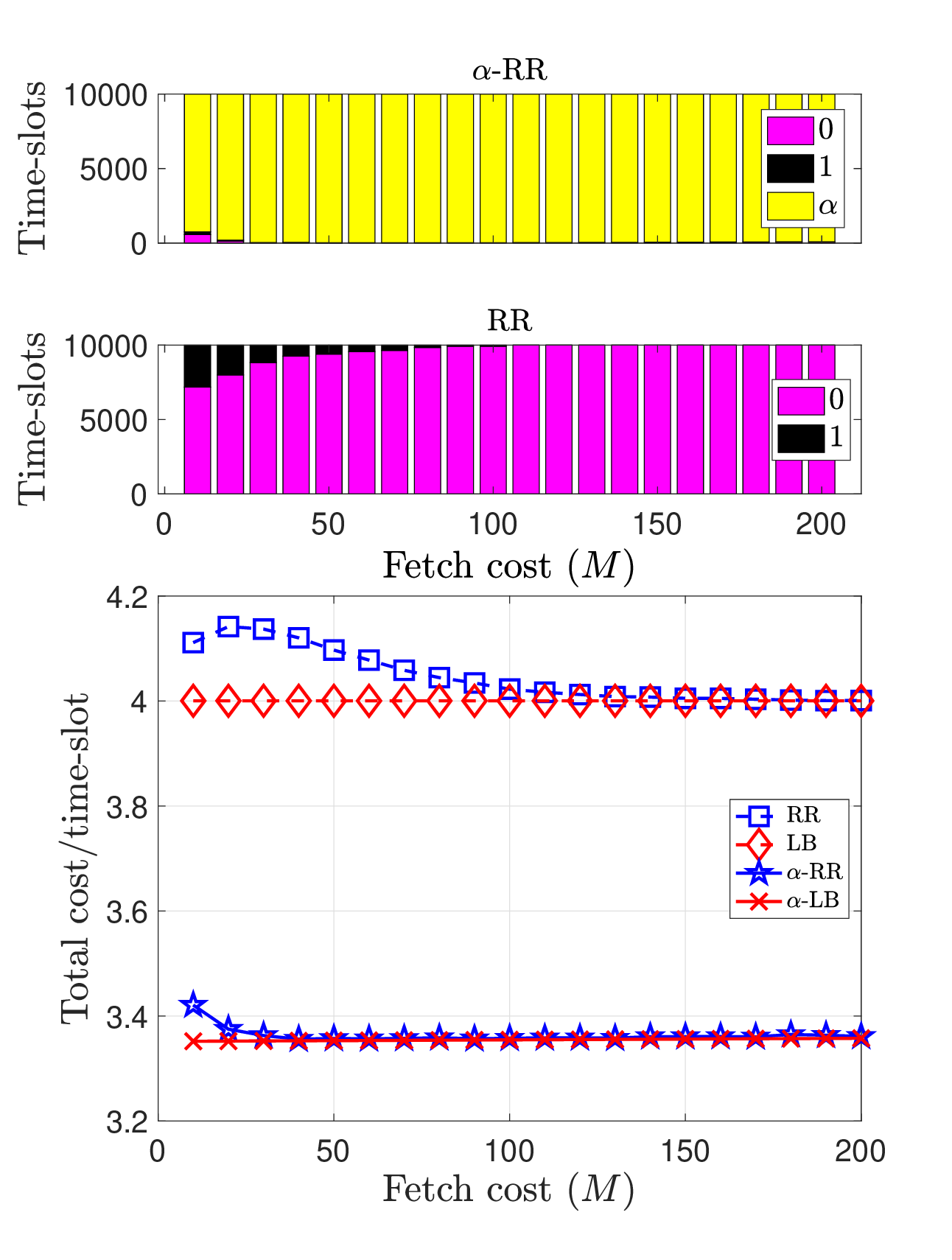}
\caption{Histogram of hosting status \& total cost per time slot as a function of fetch cost ($M$) for the case when $c=4.5$, $\alpha=0.30$, $g(\alpha)=0.50$, and $\lambda=4$}
\label{fig:b3}
\end{minipage}
\end{figure}



\begin{figure}
\centering
\parbox{7.3cm}{
\includegraphics[width=7.3cm]{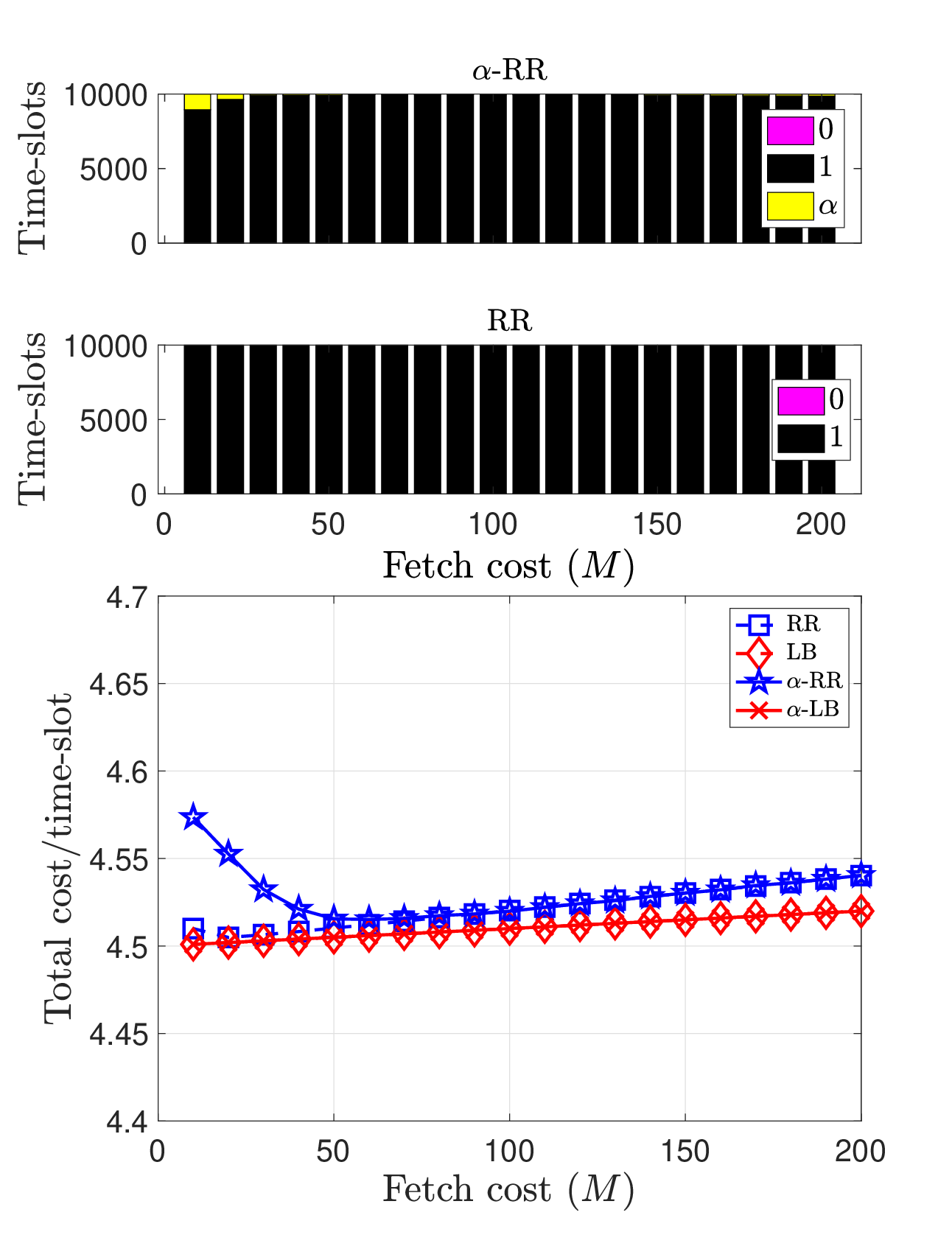}
\caption{Histogram of hosting status \& total cost per time slot as a function of fetch cost ($M$) for the case when $c=4.5$, $\alpha=0.30$, $g(\alpha)=0.50$, and $\lambda=8$ }
\label{fig:b5}}
\quad
\begin{minipage}{7.3cm}
\includegraphics[width=7.3cm]{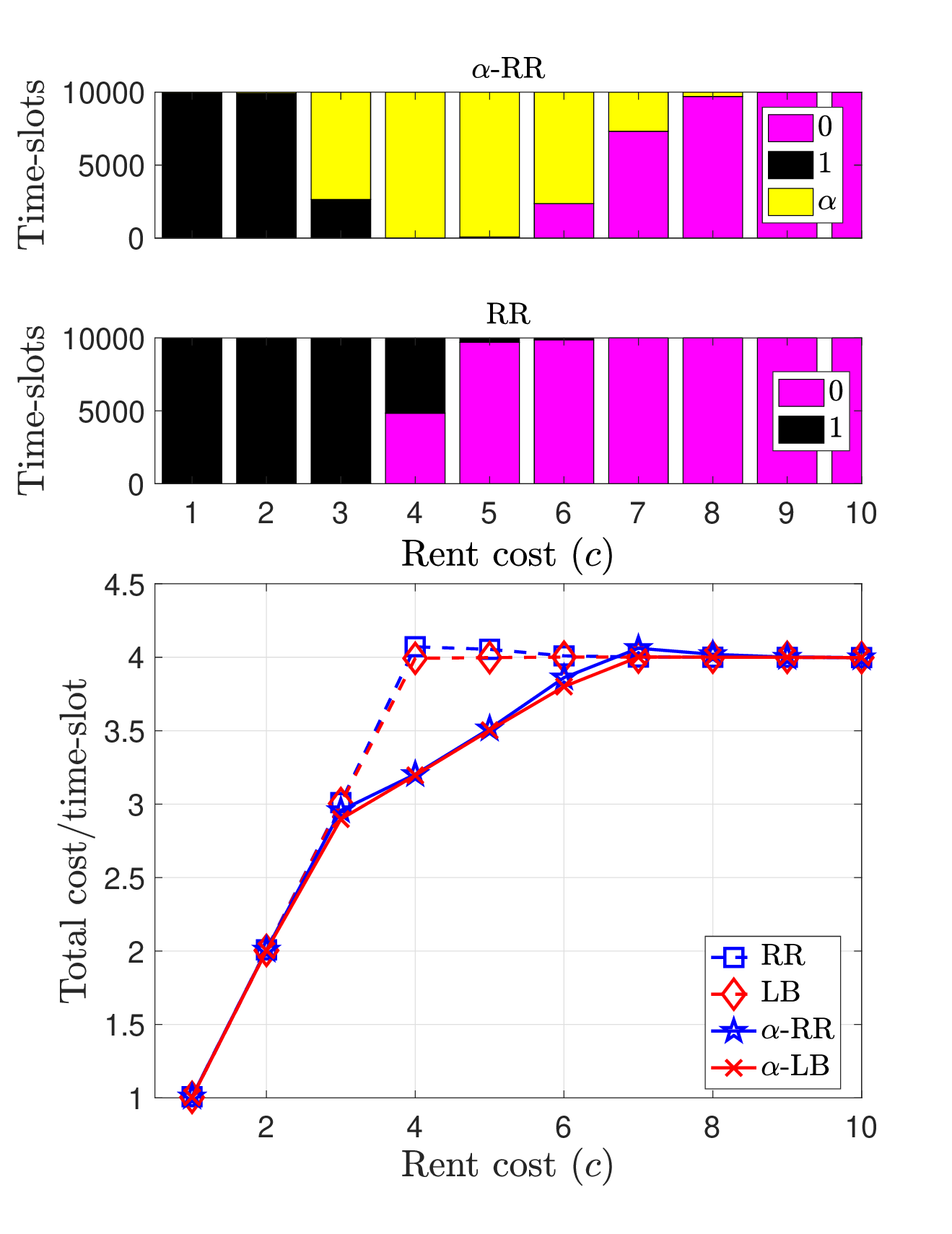}
\caption{Histogram of hosting status \& total cost per time slot as a function of rent cost ($c$) for the case $M=40$, $\alpha=0.30$, $g(\alpha)=0.50$, and $\lambda=4$}
\label{fig:b7}
\end{minipage}
\end{figure}


In Figures \ref{fig:b1}-\ref{fig:b5}, we first plot the histogram of hosting status and then the total cost per time slot under the two policies as a function of the fetch cost ($M$), for different values of request arrival intensity ($\lambda$). 

From the first two plots of Figure \ref{fig:b1}, we observe that RR does not host the service at all, while $\alpha$-RR also does the same except when $M$ is very small when it uses partial storage, albeit very rarely. Both policies lean towards not hosting the service as the average number of arrivals per time-slot $\lambda$ is smaller than the average rent cost $c$ here, and thus it is cost effective to serve requests via the cloud instead of incurring rent cost. Here, the lower bounds of online policies, namely, LB and $\alpha$-LB, coincide with each other. Also, as $M$ increases, the performance of RR and $\alpha$-RR approaches the lower bound.
In Figure \ref{fig:b3}, we note that the average number of requests and average rent cost is comparable, and here $\alpha$-RR hosts $\alpha$ fraction of the service at all times while RR hosts the entire service. We note that $\alpha$-RR outperforms RR for smaller values of $M$, while they have similar costs for larger values of $M$.
Figure \ref{fig:b5} considers a larger value of $\lambda$ where hosting the entire service is optimal, and here both the policies host the entire service at all times for all values of $M$ considered.


In Figure \ref{fig:b7}, we first plot the histogram of hosting status and then the total cost per time slot under the two policies as a function of the rent cost ($c$), for a fixed value of request arrival intensity ($\lambda$). 
From the first two plots in Figure \ref{fig:b7}, we observe that for low values of $c$, $\alpha$-RR and RR host the entire service in all time-slots. For $3\leq c \leq 7$, $\alpha$-RR hosts $\alpha$ fraction of the service for most of the time. For higher values of the rent cost ($c \geq 8$), $\alpha$-RR and RR both do not host the service. We note that $\alpha$-RR outperforms RR for values of $c$ close to average arrival rate ($\lambda$) and the performance of $\alpha$-RR and RR is comparable for low and large $c$. This is also consistent with what we observe in Figures \ref{fig:b1}-\ref{fig:b5}.

\subsubsection{Markovian Request Arrivals}

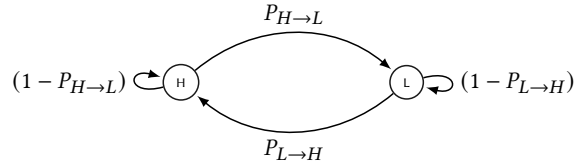
\begin{figure}[h] 
		\begin{center}
			\begin{tikzpicture}[font=\sffamily]
			\usetikzlibrary{calc,arrows.meta,positioning}
			\tikzset{node style/.style={circle,draw=black!90,thick, 
					minimum width=.4cm,
					line width=.2mm,
					fill=white!20}}
			\node[node style] at (6, 1)     (o1)     {\tiny{H}};
			\node[node style] at (9, 1)     (o2)     {\tiny{L}};
			
			\draw[every loop,
			auto=right,
			line width=0.2mm,
			>=latex,
			draw=black,
			fill=black]
		
			(o1)     edge[ bend left=40, auto=left] node {${P}_{H\rightarrow L}$} (o2)
			(o1)     edge[loop left] node {$(1-{P}_{H\rightarrow L})$} (o1)
			(o2)     edge[loop right] node {$(1-{P}_{L\rightarrow H})$} (o2)
			(o2)     edge[bend left=40, auto=left] node {${P}_{L\rightarrow H}$} (o1);

			\end{tikzpicture}
			\caption{The Gilbert-Elliot model}
			\vspace{15pt}
			\label{fig:gilbert}
		\end{center}
\end{figure}

For the next set of results, we use the Gilbert-Elliot model \cite{liu2010indexability} \color{black} for the request arrival process. The parameters of the model are as shown in Figure \ref{fig:gilbert}. The request arrival process is Poisson(200) when the Markov chain is in high state (\textit{H}), and Poisson(10) in the low state (\textit{L}). The values of partial storage $\alpha$ and forwarding cost $g(\alpha)$ used here are chosen based on a curve derived from a real dataset (see Figure \ref{fig:789}); details are presented in Section~\ref{derived galpha}. From Figure \ref{fig:789}, we see when $\alpha=0.16$, forward cost $g(0.16)$ is $0.76$ which is the minimum $\alpha + g(\alpha)$ pair among all pairs.

In the following series of experiments, we vary the transition probabilities for the request arrival process and compare the performance of various policies. We also compare the performance of RR and $\alpha$-RR with two other policies that know the statistics of the request arrivals and the expected rent cost $c$.

The first policy, referred to as the MDP policy, formulates the hosting problem as a Markov Decision Process. The second policy called Arrival Based Caching (ABC) proposed in \cite{ABC} makes hosting decisions based only on the request arrival rate in the current time-slot and the statistics of the arrival process. 

\begin{figure}
\centering
\parbox{7.3cm}{
\vspace{-0.3cm}
\includegraphics[width=7.3cm]{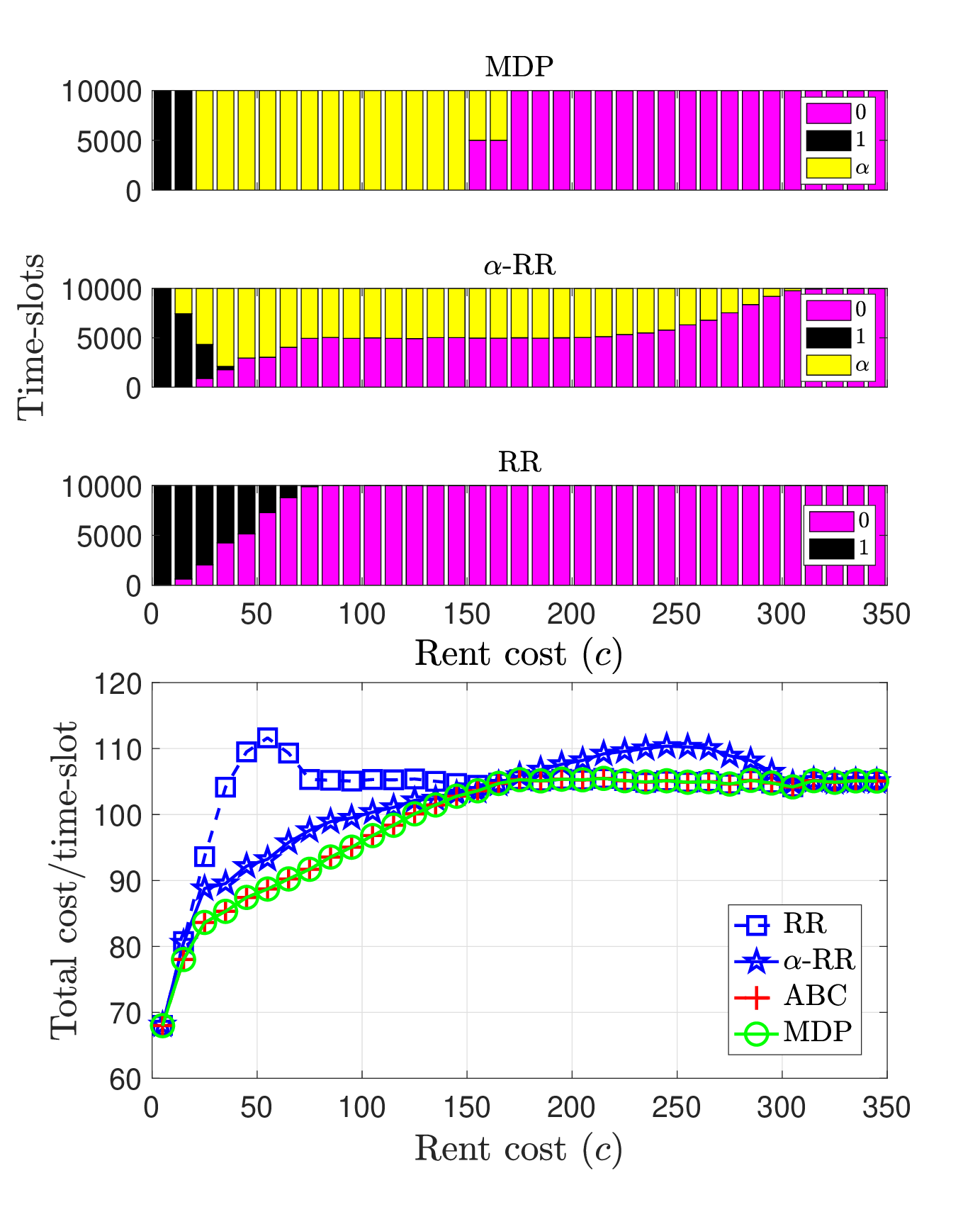}
\caption{Histogram of hosting status \& total cost per time slot as a function of rent cost ($c$) for the case when $\alpha=0.16$, $g(\alpha)=0.76$, $\alpha_1 = 1$, $g(\alpha_1) = 0.6$, $M = 50$ and ${P}_{H\rightarrow L} = 0.4$, ${P}_{L\rightarrow H} = 0.4$}
\label{fig:m1}}
\quad
\begin{minipage}{7.3cm}
\includegraphics[width=7.3cm]{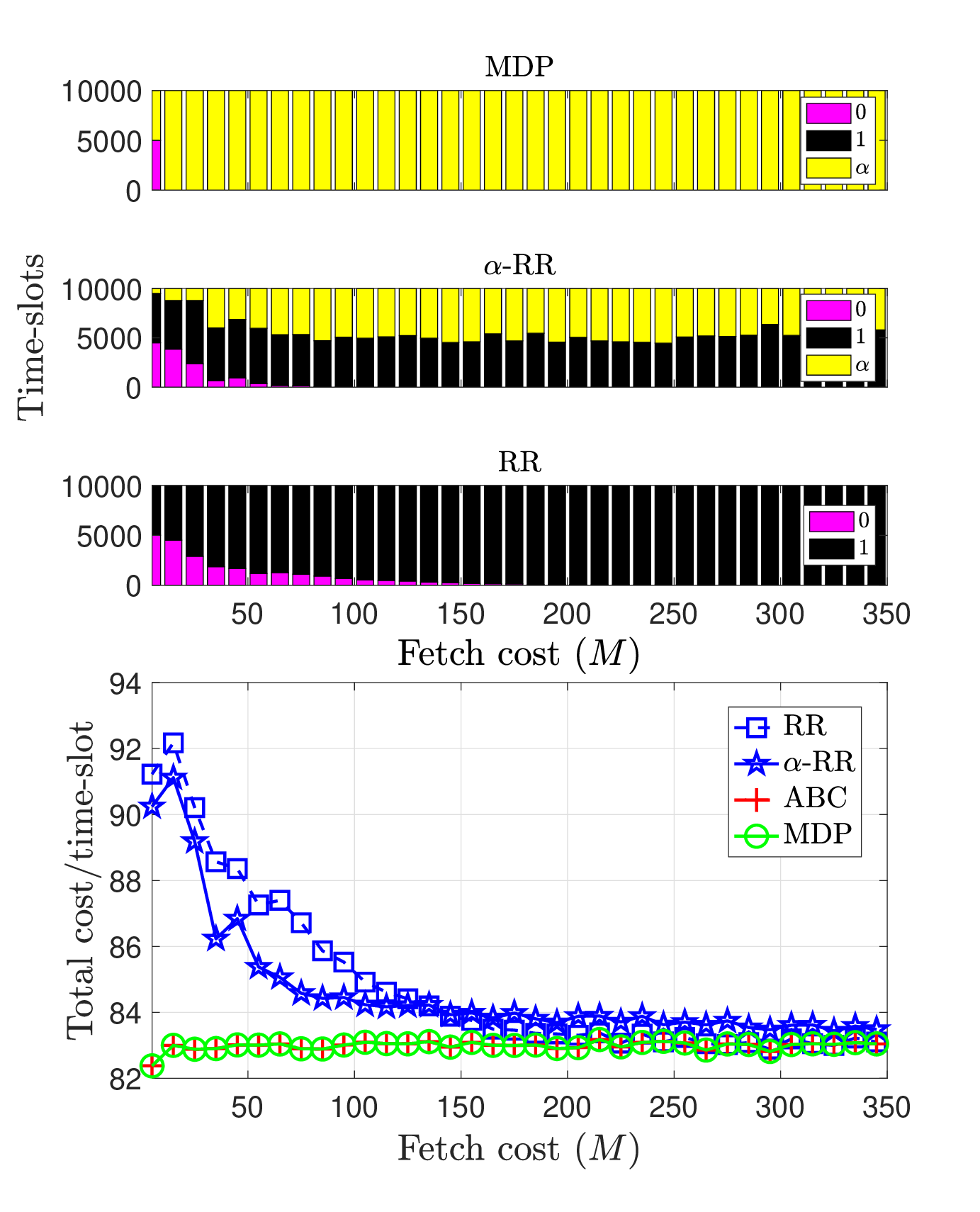}
\caption{Histogram of hosting status \& total cost per time slot as a function of fetch cost ($M$) for the case when $\alpha=0.16$, $g(\alpha)=0.76$, $\alpha_1 = 1$, $g(\alpha_1) = 0.6$, $c = 20$, and ${P}_{H\rightarrow L} = 0.4$, ${P}_{L\rightarrow H} = 0.4$}
\label{fig:m3}
\end{minipage}
\end{figure}

	


In the first three plots in Figure \ref{fig:m1}, we plot the histogram of the hosting status under the three policies as a function of the rent cost $c$. Here, the average number of requests per time slot is $105$ and transition probability ${P}_{H\rightarrow L}$ is same as ${P}_{L\rightarrow H}$ which is $0.4$. We observe that for low values of rent cost ($c \leq 20$), MDP hosts the entire service in all time-slots. For higher values of rent cost $20 < c \leq 150$, MDP hosts $\alpha$ fraction of the service in all time-slots. For $160 \leq c \leq 170$, MDP varies between hosting $\alpha$ fraction of the service and not hosting the service. For even higher values of rent cost ($c > 170$), MDP does not host the service. These trends are explained by the fact that as the expected rent cost $c$ increases, the cost incurred by serving requests via the cloud becomes lower than the cost of renting edge resources to host the service at the edge. The hosting status under $\alpha$-RR has a similar trend; however, the range of values of rent cost over which $\alpha$-RR hosts at least a part of the service is larger than MDP. RR hosts the entire service for low values of rent cost and does not host the service when rent cost $c$ is above a threshold, $c=80$. For intermediate values of rent cost, it varies between hosting the entire service and not hosting the service. As a consequence of the hosting status under the three policies, the total cost incurred under all three policies is very close for very low and very high values of rent cost. For rent costs, $c<180$, $\alpha$-RR outperforms RR and RR outperforms $\alpha$-RR for rent cost $c$ in the range $[180, 310]$.

Similarly, in the first three plots in Figure \ref{fig:m3}, we plot the histogram of hosting status under the three policies as a function of the fetch cost $M$. 
We observe that for most values of fetch cost $M$, MDP hosts $\alpha$ fraction of the service in all time-slots. 
Under $\alpha$-RR the hosting status varies between $\alpha$ fraction of the service and the entire service. We observe that $\alpha$-RR outperforms RR for low values of $M$ and their performance converges as $M$ increases. 

\begin{figure}
\centering
\parbox{7.3cm}{
\vspace{-0.3cm}
\includegraphics[width=7.3cm]{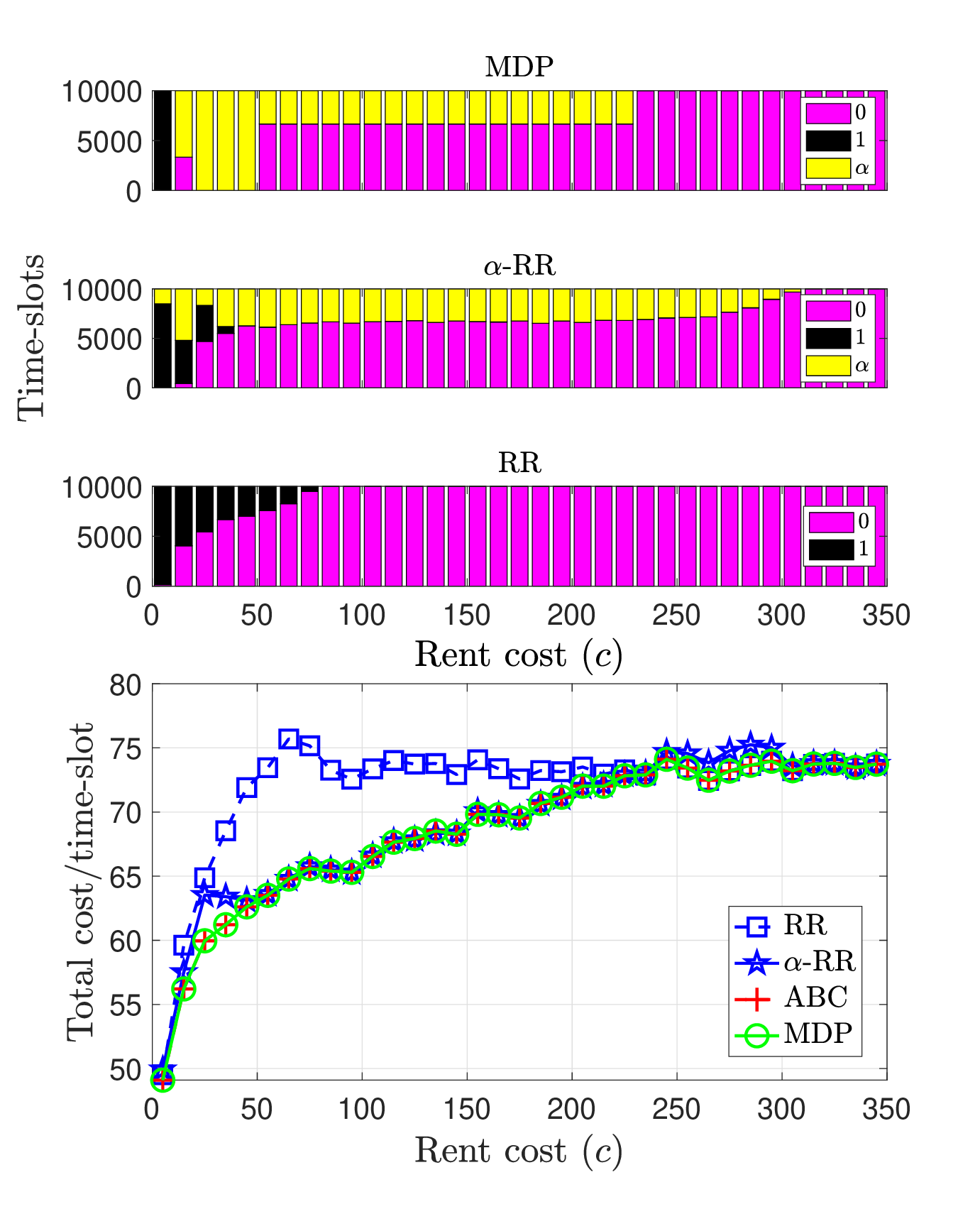}
\caption{Histogram of hosting status \& total cost per time slot as a function of rent cost ($c$) for the case when $\alpha=0.16$, $g(\alpha)=0.76$, $\alpha_1 = 1$, $g(\alpha_1) = 0.6$, $M = 50$ and ${P}_{H\rightarrow L} = 0.2$, ${P}_{L\rightarrow H} = 0.1$}
\label{fig:m5}}
\quad
\begin{minipage}{7.3cm}
\includegraphics[width=7.3cm]{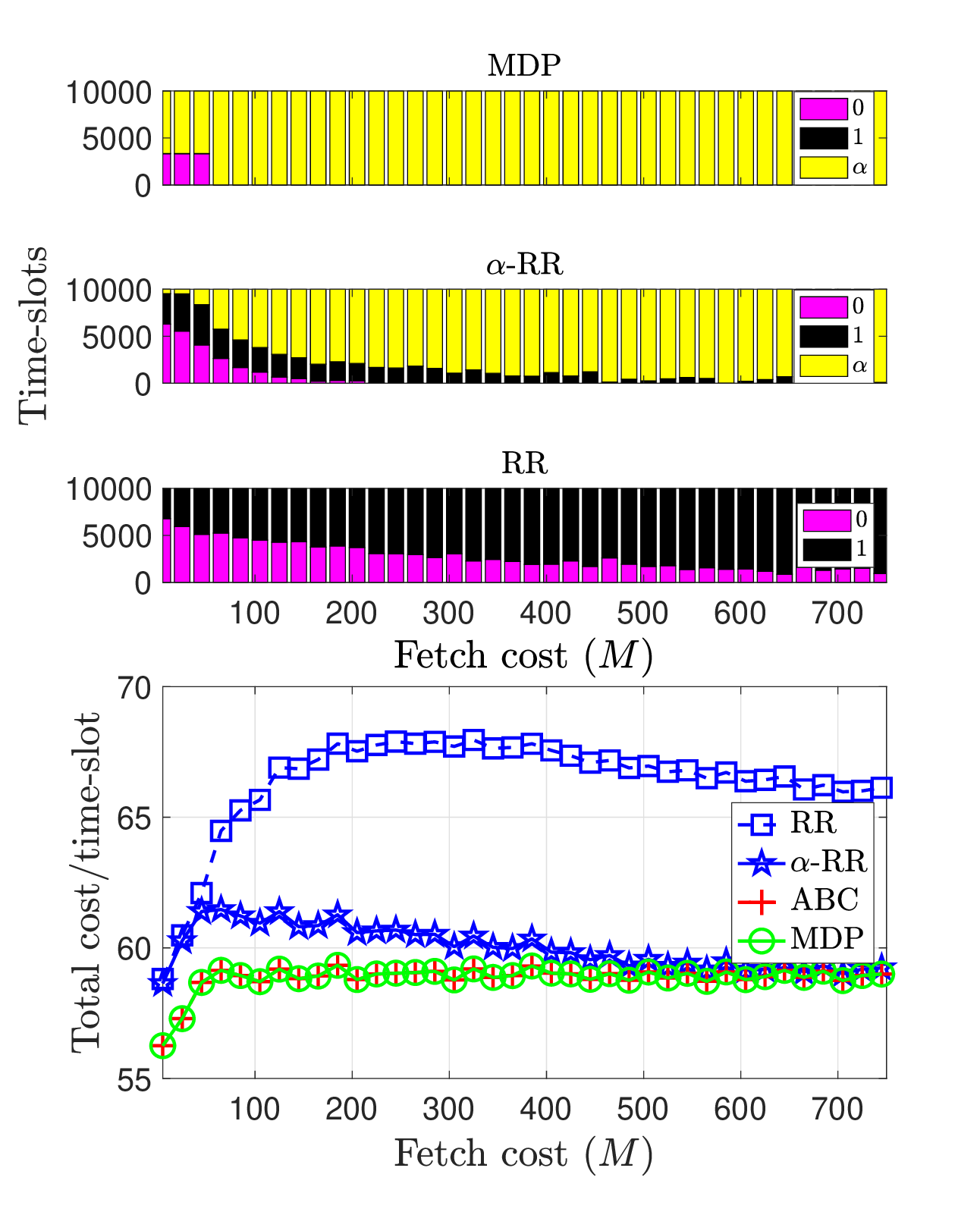}
\caption{Histogram of hosting status \& total cost per time slot as a function of fetch cost ($M$) for the case when $\alpha=0.16$, $g(\alpha)=0.76$, $\alpha_1 = 1$, $g(\alpha_1) = 0.6$, $c = 20$, and ${P}_{H\rightarrow L} = 0.2$, ${P}_{L\rightarrow H} = 0.1$}
\label{fig:m7}
\end{minipage}
\end{figure}



Unlike Figures~\ref{fig:m1} and \ref{fig:m3}, we next consider lower values of transition probabilities for the request arrival process, and again compare the performance of the three policies as a function of the rent cost $c$. For Figures \ref{fig:m5} and \ref{fig:m7}, the average number of requests per time slot is $73.33$ and transition probability ${P}_{H\rightarrow L}$ is $0.2$ whereas ${P}_{L\rightarrow H}$ is $0.1$. From Figure \ref{fig:m5}, we observe that $\alpha$-RR outperforms RR for most values of $c$. For very high values of $c$, the total cost incurred under all policies is close. We thus note that even though that MDP and ABC know the statistics of the arrival process while $\alpha$-RR does not have that information, the performance of $\alpha$-RR is comparable with MDP and ABC.  We note a similar trend in Figure \ref{fig:m7}, with RR performing poorly for all values of $M$ considered.

\begin{figure}
\centering
\parbox{7.3cm}{
\vspace{-0.3cm}
\includegraphics[width=7.3cm]{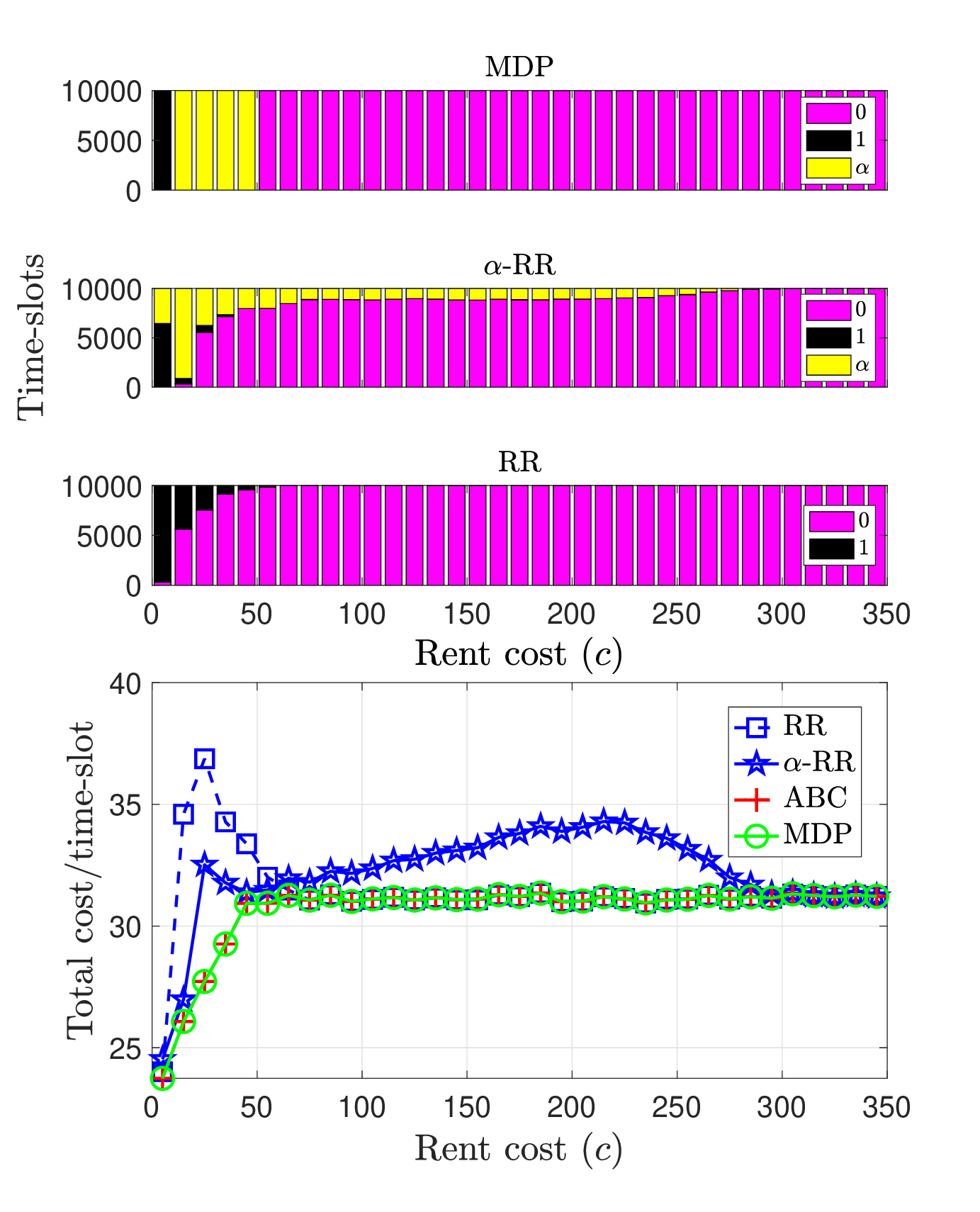}
\caption{Histogram of hosting status \& total cost per time slot as a function of rent cost ($c$) for the case when $\alpha=0.16$, $g(\alpha)=0.76$, $\alpha_1 = 1$, $g(\alpha_1) = 0.6$, $M = 50$ and ${P}_{H\rightarrow L} = 0.8$, ${P}_{L\rightarrow H} = 0.1$}
\label{fig:m9}}
\quad
\begin{minipage}{7.3cm}
\includegraphics[width=7.3cm]{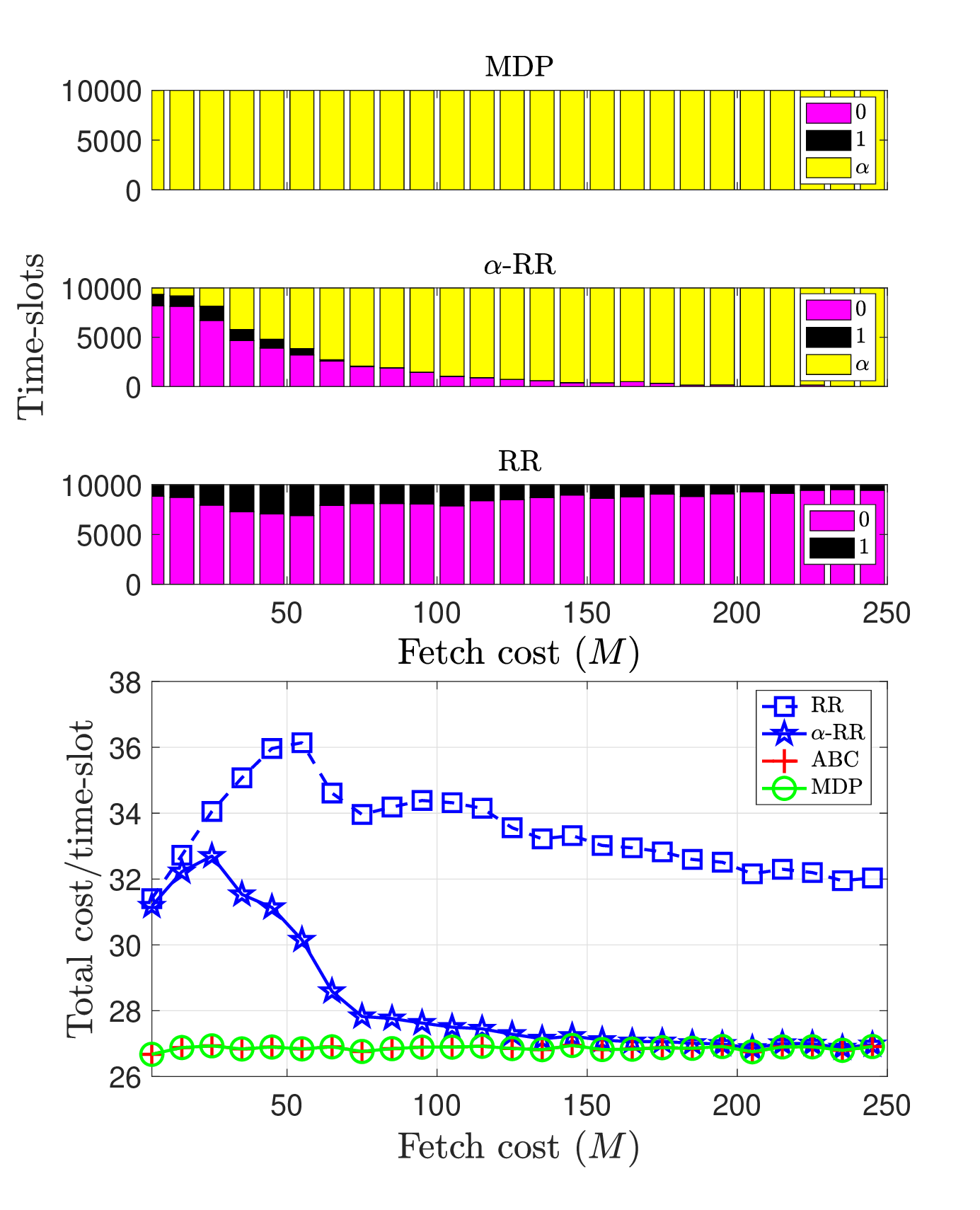}
\caption{Histogram of hosting status \& total cost per time slot as a function of fetch cost ($M$) for the case when $\alpha=0.16$, $g(\alpha)=0.76$, $\alpha_1 = 1$, $g(\alpha_1) = 0.6$, $c = 20$, and ${P}_{H\rightarrow L} = 0.8$, ${P}_{L\rightarrow H} = 0.1$}
\label{fig:m11}
\end{minipage}
\end{figure}

	

	
In the last set of  Figures \ref{fig:m9} and \ref{fig:m11},  we consider a request arrival process with very asymmetric transition probabilities. The average number of requests per time slot is $31.11$ and transition probability ${P}_{H\rightarrow L}$ is $0.8$ whereas ${P}_{L\rightarrow H}$ is $0.1$. Here, the underlying Markov chain spends most of its time in the low state.
Similar to previous experiments, from Figure \ref{fig:m9}, we conclude that for low values of rent cost $c$, $\alpha$-RR outperforms RR. For intermediate values of $c$, RR performs better than $\alpha$-RR  and for high values of $c$, the total cost incurred under all policies is close. For the same request arrival process, we plot the total cost vs $M$ in Figure \ref{fig:m11}, and find that $\alpha$-RR outperforms RR for all values of $M$ considered.


\subsection{Simulations on GPS Trajectory Dataset}\label{derived galpha}
In the next set of simulations, we use a real dataset to construct the forwarding cost function $g(\cdot)$ for a shortest path query system. In particular, we use a GPS trajectory dataset \cite{zheng2008understanding}\cite{zheng2009mining}\cite{zheng2010geolife} collected as a part of the Geolife Project by Microsoft Research Asia. This dataset contains 17,621 cab trajectories from Beijing, China, compiled over five years from 2007-12. Using this data, we model a service which when queried with a source-destination pair, returns the shortest path corresponding to that input. As before the goal is to host this service in the most cost-efficient way. Next, we describe how the edge cache size is measured and our algorithm to decide what goes in the cache.

The complete library at the server is considered to be of size $N \times P$ where $N$ is the total number of queries and $P$ is the total number of nodes in the shortest path corresponding to all $N$ queries. We host a fraction of this dataset in our locally available cache to assist with answering future queries. The cache size is measured by the total number of nodes stored across all the cached trajectories.

We extract each trajectory's start and end locations as the source and destination nodes of the shortest path query, respectively. So each query is a request to the system which returns the shortest distance from source to destination. To find the trajectory between each source-destination pair, we used Road Network data\cite{10.1145/2213836.2213872} of Beijing city. Using Dijkstra's shortest path algorithm\cite{10.1145/146585.146588}, we generated a sequence of nodes corresponding to the shortest path for each query.

We use the query data from the first three years of the GPS trajectory dataset \cite{zheng2008understanding}\cite{zheng2009mining}\cite{zheng2010geolife} to decide which paths should be hosted at the edge, given a certain cache size. We use a greedy strategy, namely the fractional knapsack\cite{fractionalknapsack} algorithm, and to enable this we sort all the paths according to their `normalized hit rate'. For any path hosted at the edge, we assume that it will result in a `hit' whenever a query has both its source and destination nodes on the path. Then, the normalized hit rate for a path is defined as the ratio of the number of hits to the total count of nodes forming that path. A path with high normalized hit rate can serve more queries per unit used cache size, if hosted at the edge. Finally, given any cache size, the fractional knapsack algorithm chooses to greedily host paths with the highest normalized hit rates. 

Next, for any given cache size and the corresponding subset of paths hosted at the edge as describe above, we use the query data of the fourth year from the GPS trajectory dataset to estimate the fraction of queries that can be answered at the edge. Figure \ref{fig:789} plots the fraction of test queries served by the edge cache vs the available cache size. Here, we represent the cache size as a fraction, defined as the ratio of the number of nodes corresponding to the cached paths and the total number of nodes involved in all paths that appear in the first three years of the GPS trajectory dataset. For any fraction $\alpha$, the value on the curve\footnote{Note that the fraction of queries answered is less than $1$ even when the cache size is large enough to store all the paths. This is because the cache content is decided based on queries from the first three years, whereas the fraction of queries answered is decided based on the data from the fourth year wherein new paths were queried whose source-destination pairs did not lie on any of the previously seen paths.} is used as a proxy for $1-g(\alpha)$. More details about the above procedure are available in \cite{GpsRoadNetworkAnalysis}.

\begin{figure}
\centering
\parbox{7.3cm}{
\includegraphics[width=7.3cm]{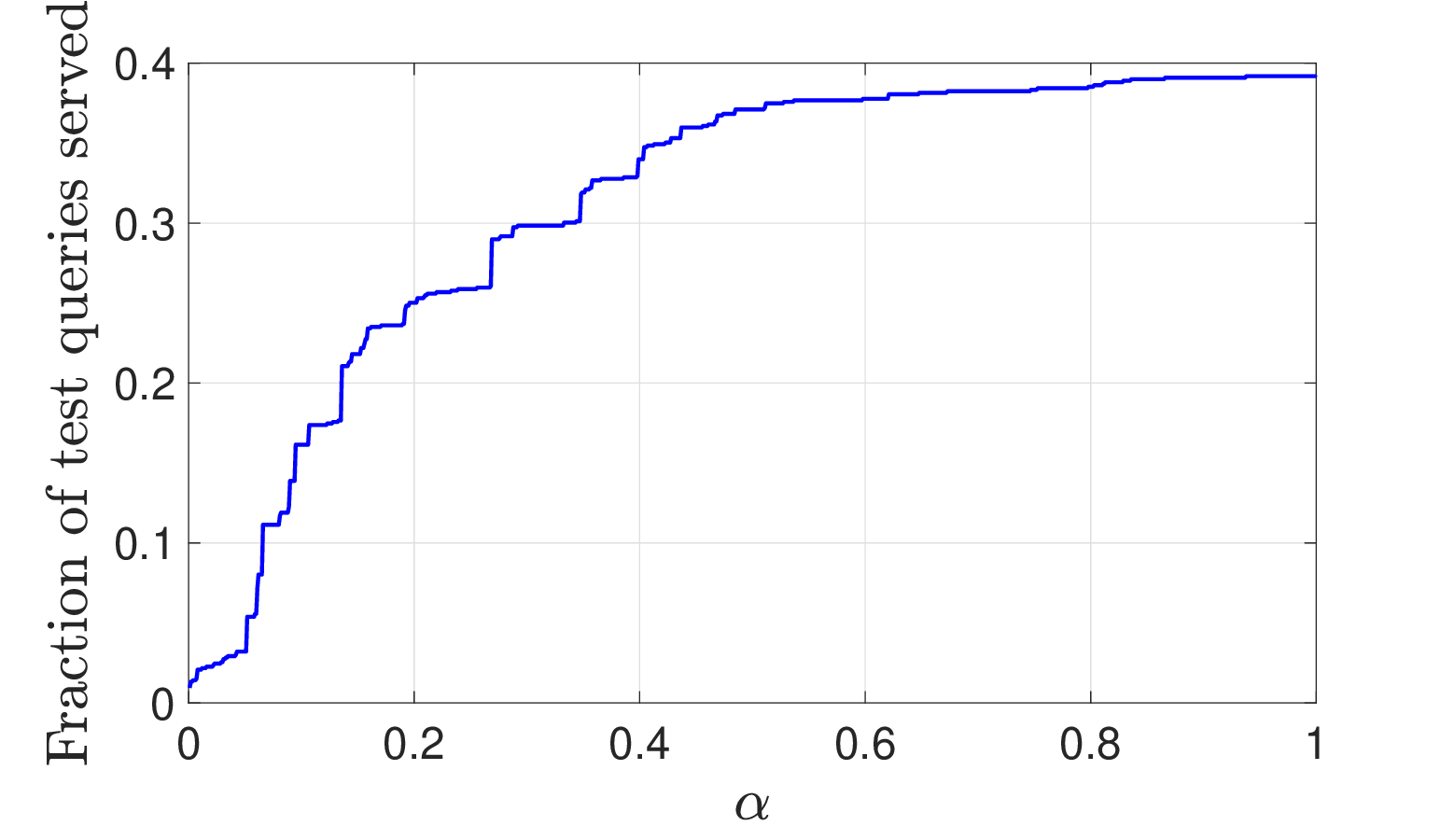}
\caption{Percentage of queries answered by the fraction of cache stored at the edge}
\label{fig:789}}
\quad
\begin{minipage}{7.3cm}
\includegraphics[width=7.3cm]{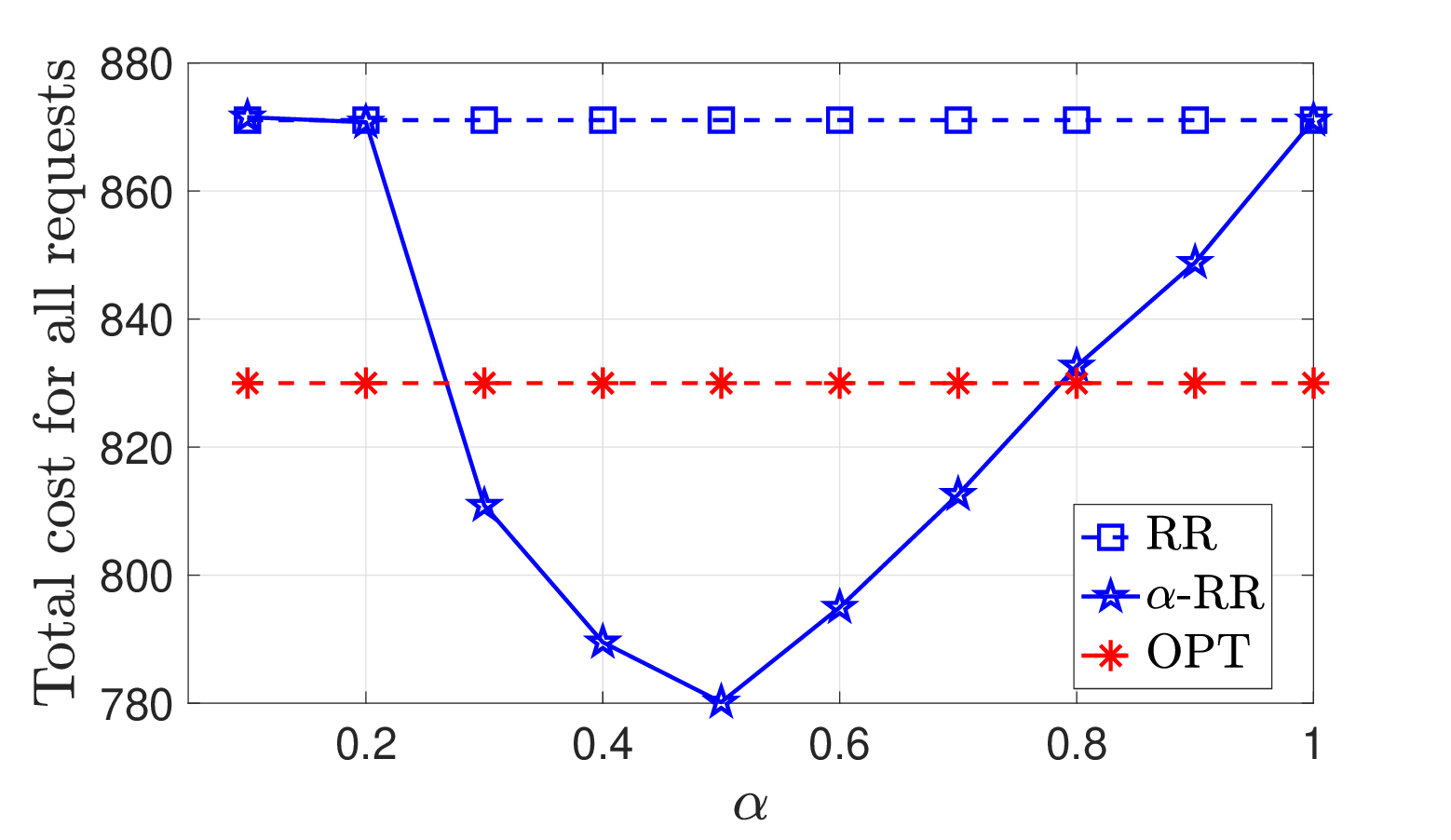}
\caption{Total cost of all requests as a function of cache stored when $c = 0.4$}
\label{fig:costvsalpha}
\end{minipage}
\end{figure}


We now compare the performance of RR and $\alpha$-RR policies on the queries from the fourth year of the GPS trajectory dataset. In Figure \ref{fig:costvsalpha}, we plot the total cost of all requests as a function of cache capacity $\alpha$ with $c = 0.4$. We observe that the cost for $\alpha$-RR is minimum when $\alpha$ is 0.5.
We use this cache size and rent cost $c = 0.4$ in Figure \ref{fig:costvsm}, where we plot the total cost per time slot as a function of fetch cost ($M$). We observe that $\alpha$-RR performs much better than RR, thus highlighting the potential cost benefits that can be derived by efficiently utilizing partial storage. Moreover, it even outperforms offline optimal policy without partial storage (OPT).


\begin{figure}
	\begin{center}
		\includegraphics[width=0.5\columnwidth]{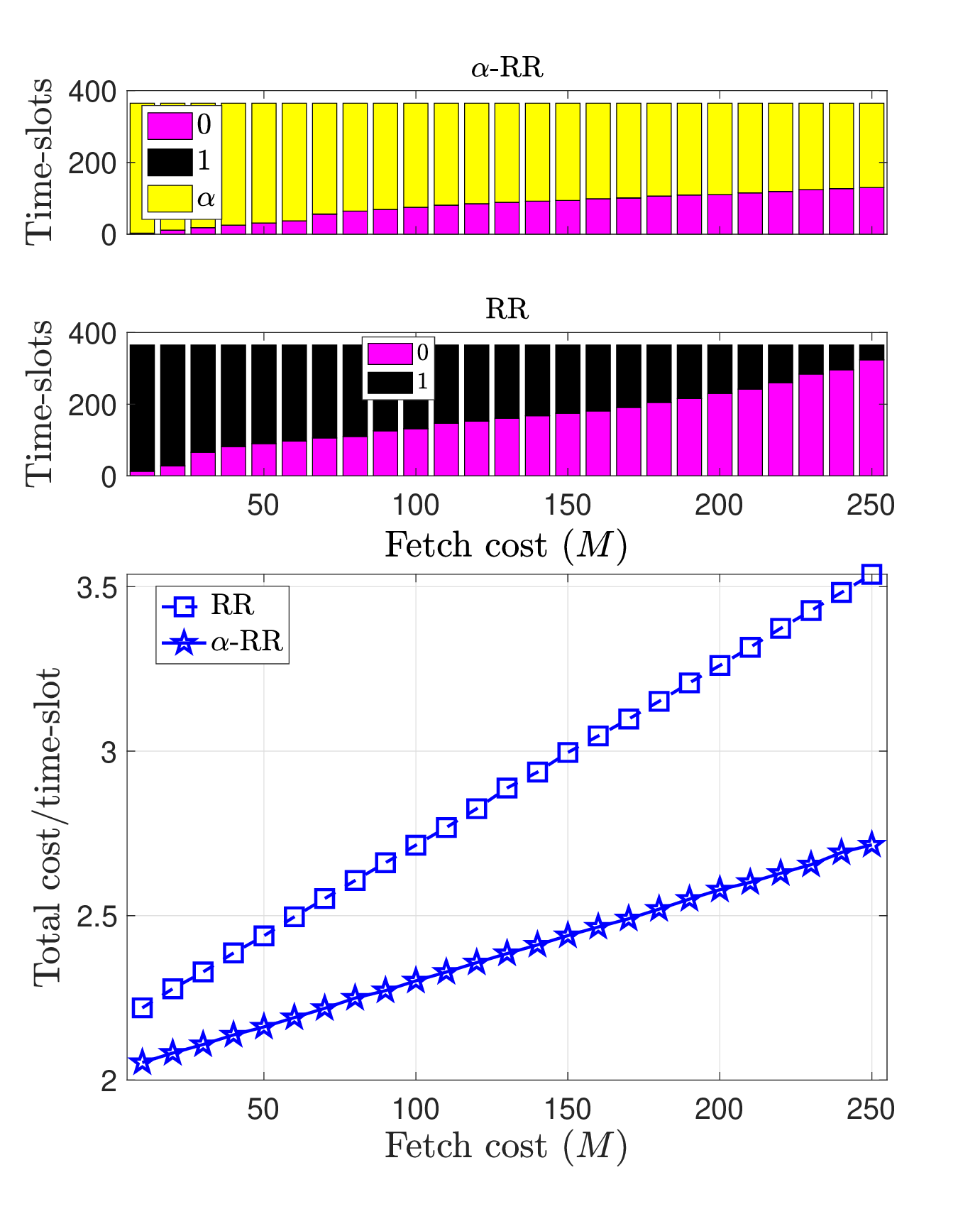}
		\caption{Histogram of hosting status \& total cost per time slot as a function of fetch cost ($M$) for the case when $\alpha = 0.5$ and $c = 0.4$}\label{fig:costvsm}
	\end{center}
\end{figure}

%% file: conclusions.tex
We consider the algorithmic challenge   of  dynamic  service  hosting  on third-party edge computing platforms  in  the setting  where  the  service  can  also  be  partially hosted.  We show that the  benefits  of  partial hosting  are  limited  if reduction in latency cost at the intermediate hosting level is not more than fraction of service hosted. We propose an online policy named $\alpha-$RR and provide performance guarantees for both adversarial and stochastic request arrival models. In addition, via simulations for synthetic and trace-based arrivals, we show that $\alpha-$RR performs well for a wide array of request arrival processes  and  rent cost sequences. The benefits of using more than three levels of service hosting is an open problem.

%% file: appendix_proofs_adversarial.tex
In this section, we discuss the proofs of the results presented in Section \ref{sec:analyticalResults_adversarial}.  Note that we use Model  \ref{model:partial} to prove results in this section. 
\subsection{Proof of Theorem \ref{thm:opt_nointermediate}}

We use the following lemmas about the offline optimal policy $\alpha$-OPT to prove our results. 

The first lemma characterizes sufficient conditions on the various costs and arrival/rent costs processes over a time interval in which the hosting status goes from $a$ to $a+b$ and then returns to $a$ for $(a,b)=(0,\alpha)$, $(a,b)=(0,1)$, and $(a,b)=(\alpha,1-\alpha)$.
\begin{lemma}\label{lem:lemma_opt}
Let $r^*_t$ denote the hosting state under $\alpha$-OPT policy in time-slot $t$. If $r_{n-1}^* = a$, $r_t^* = a+b$ for $n \leq t \leq m$ and $r_{m+1}^* = a$, then for $(a,b)=(0,\alpha)$, $(a,b)=(0,1)$, and $(a,b)=(\alpha,1-\alpha)$, 
 $$
 (g(a)-g(a+b))\displaystyle\sum_{l=n}^m x_l \geq b M+b\displaystyle\sum_{l=n}^m c_l.
$$
\end{lemma}
\begin{proof}
If $r_{n-1}^* = a$, $r_t^* = a+b$ for $n \leq t \leq m$ and $r_{m+1}^* = a$, then the cost incurred by $\alpha$-OPT in $n \leq t \leq m$ is $C^{\alpha-\text{OPT}}=bM+(a+b)\displaystyle\sum_{l=n}^m c_l+g(a+b)\displaystyle\sum_{l=n}^m x_l.$ 
We prove the result by contradiction. Let us assume that $(g(a)-g(a+b))\displaystyle\sum_{l=n}^m x_l < b M+b\displaystyle\sum_{l=n}^m c_l.$

 We construct another policy $\eta$ which behaves same as $\alpha$-OPT except that $r^{\eta}_t = a$ for $n \leq t \leq m.$
The total cost incurred by $\eta$ in $n \leq t \leq m$ is $C^{\eta}=a\displaystyle\sum_{l=n}^m c_l+g(a)\displaystyle\sum_{l=n}^m x_l.$ It follows that
 $C^{\eta}-C^{\alpha-\text{OPT}}=(g(a)-g(a+b)) \displaystyle\sum_{l=n}^m x_l-bM-b\displaystyle\sum_{l=n}^m c_l,$ which is negative by our assumption. This contradicts the definition of the $\alpha$-OPT policy, thus proving the result.
 \end{proof}

The next lemma characterizes a sufficient condition on the various costs and arrival/rent costs processes over an interval of time which ensures that the hosting status under $\alpha$-OPT does not remain static in that interval.

\begin{lemma}\label{lem:OPT_download}
Let $r^*_t$ denote the hosting state under the offline optimal policy in time-slot $t$. If $r_{n-1}^* = a$, and $(g(a)-g(a+b))\displaystyle\sum_{l=n}^m x_l \geq bM+b\displaystyle\sum_{l=n}^m c_l,$ then for $(a,b)=(0,\alpha)$, $(a,b)=(0,1)$ or $(a,b)=(\alpha,1-\alpha)$, $r^*_{\tau} \neq a$ for at least one value of $\tau$ such that $n\leq \tau \leq m.$
\end{lemma}
\begin{proof}
 We prove Lemma \ref{lem:OPT_download} by contradiction. Let us assume that $r^*_{\tau} = a$ during $n\leq \tau \leq m$. So the cost incurred by $\alpha$-OPT in $n \leq t \leq m$ is $C^{\alpha-\text{OPT}}=a\displaystyle\sum_{l=n}^m c_l+g(a)\displaystyle\sum_{l=n}^m x_l.$ 
 
 We  construct another policy $\eta$ which behaves same as $\alpha$-OPT except that $r^{\eta}_t = a+b$ for $n \leq t \leq m.$. 
 The total cost incurred by $\eta$ in $n \leq t \leq m$ is $C^{\eta}=Mb+(a+b)\displaystyle\sum_{l=n}^m c_l+g(a+b)\displaystyle\sum_{l=n}^m x_l.$ It follows that
 $C^{\eta}-C^{\alpha-\text{OPT}}= bM+b\displaystyle\sum_{l=n}^m c_l-(g(a)-g(a+b))\displaystyle\sum_{l=n}^m x_l,$ which is negative. Hence there exists at least one policy $\eta$ which performs better than $\alpha$-OPT. This contradicts the definition of the $\alpha$-OPT policy, thus proving the result.
\end{proof}

The following lemma characterizes sufficient conditions under which $\alpha$-OPT never fetches specific parts of the service. 
\begin{lemma}\label{lem:min_rent_cost}
If $bc_{\text{min}}+g(b)\geq 1$ for $b \in \{\alpha, 1\}$ then $\alpha$-OPT does not fetch $b$ fraction of the service.
Similarly, if $\alpha c_{\text{min}}+g(\alpha)\leq c_{\text{min}}$ then $\alpha$-OPT does not fetch $(1-\alpha)$ fraction of the service.
\end{lemma}
\begin{proof}
We prove this by contradiction.
When $bc_{\text{min}}+g(b)\geq 1$, assume that $\alpha$-OPT fetches $b$ fraction of service for $b \in \{\alpha, 1\}$ at the end of time-slot $n$ and hosts till the time-slot $m$.
The cost under $\alpha$-OPT during the  time $(n, m]$ is $C^{\alpha-\text{OPT}}=bM+g(b)\displaystyle\sum_{l=n+1}^{m} x_l+b\displaystyle\sum_{l=n+1}^{m} c_l$.

We construct another policy $\eta$ which behaves same as $\alpha$-OPT except that during $(n, m]$, the policy $\eta$ forwards all incoming requests. The cost under $\eta$ during the  time $(n, m]$ is $C^{\eta}=\displaystyle\sum_{l=n+1}^{m} x_l$.
The difference $C^{\eta}-C^{\alpha-\text{OPT}}$ is negative since $bc_{\text{min}}+g(b)\geq 1$. Which is a contradiction. 

Now consider the other case $\alpha c_{\text{min}}+g(\alpha)\leq c_{\text{min}}$. Let $r^*_t$ denote the hosting state under $\alpha$-OPT in a time-slot $t$. If $r^*_n=\alpha$, assume that $\alpha$-OPT fetches remaining $(1-\alpha)$ fraction of service for  at the end of time-slot $n$ and hosts till the time-slot $m$.
The cost under $\alpha$-OPT during the  time $(n, m]$ is $C^{\alpha-\text{OPT}}=(1-\alpha)M+\displaystyle\sum_{l=n+1}^{m} c_l$.

We construct another policy $\eta$ which behaves same as $\alpha$-OPT except that during $(n, m]$, the policy $\eta$ hosts only $\alpha$ fractions of the service. The cost under $\eta$ during the  time $(n, m]$ is $C^{\eta}=g(\alpha)\displaystyle\sum_{l=n+1}^{m} x_l+\alpha\displaystyle\sum_{l=n+1}^{m} c_l$.
The difference $C^{\eta}-C^{\alpha-\text{OPT}}$ is negative since $\alpha c_{\text{min}}+g(\alpha)\leq c_{\text{min}}$. Which is a contradiction. 
\end{proof}

\color{black}
The next lemma characterizes a lower limit on the number of time-slots for which under $\alpha$-OPT, the hosting status remains unchanged following a service fetch.
\begin{lemma}\label{lem:OPT_slots}
For $(a,b)=(0,\alpha)$ or $(a,b)=(0,1)$ or $(a,b)=(\alpha,1-\alpha)$ and $(a+b)c_{\text{min}}+g(a+b)<1$,
once $\alpha$-OPT fetches $b \in \{1,\alpha, 1-\alpha\}$ fraction of the service, the hosting status remains unchanged for the next $\frac{bM}{g(a)-g(a+b)-bc_{\text{min}}}$ time-slots.
\end{lemma}
\begin{proof}
 Suppose $\alpha$-OPT fetches $b$ fraction of the service at the end of the $(n-1)^{\text{th}}$ time-slot and evicts it at the end of time-slot $m > n$.
 From Lemma \ref{lem:lemma_opt}, $(g(a)-g(a+b))\displaystyle\sum_{l=n}^m x_l \geq bM+b\displaystyle\sum_{l=n}^m c_l.$
 Since $\displaystyle\sum_{l=n}^m x_l \leq (m-n+1)$ and $\displaystyle\sum_{l=n}^m c_l \geq (m-n+1)c_{\text{min}}$ , $(g(a)-g(a+b))(m-n+1) \geq Mb+b(m-n+1) c_{\text{min}}$, i.e, $(m-n+1) \geq \frac{bM}{g(a)-g(a+b)-bc_{\text{min}}}.$ This proves the result.
\end{proof}

\begin{proof}[Proof of Theorem \ref{thm:opt_nointermediate}]
	
We first focus on the offline optimal policy and prove the result by contradiction.
	We start with $r^*_{n-1}=0$ and  $r^*_n= \alpha$. 
	Hence from Lemma \ref{lem:lemma_opt}, for some $m >n$ 
	\begin{align} \label{eq:2opt}
	(1-g(\alpha))\displaystyle\sum_{l=n}^{m} x_l \geq \alpha\displaystyle\sum_{l=n}^{m} c_l+\alpha M.
	\end{align}
	
	Since $\alpha> 1-g(\alpha)$, from (\ref{eq:2opt}) we have $\displaystyle\sum_{l=n}^{m} x_l \geq \displaystyle\sum_{l=n}^{m} c_l+ M$.
	We construct a policy $\eta$ which behaves same as $\alpha$-OPT except that $r^{\eta}_n =1$ for $n \leq t \leq m.$
	The total cost incurred by $\eta$ in $n \leq t \leq m$ is $C^{\eta}(n,m)=M+\displaystyle\sum_{l=n}^{m} c_l.$
	Where as the total cost incurred by $\alpha$-OPT in $n \leq t \leq m$ is  $C^{\text{$\alpha$-OPT}}(n,m)=  \alpha\displaystyle\sum_{l=n}^{m} c_l+\alpha M+g(\alpha)\displaystyle\sum_{l=n}^{m} x_l$.
	It follows that
	\begin{align*}
	&C^{\eta}(n,m)-C^{\text{$\alpha$-OPT}}(n,m)\\
	&=(1-\alpha)M+(1-\alpha)\displaystyle\sum_{l=n}^{m} c_l-g(\alpha)\displaystyle\sum_{l=n}^{m} x_l\\
	& < g(\alpha)(\displaystyle\sum_{l=n}^{m} c_l+ M-\displaystyle\sum_{l=n}^{m} x_l) < 0.
	\end{align*}

	This shows that fetching and hosting $\alpha$ fraction of service is costlier compared to fetching and renting the entire service during the period $[n, m]$.
	
	 Along similar lines we prove that the other case $r^*_{n-1}=1, r^*_t= \alpha$ for  $n \leq t \leq m$ is not possible under $\alpha$-OPT when $\alpha+g(\alpha)\geq 1.$

We now focus on the $\alpha$-RR policy and prove the result by contradiction.
Let $r^{\alpha\text{-RR}}_{\tilde{t}}=0$ and $r^{\alpha\text{-RR}}_t = \alpha$ for some $t>\tilde{t}$. 
Hence from Algorithm \ref{algo:alphaRR}, for some $\tau_{\alpha} \in (\tilde{t},t)$ 
\begin{equation} \label{eq:2}
(1-g(\alpha))\displaystyle\sum_{l=\tau_{\alpha}}^t x_l \geq \alpha \displaystyle\sum_{l=\tau_{\alpha}}^t c_l+\alpha M.
\end{equation}
Since $\alpha \geq 1-g(\alpha)$, from (\ref{eq:2}) we have $\displaystyle\sum_{l=\tau_{\alpha}}^t x_l \geq \displaystyle\sum_{l=\tau_{\alpha}}^t c_l+M $.
In addition, 
\begin{align*}
 \alpha M&+\alpha\displaystyle\sum_{l=\tau_{\alpha}}^t c_l+g(\alpha) \displaystyle\sum_{l=\tau_{\alpha}}^t x_l\\
 &\geq  \alpha M+\alpha\displaystyle\sum_{l=\tau_{\alpha}}^t c_l+(1-\alpha) \displaystyle\sum_{l=\tau_{\alpha}}^t x_l \\
 &\geq \alpha M+\alpha\displaystyle\sum_{l=\tau_{\alpha}}^t c_l+(1-\alpha) \left(\displaystyle\sum_{l=\tau_{\alpha}}^t c_l+M\right)\\
 &\geq \displaystyle\sum_{l=\tau_{\alpha}}^t c_l+M.
\end{align*}
The last inequality implies that the cost of fetching and hosting $\alpha$ fraction of service incurs more cost than that of fetching and renting the entire service during the period $(\tau_{\alpha},t]$. 
Therefore $r^{\alpha\text{-RR}}_t \neq \alpha$.

Along similar lines we prove that the other case, $r^{\alpha\text{-RR}}_{\tilde{t}}=1, r^{\alpha\text{-RR}}_t = \alpha$ for any $t>\tilde{t}$, is not possible $\alpha\text{-RR}$ when $\alpha+g(\alpha)\geq 1.$

	\end{proof}

\subsection{Proof of Theorem \ref{thm:RR_adv}}

We first focus on the first part of the theorem. 

\begin{proof}[Proof of Theorem \ref{thm:RR_adv}(a)]
We first focus on the case when $c_{\text{min}}\geq 1$ and $r_t=0$ for some $t$ under $\alpha$-RR. In this case, 
$$\text{totalCost}(R_0^{(\tau_0)},I_t) < \text{totalCost}(R_1^{(\tau_1)},I_t) \  \forall \tau_1 \in (t_{\text{recent}},t).$$
Therefore, under $\alpha$-RR, $r_{t+1} \neq 1$. 

When $\alpha c_{\text{min}}+g(\alpha)\geq 1$ and $r_t=0$ for some $t$ under $\alpha$-RR,  $$\text{totalCost}(R_0^{(\tau_0)},I_t) < \text{totalCost}(R_{\alpha}^{(\tau_{\alpha})},I_t) \ \forall \tau_{\alpha} \in (t_{\text{recent}},t).$$ 
Therefore, under $\alpha$-RR, $r_{t+1} \neq \alpha$.

Since $r_0 = 0$ under $\alpha$-RR, it follows that, $\alpha$-RR does not fetch the service if $\alpha c_{\text{min}}+g(\alpha)\geq 1$ and $c_{\text{min}}\geq 1$. Also from Lemma \ref{lem:min_rent_cost}, it follows that the $\alpha$-OPT does not fetch any fraction of the service  when  $\alpha c_{\text{min}}+g(\alpha)\geq 1$ and $c_{\text{min}}\geq 1$. Thus the performance of $\alpha$-RR is same as that of $\alpha$-OPT when $\alpha c_{\text{min}}+g(\alpha)\geq 1$ and $c_{\text{min}}\geq 1$. 

\end{proof}

We now focus on the second part of Theorem \ref{thm:RR_adv}. We use the following lemmas about $\alpha$-RR to prove our results. 

\begin{lemma}\label{lem:RR_slots}
Let $r^{\alpha\text{-RR}}_t$ denote the hosting status under $\alpha$-RR in time-slot $t$. If $r^{\alpha\text{-RR}}_t = a$ and $r^{\alpha\text{-RR}}_{t+1} =b$, then for $(a,b)=(0,\alpha)$, $(a,b)=(0,1)$ or $(a,b)=(\alpha,1-\alpha)$,  by the definition of the $\alpha\text{-RR}$ policy, $\exists \tau_{a+b} < n$ such that 
 $$\text{totalCost}(R_a^{(\tau_a)},I_t) > \text{totalCost}(R_{a+b}^{(\tau_{a+b})},I_t) \ \forall \tau_a \in (t_{\text{recent}},t].$$ 

 Then, 
$$(t-\tau_{a+b}) \geq \frac{bM}{g(a)-g(a+b)-bc_{\text{min}}}.$$ 
\end{lemma}
\begin{proof}
Since $\displaystyle\sum_{l=n-\tau+1}^n x_l \leq \tau,$ and $\displaystyle\sum_{l=n-\tau+1}^n c_l \geq \tau c_{\text{min}},$  the results follow. 
\end{proof}

The next lemma characterizes sufficient conditions on the various costs and arrival/rent costs processes over a time interval in which the hosting status goes from $a$ to $a+b$ and then returns to $a$ for $(a,b)=(0,\alpha)$, $(a,b)=(0,1)$, and $(a,b)=(\alpha,1-\alpha)$.

\begin{lemma}\label{lem:max_requests}
Let $r^{\alpha\text{-RR}}_t$ denote the hosting status under $\alpha$-RR in time-slot $t$. If $r^{\alpha\text{-RR}}_{n-1}=a+b$, $r^{\alpha\text{-RR}}_t=a$ for $n \leq t \leq m$ and $r^{\alpha\text{-RR}}_{m+1}=a+b$ then for $(a,b)=(0,\alpha)$, $(a,b)=(0,1)$ or $(a,b)=(\alpha,1-\alpha)$,
	$$
	(g(a)-g(a+b))\displaystyle\sum_{l=n}^m x_l < b\displaystyle\sum_{l=n}^{m-1} c_l+Mb+(g(a)-g(a+b)).
	$$   
\end{lemma}
\begin{proof}
 Given $r^{\alpha\text{-RR}}_m=a$ and $r^{\alpha\text{-RR}}_{m+1}=a+b.$ So from Algorithm \ref{algo:alphaRR}, we have that
 $\text{totalCost}(R_a^n,I_m) \geq \text{totalCost}(R_{a+b}^{(\tau_{a+b})},I_m) \ \text{ for } \tau_{a+b} \in [n,m]$
 and hence
 $\text{totalCost}(R_a^{n},I_{m-1}) < \text{totalCost}(R_{a+b}^n,I_{m-1})$. Which implies
 $g(a)\displaystyle\sum_{l=n}^{m-1} x_l+a\displaystyle\sum_{l=n}^{m-1} c_l < g(a+b)\displaystyle\sum_{l=n}^{m-1} x_l+ (a+b)\displaystyle\sum_{l=n}^{m-1} c_l+bM$.\\
 By definition,
 \begin{align*}
  &\displaystyle\sum_{l=n}^m x_l= \left(\displaystyle\sum_{l=n}^{m-1} x_l\right)+x_{m}\\
 &(g(a)-g(a+b))\displaystyle\sum_{l=n}^{m-1} x_l  <b(\displaystyle\sum_{l=n}^m c_l+M)+(g(a)-g(a+b))\\
  &< b\displaystyle\sum_{l=n}^{m-1} c_l+Mb+(g(a)-g(a+b)).
 \end{align*}
 Thus proving the result.
\end{proof}

The next lemma characterizes sufficient conditions on the various costs and arrival/rent costs processes over a time interval in which the hosting status goes from $a+b$ to $a$ and then returns to $a+b$ for $(a,b)=(0,\alpha)$, $(a,b)=(0,1)$ or $(a,b)=(\alpha,1-\alpha)$.
\begin{lemma}\label{lem:min_requests}
Let $r^{\alpha\text{-RR}}_t$ denote the hosting status under $\alpha$-RR in time-slot $t$.
  If $r_{n-1}^*= a+b,  r^*_n=a, r^{\alpha\text{-RR}}_t =a+b$ for $n \leq t \leq m$ and  $r^{\alpha\text{-RR}}_{m+1}=a.$ then for $(a,b)=(0,\alpha)$, $(a,b)=(0,1)$ or $(a,b)=(\alpha,1-\alpha)$, $$(g(a)-g(a+b))\displaystyle\sum_{l=n}^{m} x_l\geq b\displaystyle\sum_{l=n}^{m-1} c_l-Mb.$$    	
 \end{lemma}
 \begin{proof} 
 	Given $r^{\alpha\text{-RR}}_m=a+b$ and $r^{\alpha\text{-RR}}_{m+1}=a.$ So from Algorithm \ref{algo:alphaRR}, we have
 	$\text{totalCost}(R_{a+b}^n,I_m) > \text{totalCost}(R_a^{(\tau_a)},I_m)$  for  $\tau_a \in [n,m]$ and hence
 	$\text{totalCost}(R_a^n,I_{m-1}) \geq \text{totalCost}(R_{a+b}^n,I_{m-1})$.
 	Which implies
 	$g(a)\displaystyle\sum_{l=n}^{m-1} x_l+a\displaystyle\sum_{l=n}^{m-1} c_l+Mb \geq g(a+b)\displaystyle\sum_{l=n}^{m-1} x_l+(a+b)\displaystyle\sum_{l=n}^{m-1} c_l$.\\
 By definition,
 \begin{align*}
 \displaystyle\sum_{l=n}^m x_l&= \left(\displaystyle\sum_{l=n}^{m-1} x_l\right)+x_{m}\\
(g(a)-g(a+b))\displaystyle\sum_{l=n}^m x_l &\geq b\displaystyle\sum_{l=n}^{m-1} c_l-M+0\\
 &\geq b\displaystyle\sum_{l=n}^{m-1}x_l-Mb.
 \end{align*}
 Thus proving the result.
 \end{proof}
 The next lemma characterizes upper limit on the past time from which $\alpha$-RR compares costs  to make a fetching decision.   

 \begin{lemma}\label{lem:RR_goback_time}
 Let $r^*_t$, $r^{\alpha\text{-RR}}_t$ denote the hosting states under $\alpha$-OPT, $\alpha$-RR policies respectively in time-slot $t$.
If $r_{n-1}^*=r^{\alpha\text{-RR}}_{n-1}\neq r^*_n$, then  $\alpha\text{-RR}$ checks back only till time-slot $n$ to make a fetching decision.
\end{lemma}
\begin{proof}
We prove this result by contradiction.
Suppose $r_{n-1}^*=r^{\alpha\text{-RR}}_{n-1}=a$, $r^*_n =a+b$ and let $r^{\alpha\text{-RR}}_{m_1}=a+b$  at some $m_1 \in [n+1,m]$ for $(a,b)=(0,\alpha)$, $(a,b)=(0,1)$, and $(a,b)=(\alpha,1-\alpha)$.
  This implies there exists an $n_1<m_1$ such that $\text{totalCost}(R_a^{n_1},I_{m_1})\geq \text{totalCost}(R_{a+b}^{n_1},I_{m_1}).$
  That is  $(g(a)-g(a+b))\displaystyle\sum_{l=n_1}^{m_1} x_l \geq b \displaystyle\sum_{l=n_1}^{m_1} c_l+b M$. 
  If $n_1<n$ then we have $(g(a)-g(a+b))\displaystyle\sum_{l=n}^{m_1} x_l < b \displaystyle\sum_{l=n}^{m_1} c_l+b M$ that implies $(g(a)-g(a+b))\displaystyle\sum_{l=n_1}^{n-1} x_l >b \displaystyle\sum_{l=n_1}^{n-1} c_l$.
  We  construct another policy $\eta$ which is same as $\alpha$-OPT except that  $r^{\eta}_t = a+b$ during  $n_1 \leq t < n$. 
  The total cost incurred by $\eta$ in $n_1 \leq t < n$ is $C^{\eta}=bM+b \displaystyle\sum_{l=n_1}^{n} c_l+g(a+b)\displaystyle\sum_{l=n_1}^{n} x_l.$
  The total cost incurred by $\alpha$-OPT in $n_1 \leq t < n$ is $C^{\alpha\text{-OPT}}=b M+g(a)\displaystyle\sum_{l=n_1}^{n-1} x_l+(a+b)c_n+g(a+b)x_n.$
  The difference $C^{\eta}-C^{\alpha\text{-OPT}}=b \displaystyle\sum_{l=n_1}^{n-1} c_l-(g(a)-g(a+b))\displaystyle\sum_{l=n_1}^{n-1} x_l$,  is negative. 

   Thus the result follows.
\end{proof}

 \begin{lemma}\label{lem:download_oneminusalpha}
 Let $r^{\alpha\text{-RR}}_t$ denote the hosting status under $\alpha$-RR in time-slot $t$.
If  $r_{n-1}^*=r^{\alpha\text{-RR}}_{n-1}=\alpha$, $r_t^* =1$ for $n \leq t \leq m$, and $r_{m+1}^*\neq 1$ then $\alpha\text{-RR}$ fetches  $(1-\alpha)$ fraction of service at some time-slot in $[n+1, m]$. 
After fetching $(1-\alpha)$ fraction of service, $\alpha\text{-RR}$ does not evict any fraction of service till the end of time-slot $m.$
\end{lemma}
\begin{proof}
 Since $r^*_{n-1}=\alpha$, $r_t^* =1$ for $n \leq t \leq m$, and $r_{m+1}^* \neq 1$, the total cost incurred by $\alpha$-OPT in $n \leq t \leq m$ is $C^{\alpha\text{-OPT}}=\displaystyle\sum_{l=n}^{m} c_l+(1-\alpha)M.$
 Assume that $\alpha\text{-RR}$ evicts $\alpha$ fraction of service at $m_1 \in [n+1,m]$. Therefore there exists $n< n_1<m_1$  such that $\text{totalCost}(R_{\alpha}^{n_1},I_{m_1})\geq \text{totalCost}(R_0^{n_1},I_{m_1})$.  Which implies
 $(1-g(\alpha))\displaystyle\sum_{l=n_1}^{m_1} x_l+\alpha M < \alpha \displaystyle\sum_{l=n_1}^{m_1} c_l$ and hence $\displaystyle\sum_{l=n_1}^{m_1} x_l + M< \displaystyle\sum_{l=n_1}^{m_1} c_l$ when $\alpha+g(\alpha)<1$. 
 We  construct another policy $\eta$ which is same $\alpha$-OPT except that  $r^{\eta}_t = 0$ during  $n_1 \leq t \leq m_1$. 
 The total cost incurred by $\eta$ in $n \leq t \leq m$ is $C^{\eta}=(1-\alpha)M+\displaystyle\sum_{l=n}^{n_1-1} c_l+\displaystyle\sum_{l=n_1}^{m_1} x_l+M+\displaystyle\sum_{l=m_1+1}^{m-1} c_l.$
 The difference $C^{\eta}-C^{\alpha-\text{OPT}}=\displaystyle\sum_{l=n_1}^{m_1} x_l+ M-\displaystyle\sum_{l=n_1}^{m_1} c_l$,  is negative. Therefore our assumption that $\alpha\text{-RR}$ evicts $\alpha$ fraction of service during $[n+1, m]$ is false.
 
 Assume that $\alpha\text{-RR}$ does not fetch remaining $(1-\alpha)$ fraction of service at any time-slot during $[n+1,m].$
 Therefore for any $n< n_1<m$, $\text{totalCost}(R_{\alpha}^{n_1},I_m)< \text{totalCost}(R_1^{n_1},I_m)$. Which implies $g(\alpha)\displaystyle\sum_{l=n_1}^m x_l < (1-\alpha)\displaystyle\sum_{l=n_1}^m c_l+M(1-\alpha)$ for any $n_1 \in (n,m)$.
 We  construct another policy $\eta$ which is same $\alpha$-OPT except that  $r^{\eta}_t = \alpha$ during  $n \leq t \leq m$.
 The total cost incurred by $\eta$ in $n\leq t \leq m$ is $C^{\eta}=\alpha \displaystyle\sum_{l=n}^{m} c_l+ g(\alpha)\displaystyle\sum_{l=n}^{m} x_l.$
The difference $C^{\eta}-C^{\alpha\text{-OPT}}=g(\alpha)\displaystyle\sum_{l=n}^{m} x_l-(1-\alpha)M-(1-\alpha)\displaystyle\sum_{l=n}^{m} c_l$,  is negative. So our assumption is false.

Assume that after fetching $(1-\alpha)$ fraction of service at $\widetilde{n} \in [n+1, m] $, $\alpha\text{-RR}$ evicts full service at $m_1 \in [\widetilde{n}+1,m]$. 
Therefore there exists $\widetilde{n}< n_1<m_1$  such that
$\text{totalCost}(R_0^{n_1},I_{m_1})< \text{totalCost}(R_1^{n_1},I_{m_1})$. Which means
$\displaystyle\sum_{l=n_1}^{m_1} x_l+ M < \displaystyle\sum_{l=n_1}^{m_1} c_l$. 
 We  construct another policy $\eta$ which is same $\alpha$-OPT except that  $r^{\eta}_t = 0$ during  $n_1 \leq t \leq m_1$. 
 The total cost incurred by $\eta$ in $n \leq t \leq m$ is $C^{\eta}=M(1-\alpha)+\displaystyle\sum_{l=n}^{n_1-1} c_l+\displaystyle\sum_{l=n_1}^{m_1} x_l+M+\displaystyle\sum_{l=m_1+1}^{m} c_l.$
 The difference $C^{\eta}-C^{\alpha-\text{OPT}}=\displaystyle\sum_{l=n_1}^{m_1} x_l+M-\displaystyle\sum_{l=n_1}^{m_1} c_l$, is negative. So our assumption is false.
 
 Assume that after fetching $(1-\alpha)$ fraction of service at $\widetilde{n} \in [n+1, m] $, $\alpha\text{-RR}$ evicts $(1-\alpha)$ service at $m_1 \in [\widetilde{n}+1,m]$. 
Therefore there exists $\widetilde{n}< n_1<m_1$  such that 
$\text{totalCost}(R_{\alpha}^{n_1},I_{m_1})< \text{totalCost}(R_1^{n_1},I_{m_1})$. Which implies
$g(\alpha)\displaystyle\sum_{l=n_1}^{m_1} x_l+(1-\alpha) M < (1-\alpha)\displaystyle\sum_{l=n_1}^{m_1} c_l$. 
 We  construct another policy $\eta$ which is same $\alpha$-OPT except that  $r^{\eta}_t = \alpha$ during  $n_1 \leq t \leq m_1$. 
 The total cost incurred by $\eta$ in $n \leq t \leq m$ is $C^{\eta}=M(1-\alpha)+\displaystyle\sum_{l=n}^{n_1-1} c_l+g(\alpha)\displaystyle\sum_{l=n_1}^{m_1} x_l+\alpha \displaystyle\sum_{l=n_1}^{m_1} c_l+M(1-\alpha)+\displaystyle\sum_{l=m_1+1}^{m} c_l.$
 The difference $C^{\eta}-C^{\alpha-\text{OPT}}=g(\alpha)\displaystyle\sum_{l=n_1}^{m_1} x_l+M(1-\alpha)-(1-\alpha)\displaystyle\sum_{l=n_1}^{m_1} c_l$, is negative. So our assumption is false.
 This proves the result.
\end{proof}

\begin{lemma}\label{lem:evict_oneminusalpha}
Let $r^{\alpha\text{-RR}}_t$ denote the hosting status under $\alpha$-RR in time-slot $t$.
If $r_{n-1}^*=r^{\alpha\text{-RR}}_{n-1}=1$, $r_t^* =\alpha$ for $n \leq t \leq m$, and $r_{m+1}^*=1$ then $\alpha\text{-RR}$ evicts  $(1-\alpha)$ fraction of service at some time-slot in $[n+1, m]$. 
After evicting $(1-\alpha)$ fraction of service, $\alpha\text{-RR}$ does not evict $\alpha$ fraction of service or does not fetch $(1-\alpha)$ fraction of service till the end of time-slot $m.$
\end{lemma}
\begin{proof}
 Since $r_{n-1}^*=1$, $r_t^* =\alpha$ for $n \leq t \leq m$, and $r_{m+1}^*\neq \alpha$ the total cost incurred by $\alpha$-OPT in $n \leq t \leq m$ is $C^{\alpha\text{-OPT}}=g(\alpha)\displaystyle\sum_{l=n}^m x_l +\alpha \displaystyle\sum_{l=n}^m c_l.$
 Assume that $\alpha\text{-RR}$  evicts full service at  $m_1 \in [n+1,m]$. Therefore there exists $n_1 \in (n, m_1]$ such that
 $\text{totalCost}(R_0^{n_1},I_{m_1})< \text{totalCost}(R_1^{n_1},I_{m_1})$ and $\text{totalCost}(R_0^{n_1},I_{m_1})< \text{totalCost}(R_{\alpha}^{n_1},I_{m_1})$. Which implies
 $\displaystyle\sum_{l=n_1}^{m_1} x_l+ M < \displaystyle\sum_{l=n_1}^{m_1} c_l$ and $(1-g(\alpha))\displaystyle\sum_{l=n_1}^{m_1} x_l+\alpha M< \alpha \displaystyle\sum_{l=n_1}^{m_1} c_l$.
 We  construct another policy $\eta$ which is same $\alpha$-OPT except that  $r^{\eta}_t = 0$ during  $n_1 \leq t \leq m_1$. 
 The total cost incurred by $\eta$ in $n \leq t \leq m$ is $C^{\eta}=\alpha\displaystyle\sum_{l=n}^{n_1-1} c_l+g(\alpha)\displaystyle\sum_{l=n}^{n_1-1} x_l+\displaystyle\sum_{l=n_1}^{m_1} x_l+\alpha M +\alpha\displaystyle\sum_{l=m_1+1}^{m} c_l+g(\alpha)\displaystyle\sum_{l=m_1+1}^{m}x_l.$
 The difference $C^{\eta}-C^{\alpha\text{-OPT}}=(1-g(\alpha))\displaystyle\sum_{l=n_1}^{m_1} x_l+\alpha M- \alpha \displaystyle\sum_{l=n_1}^{m_1} c_l$, is negative. So our assumption is false.
 
Assume that $\alpha\text{-RR}$ does not evict $(1-\alpha)$ fraction of service at any time-slot in $[n+1,m]$.
Therefore for any $ n_1\in (n, m]$,
$\text{totalCost}(R_{\alpha}^{n_1},I_m)\geq \text{totalCost}(R_1^{n_1},I_m)$. Which implies
$g(\alpha)\displaystyle\sum_{l=n_1}^{m} x_l+(1-\alpha) M \geq (1-\alpha) \displaystyle\sum_{l=n_1}^{m} c_l$.
 We  construct another policy $\eta$ which is same $\alpha$-OPT except that  $r^{\eta}_t = 1$ during  $n \leq t \leq m$. 
 The total cost incurred by $\eta$ in $n \leq t \leq m$ is $C^{\eta}=\displaystyle\sum_{l=n}^{m} c_l.$
 When $r_{m+1}^*=1$, the difference $C^{\eta}-C^{\alpha-\text{OPT}}=(1-\alpha) \displaystyle\sum_{l=n}^{m} c_l-g(\alpha)\displaystyle\sum_{l=n}^{m} x_l -(1-\alpha) M$, is negative. So our assumption is false.

 Assume that after evicting $(1-\alpha)$ fraction of service  at $\widetilde{n} \in [n+1, m] $, $\alpha\text{-RR}$ evicts $\alpha$ fraction of service  at $m_1 \in [\widetilde{n}+1,m]$.
 Therefore there exists $\widetilde{n}< n_1<m_1$  such that
 $\text{totalCost}(R_0^{n_1},I_{m_1})< \text{totalCost}(R_{\alpha}^{n_1},I_{m_1})$. Which implies
 $(1-g(\alpha))\displaystyle\sum_{l=n_1}^{m_1} x_l+\alpha M < \alpha \displaystyle\sum_{l=n_1}^{m_1} c_l$. 
 We  construct another policy $\eta$ which is same $\alpha$-OPT except that  $r^{\eta}_t = 0$ during  $n_1 \leq t \leq m_1$. 
 The total cost incurred by $\eta$ in $n \leq t \leq m$ is $C^{\eta}=\alpha \displaystyle\sum_{l=n}^{n_1-1} c_l+g(\alpha)\displaystyle\sum_{l=n}^{n_1-1} x_l+\displaystyle\sum_{l=n_1}^{m_1} x_l+\alpha \displaystyle\sum_{l=m_1+1}^{m} c_l+g(\alpha)\displaystyle\sum_{l=m_1+1}^{m} x_l+\alpha M.$
 The difference $C^{\eta}-C^{\alpha-\text{OPT}}=(1-g(\alpha))\displaystyle\sum_{l=n_1}^{m_1} x_l+\alpha M-\alpha \displaystyle\sum_{l=n_1}^{m_1} c_l $, is negative. So our assumption is false.
 
 Assume that after evicting $(1-\alpha)$ fraction of service at $\widetilde{n} \in [n+1, m] $, $\alpha\text{-RR}$ fetches $(1-\alpha)$ fraction of service  at $m_1 \in [\widetilde{n}+1,m]$.
 Therefore there exists $\widetilde{n}< n_1<m_1$  such that
 $\text{totalCost}(R_{\alpha}^{n_1},I_{m_1})\geq \text{totalCost}(R_1^{n_1},I_{m_1})$. Which implies
 $g(\alpha)\displaystyle\sum_{l=n_1}^{m_1} x_l\geq (1-\alpha) \displaystyle\sum_{l=n_1}^{m_1} c_l+(1-\alpha)M$.
 We  construct another policy $\eta$ which is same $\alpha$-OPT except that  $r^{\eta}_t = 1$ during  $n_1 \leq t \leq m_1$. 
 The total cost incurred by $\eta$ in $n \leq t \leq m$ is $C^{\eta}=\alpha \displaystyle\sum_{l=n}^{n_1-1} c_l+g(\alpha)\displaystyle\sum_{l=n}^{n_1-1} x_l+(1-\alpha) M+\displaystyle\sum_{l=n_1}^{m_1} c_l+\alpha \displaystyle\sum_{l=m_1+1}^{m} c_l+g(\alpha)\displaystyle\sum_{l=m_1+1}^{m} x_l.$
 The difference $C^{\eta}-C^{\alpha-\text{OPT}}=(1-\alpha) \displaystyle\sum_{l=n_1}^{m_1} c_l+(1-\alpha) M-g(\alpha)\displaystyle\sum_{l=n_1}^{m_1} x_l$,  is negative. So our assumption is false. 
 This proves the result.
 \end{proof}

 \begin{lemma}\label{lem:evict_full}
 Let $r^*_t$ and $r^{\alpha\text{-RR}}_t$ denote the hosting status under $\alpha$-OPT and $\alpha$-RR respectively in time-slot $t$. If $r_{n-1}^*=r^{\alpha\text{-RR}}_{n-1}=1$, $r_t^* =0$ for $n \leq t \leq m$, and $r_{m+1}^* \neq 0$ then $\alpha\text{-RR}$ evicts  full service at some time-slot in $[n+1, m]$. 
After evicting full service, $\alpha\text{-RR}$ does not fetch any fraction of service  till the end of time-slot $m.$
\end{lemma}
\begin{proof}
 Since $r_{n-1}^*=1$, $r_t^* =0$ for $n \leq t \leq m$, and $r_{m+1}^* \neq 0$,  the total cost incurred by $\alpha$-OPT in $n \leq t \leq m$ is $C^{\alpha\text{-OPT}}=\displaystyle\sum_{l=n}^{m} x_l.$
 
 Assume that $\alpha\text{-RR}$ does not evict full service at any $m_1 \in [n+1,m]$. Therefore for any $n\leq n_1<m_1\leq m$,
 $\text{totalCost}(R_0^{n_1},I_{m_1}) \geq \text{totalCost}(R_1^{n_1},I_{m_1})$ and $\text{totalCost}(R_0^{n_1},I_{m_1}) \geq \text{totalCost}(R_{\alpha}^{n_1},I_{m_1})$. Which implies
 $\displaystyle\sum_{l=n_1}^{m_1} x_l+ M \geq \displaystyle\sum_{l=n_1}^{m_1} c_l$ and $\displaystyle\sum_{l=n_1}^{m_1} x_l+ M \geq g(\alpha)\displaystyle\sum_{l=n_1}^{m_1} x_l+ (1-\alpha)M \alpha\displaystyle\sum_{l=n_1}^{m_1} c_l$.
 Suppose $r_{m+1}^*=1$, we  construct another policy $\eta$ which is same $\alpha$-OPT except that  $r^{\eta}_t = 1$ during  $n \leq t \leq m$. When , the total cost incurred by $\eta$ in $n \leq t \leq m+1$ is $C^{\eta}=\displaystyle\sum_{l=n}^{m+1} c_l.$
  In this case the difference $C^{\eta}-C^{\alpha-\text{OPT}}=\displaystyle\sum_{l=n}^{m} c_l-\displaystyle\sum_{l=n}^{m} x_l-M $, is negative. So our assumption is false. Therefore, $\alpha\text{-RR}$ evicts  full service at some time-slot in $[n+1, m]$. Suppose $r_{m+1}^*=\alpha$, we  construct another policy $\eta$ which is same $\alpha$-OPT except that  $r^{\eta}_t = \alpha$ during  $n \leq t \leq m$. When , the total cost incurred by $\eta$ in $n \leq t \leq m+1$ is $C^{\eta}=\alpha \displaystyle\sum_{l=n}^{m+1} c_l+g(\alpha) x_l.$
  In this case the difference $C^{\eta}-C^{\alpha-\text{OPT}}=\alpha\displaystyle\sum_{l=n}^{m} c_l-(1-g(\alpha))\displaystyle\sum_{l=n}^{m} x_l-\alpha M $, is negative.
  
  Assume that after evicting full service at $\widetilde{n} \in [n+1, m] $, $\alpha\text{-RR}$ fetches $\alpha$ fraction of service  at $m_1 \in [\widetilde{n}+1,m]$.
 Therefore there exists $\widetilde{n}< n_1<m_1$  such that
 $\text{totalCost}(R_0^{n_1},I_{m_1}) \geq \text{totalCost}(R_{\alpha}^{n_1},I_{m_1})$. Which implies
 
 $(1-g(\alpha))\displaystyle\sum_{l=n_1}^{m_1} x_l \geq \alpha \displaystyle\sum_{l=n_1}^{m_1} c_l+\alpha M$.
 We  construct another policy $\eta$ which is same $\alpha$-OPT except that  $r^{\eta}_t = \alpha$ during  $n_1 \leq t \leq m_1$. 
 The total cost incurred by $\eta$ in $n \leq t \leq m$ is $C^{\eta}=\displaystyle\sum_{l=n}^{n_1-1} c_l+\alpha \displaystyle\sum_{l=n_1}^{m_1} c_l+g(\alpha)\displaystyle\sum_{l=n_1}^{m_1} x_l+\alpha M+\displaystyle\sum_{l=m_1+1}^{m} x_l.$
 The difference $C^{\eta}-C^{\alpha-\text{OPT}}=\alpha M+ \alpha\displaystyle\sum_{l=n_1}^{m_1} c_l-(1-g(\alpha)) \displaystyle\sum_{l=n_1}^{m_1} x_l $, is negative. So our assumption is false.
 
 Assume that after evicting full service at $\widetilde{n} \in [n+1, m] $, $\alpha\text{-RR}$ fetches full service  at $m_1 \in [\widetilde{n}+1,m]$.
 Therefore there exists $\widetilde{n}< n_1<m_1$  such that
 $\text{totalCost}(R_0^{n_1},I_{m_1}) \geq \text{totalCost}(R_1^{n_1},I_{m_1})$.Which implies
 $\displaystyle\sum_{l=n_1}^{m_1} x_l \geq \displaystyle\sum_{l=n_1}^{m_1} c_l+ M$. 
 We  construct another policy $\eta$ which is same $\alpha$-OPT except that  $r^{\eta}_t = 1$ during  $n_1 \leq t \leq m_1$. 
 The total cost incurred by $\eta$ in $n \leq t \leq m$ is $C^{\eta}=\displaystyle\sum_{l=n}^{n_1-1} x_l+ M+\displaystyle\sum_{l=n_1}^{m_1} c_l+\displaystyle\sum_{l=m_1+1}^{m} x_l.$
 The difference $C^{\eta}-C^{\alpha-\text{OPT}}= M+\displaystyle\sum_{l=n_1}^{m_1} c_l-\displaystyle\sum_{l=n_1}^{m_1} x_l $ is negative. So our assumption is false. 

 This proves the result.
 \end{proof}

 \begin{lemma}\label{lem:download_full}
 Let $r^{\alpha\text{-RR}}_t$ denote the hosting status under $\alpha$-RR in time-slot $t$.
If $r_{n-1}^*=r^{\alpha\text{-RR}}_{n-1}=0$, $r_t^* =1$ for $n \leq t \leq m$, and $r_{m+1}^* \neq 1$ then $\alpha\text{-RR}$ fetches  full service at some time-slot in $[n+1, m]$. 
After fetching full service, $\alpha\text{-RR}$ does not evict any fraction of service till the end of time-slot $m.$
\end{lemma}
\begin{proof}
  Since $r_{n-1}^*=0$, $r_t^* =1$ for $n \leq t \leq m$, and $r_{m+1}^* \neq 1$, the total cost incurred by $\alpha$-OPT in $n \leq t \leq m$ is $C^{\alpha-\text{OPT}}(n,m)=M+\displaystyle\sum_{l=n}^{m} c_l.$
  
If $\alpha+g(\alpha)\geq 1$ then by Theorem \ref{thm:opt_nointermediate}, RR will not fetch $\alpha$ fraction of service, so we focus on the case when $\alpha+g(\alpha)< 1$.
 Assume that $\alpha\text{-RR}$ does not fetch any service during $[n+1,m]$. 
 Therefore for any $n_1 \in [n+1,m]$, $\text{totalCost}(R_0^{n_1},I_m) < \text{totalCost}(R_1^{n_1},I_m)$. Which implies
 
 $\displaystyle\sum_{l=n_1}^m x_l < \displaystyle\sum_{l=n_1}^m c_l +M.$
 We  construct another policy $\eta$ which is same $\alpha$-OPT except that  $r^{\eta}_t = 0$ during  $n \leq t \leq m$. 
 The total cost incurred by $\eta$ in $n \leq t \leq m$ is $C^{\eta}=\displaystyle\sum_{l=n}^{m} x_l.$
 The difference $C^{\eta}-C^{\alpha-\text{OPT}}=\displaystyle\sum_{l=n}^{m} x_l-\displaystyle\sum_{l=n_1}^{m_1} c_l-M$ is negative. So our assumption is false.

 If $\alpha\text{-RR}$  fetches  $\alpha$ fraction of service at any $\widetilde{n} \in [n+1,m)$, assume that it evicts the service at $m_1 \in [\widetilde{n}+1,m]$.
 Therefore there exists $\widetilde{n}< n_1<m_1$ such that $\text{totalCost}(R_0^{n_1},I_{m_1}) < \text{totalCost}(R_{\alpha}^{n_1},I_{m_1})$.
 Which implies $(1-g(\alpha))\displaystyle\sum_{l=n_1}^{m_1} x_l+\alpha M < \alpha \displaystyle\sum_{l=n_1}^{m_1} c_l$. Using $\alpha<1-g(\alpha)$ we get $\displaystyle\sum_{l=n_1}^{m_1} x_l+M<\displaystyle\sum_{l=n_1}^{m_1} c_l$
 We  construct another policy $\eta$ which is same $\alpha$-OPT except that  $r^{\eta}_t = 0$ during  $n_1 \leq t \leq m_1$. 
 The total cost incurred by $\eta$ in $n_1 \leq t \leq m_1$ is $C^{\eta}=\displaystyle\sum_{l=n_1}^{m_1} x_l+M.$
 The difference $C^{\eta}-C^{\alpha-\text{OPT}}=\displaystyle\sum_{l=n_1}^{m_1} x_l+M-\displaystyle\sum_{l=n_1}^{m_1} c_l$, is negative. So our assumption is false.
 
 When  $\alpha\text{-RR}$  fetches  $\alpha$ fraction of service at any $\widetilde{n} \in [n+1,m)$, assume that $\alpha\text{-RR}$ does not fetch remaining $(1-\alpha)$ fraction of service at any $m_1 \in [\widetilde{n}+1,m]$.
 Therefore for any $n< n_1<m_1$, $\text{totalCost}(R_{\alpha}^{n_1},I_{m_1}) < \text{totalCost}(R_1^{n_1},I_{m_1})$. $g(\alpha)\displaystyle\sum_{l=n_1}^{m_1} x_l < (1-\alpha) \displaystyle\sum_{l=n_1}^{m_1} c_l+M(1-\alpha).$
 We  construct another policy $\eta$ which is same $\alpha$-OPT except that  $r^{\eta}_t = \alpha$ during  $n\leq t \leq m$. 
 The total cost incurred by $\eta$ in $n \leq t \leq m$ is $C^{\eta}=\alpha M+g(\alpha)\displaystyle\sum_{l=n}^{m} x_l+\displaystyle\sum_{l=n}^{m} c_l.$
 The difference $C^{\eta}-C^{\alpha-\text{OPT}}=g(\alpha)\displaystyle\sum_{l=n}^{m_1} x_l-\displaystyle\sum_{l=n}^{m_1} c_l-(1-\alpha)M$ is negative. So our assumption is false. Therefore $\alpha\text{-RR}$ fetches  full service at some time-slot in $[n+1, m]$.
 
 Using Lemma \ref{lem:download_oneminusalpha}, we conclude that when $r_t^*=r^{\alpha\text{-RR}}_t=1$ for $t\leq m$ then $\alpha\text{-RR}$ does not evict any fraction of service till the time-slot $m.$
\end{proof}

\begin{lemma}\label{lem:download_alpha}
Let $r^{\alpha\text{-RR}}_t$ denote the hosting status under $\alpha$-RR in time-slot $t$.
If $r_{n-1}^*=r^{\alpha\text{-RR}}_{n-1}=0$, $r_t^* =\alpha$ for $n \leq t \leq m$, and $r_{m+1}^* \neq \alpha$ then $\alpha\text{-RR}$ fetches  $\alpha$ fraction of service at some time-slot in $[n+1, m]$. 
After fetching $\alpha$ fraction of service,$\alpha\text{-RR}$ does not evict $\alpha$ fraction of service or does not fetch $(1-\alpha)$ fraction of service till the end of  time-slot $m.$
\end{lemma}
\begin{proof}
 Since $r_{n-1}^*=0$, $r_t^* =\alpha$ for $n \leq t \leq m$, and $r_{m+1}^* \neq \alpha$, the total cost incurred by $\alpha$-OPT in $n \leq t \leq m$ is $C^{\alpha-\text{OPT}}(n,m)=\alpha M+g(\alpha)\displaystyle\sum_{l=n}^{m} x_l+\alpha \displaystyle\sum_{l=n}^{m} c_l.$
Assume that $\alpha\text{-RR}$  fetches full service at  $m_1 \in [n+1,m]$. Therefore there exists $n< n_1<m_1$, such that 
$\text{totalCost}(R_0^{n_1},I_{m_1}) \geq \text{totalCost}(R_1^{n_1},I_{m_1})$ and $\text{totalCost}(R_{\alpha}^{n_1},I_{m_1}) \geq \text{totalCost}(R_1^{n_1},I_{m_1})$. Which implies
$\displaystyle\sum_{l=n_1}^{m_1} x_l \geq \displaystyle\sum_{l=n_1}^{m_1} c_l+M$ and $g(\alpha)\displaystyle\sum_{l=n_1}^{m_1} x_l \geq (1-\alpha)\displaystyle\sum_{l=n_1}^{m_1} c_l+(1-\alpha)M.$
We  construct another policy $\eta$ which is same $\alpha$-OPT except that  $r^{\eta}_t = 1$ during  $n_1 \leq t \leq m_1$.  
The total cost incurred by $\eta$ in $n \leq t \leq m$ is $C^{\eta}=\alpha M+g(\alpha)\displaystyle\sum_{l=n}^{n_1-1} x_l+\alpha \displaystyle\sum_{l=n}^{n_1-1} c_l+(1-\alpha)M+\displaystyle\sum_{l=n_1}^{m_1} c_l+g(\alpha)\displaystyle\sum_{l=m_1+1}^{m} x_l+\alpha \displaystyle\sum_{l=m_1+1}^{m} c_l.$
The difference $C^{\eta}-C^{\alpha\text{-OPT}}=(1-\alpha)\displaystyle\sum_{l=n_1}^{m_1} c_l+(1-\alpha)M-g(\alpha)\displaystyle\sum_{l=n_1}^{m_1} x_l$ is negative. So our assumption is false.
 
Assume that $\alpha\text{-RR}$ does not fetch $\alpha$ fraction of service at any time-slot in  $[n+1,m]$. Therefore for any $n_1 \in [n+1,m]$, $\text{totalCost}(R_0^{n_1},I_m) < \text{totalCost}(R_{\alpha}^{n_1},I_m)$. Which implies
$(1-g(\alpha))\displaystyle\sum_{l=n_1}^{m} x_l <  \alpha \displaystyle\sum_{l=n_1}^{m} c_l+\alpha M.$
We  construct another policy $\eta$ which is same $\alpha$-OPT except that  $r^{\eta}_t = 0$ during  $n \leq t \leq m$. 
The total cost incurred by $\eta$ in $n \leq t \leq m$ is $C^{\eta}=\displaystyle\sum_{l=n}^{m} x_l.$
The difference $C^{\eta}-C^{\alpha-\text{OPT}}=(1-g(\alpha))\displaystyle\sum_{l=n}^{m} x_l-\alpha \displaystyle\sum_{l=n}^{m} c_l-M\alpha$ is negative. So our assumption is false. Therefore $\alpha\text{-RR}$ fetches  $\alpha$ fraction of service  at some time-slot in $[n+1, m]$.
Using Lemma \label{lem:evict_oneminusalpha}, we conclude that when $r_t^*=r^{\alpha\text{-RR}}_t=\alpha$ for $t\leq m$ then $\alpha\text{-RR}$ does not evict $\alpha$ fraction of service or does not fetch $(1-\alpha)$ fraction of service till the end of  time-slot $m.$
 \end{proof}

We provide a proof  of Theorem \ref{thm:RR_adv}(b). 
To compare the costs incurred by $\alpha$-RR and $\alpha$-OPT we divide time into frames $[1,t_1-1]$, $[t_1,t_2-1], [t_2,t_3-1],\ldots,$ where $t_j-1$ is the time-slot in which $\alpha$-OPT downloads a fraction of or full service for the $j^{\text{th}}$ time for $j\in\{1,2,\ldots\}.$ We have three types of frames. 

\begin{enumerate}
	\item[1.] Type-1 frame: It starts when $\alpha$-OPT fetches the entire service. By Lemma \ref{lem:min_rent_cost}, Type-1 frames exist only if $c_{\text{min}}<1$.
	\item[2.] Type-$\alpha$ frame: It starts when $\alpha$-OPT fetches $\alpha$ fraction of the service. By Lemma \ref{lem:min_rent_cost}, Type-$\alpha$ frames exist only if $\alpha c_{\text{min}}+g(\alpha)<1$.
	\item[3.] Type-$(1-\alpha)$ frame: It starts when $\alpha$-OPT fetches $(1-\alpha)$ fraction of the service. By Lemma \ref{lem:min_rent_cost}, Type-$(1-\alpha)$ frames exist only if $c_{\text{min}}<\alpha c_{\text{min}}+g(\alpha)$.
\end{enumerate}
Refer to Figure \ref{fig:a_OPT_RR_frame_proofoutline} for an illustration of the frame.
\color{violet} Note that the main difference between the proof techniques in this work and in  \cite{narayana2021renting} is that in this work we have three different types of frames possible whereas in \cite{narayana2021renting} there is a possibility of only one frame. \color{black}
We first focus on the Type-$(1-\alpha)$ frame. By Lemma \ref{lem:download_oneminusalpha}, $\alpha$-RR fetches $(1-\alpha)$ fraction of service at sometime $t_{f}^{\alpha\text{-RR}} \in [t_1, \tau]$, and by Lemma \ref{lem:evict_oneminusalpha}, $\alpha$-RR evicts $(1-\alpha)$ fraction of service at sometime $t_{e}^{\alpha\text{-RR}} \geq \tau$.
Both $\alpha$-OPT and $\alpha$-RR makes one fetch in this frame. Hence the difference in the fetch costs is zero. We now focus on the service and rent cost incurred by the two policies. 
\begin{itemize}
 \item [--]  Let $\tau_1 = t_{f}^{\alpha\text{-RR}} - t_1.$
 Since $\alpha\text{-RR}$ does not fetch $(1-\alpha)$ fraction of the service in $[t_1, t_{f}^{\alpha\text{-RR}}]$, the cost incurred in $[t_1, t_{f}^{\alpha\text{-RR}}]$ by $\alpha\text{-RR}$ is 
 $g(\alpha)\displaystyle\sum_{l=t_1}^{t_{f}^{\alpha\text{-RR}}} x_l+ \alpha \displaystyle\sum_{l=t_1}^{t_{f}^{\alpha\text{-RR}}} c_l.$
 $\alpha$-OPT fetches full service in $[t_1, t_{f}^{\alpha\text{-RR}}]$ hence the cost incurred in $[t_1, t_{f}^{\alpha\text{-RR}}]$ by $\alpha$-OPT is $ \displaystyle\sum_{l=t_1}^{t_{f}^{\alpha\text{-RR}}} c_l.$
 By Lemma \ref{lem:max_requests},$$
	g(\alpha)\displaystyle\sum_{l=t_1}^{t_{f}^{\alpha\text{-RR}}} x_l < (1-\alpha)\displaystyle\sum_{l=t_1}^{t_{f}^{\alpha\text{-RR}}-1} c_l +M(1-\alpha)+g(\alpha).
	$$  
 Hence difference in the service and rent cost incurred by $\alpha\text{-RR}$ and $\alpha$-OPT in $[t_1, t_{f}^{\alpha\text{-RR}}]$   is at most $(M-c_{\text{min}})(1-\alpha)+g(\alpha).$
 \item[--] The service and rent cost incurred by $\alpha$-OPT and $\alpha\text{-RR}$ in $[t_{f}^{\alpha\text{-RR}}+1, \tau]$ are equal. 

\item [--] Let $\tau_2 = t_{e}^{\alpha\text{-RR}} - \tau$. 
The cost incurred by $\alpha$-OPT in $[\tau+1, t_{e}^{\alpha\text{-RR}}]$ is $g(\alpha)\displaystyle\sum_{l=\tau}^{t_{e}^{\alpha\text{-RR}}} x_l+ \alpha \displaystyle\sum_{l=\tau}^{t_{e}^{\alpha\text{-RR}}} c_l$
The  cost incurred by $\alpha\text{-RR}$ in $[\tau+1, t_{e}^{\alpha\text{-RR}}]$ is $\displaystyle\sum_{l=\tau}^{t_{e}^{\alpha\text{-RR}}} c_l$.
By Lemma \ref{lem:min_requests},$$g(\alpha)\displaystyle\sum_{l=\tau}^{t_{e}^{\alpha\text{-RR}}} x_l\geq (1-\alpha) \displaystyle\sum_{l=\tau}^{t_{e}^{\alpha\text{-RR}}-1} c_l-M(1-\alpha).$$
Hence difference in the service and rent cost incurred by $\alpha\text{-RR}$ and $\alpha$-OPT in $[t_1, t_{f}^{\alpha\text{-RR}}]$  is at most $(M+c_{\text{max}})(1-\alpha).$

 \item[--] The service and rent cost incurred by $\alpha$-OPT and $\alpha\text{-RR}$ in $[t_{e}^{\alpha\text{-RR}}+1, t_2-1]$ are equal. 
\end{itemize}
We therefore have that,
\begin{align}
 SC_{1-\alpha}^{\alpha\text{-RR}}(j)&-SC_{1-\alpha}^{\text{$\alpha$-OPT}}(j)\leq (2M+c_{\text{max}}-c_{\text{min}})(1-\alpha)+g(\alpha)\nonumber\\
 &\leq 2M(1-\alpha)+(M+1)(1-\alpha)+(1-\alpha)\nonumber\\
 &\leq 3M(1-\alpha)+2(1-\alpha).
 \label{ineq:bound1_RR}
\end{align}
By Lemma \ref{lem:OPT_slots}, once $\alpha$-OPT downloads $(1-\alpha)$ fraction of the service, it will not evict for at least $\frac{(1-\alpha)M}{g(\alpha)-(1-\alpha)c_{\text{min}}}$ slots. Therefore,
\begin{align}
&SC_{1-\alpha}^{\text{$\alpha$-OPT}}(j)\geq (1-\alpha)M+(1-\alpha) c_{\text{min}}\frac{(1-\alpha)M}{g(\alpha)-(1-\alpha)c_{\text{min}}}\nonumber\\
\label{ineq:bound1_OPT}
&\implies M(1-\alpha)\leq \frac{g(\alpha)-(1-\alpha)c_{\text{min}}}{g(\alpha)}C^{\text{$\alpha$-OPT}}(j).
\end{align}
From \eqref{ineq:bound1_RR} and \eqref{ineq:bound1_OPT},
\begin{align}
  SC_{1-\alpha}^{\alpha\text{-RR}}(j) &\leq \left(4+\frac{2}{M}-\frac{(1-\alpha)c_{\text{min}}}{g(\alpha)}(2+\frac{1}{M})\right)SC_{1-\alpha}^{\text{$\alpha$-OPT}}(j)\nonumber\\
  &< \left(4+\frac{2}{M}\right) SC_{1-\alpha}^{\text{$\alpha$-OPT}}(j). \label{ineq:bound2_RRgeneral}
\end{align}
Hence we conclude that in every Type-$(1-\alpha)$ frame, $\alpha\text{-RR}$ is (4+$\frac{2}{M}$)-optimal.
\begin{figure}[ht]
	\centering
	\begin{tikzpicture}
	\foreach \x in {0,0.3,0.6,...,7.2}{
		\draw[] (\x,0) --  (\x+0.3,0);
		\draw[gray] (\x,-1mm) -- (\x,1mm);
	}
	
	\draw[<->,color=black] (0.25,14mm) -- 
	node[above=0mm,pos=0.5]{\rom{1}} (2.35,14mm);
	\draw[<->,color=black] (2.35,14mm) -- 
	node[above=0mm,pos=0.5]{\rom{2}} (4.15,14mm);
	\draw[<->,color=black] (4.15,14mm) -- 
	node[above=0mm,pos=0.5]{\rom{3}} (5.65,14mm);
	\draw[<->,color=black] (5.65,14mm) -- 
	node[above=0mm,pos=0.5]{\rom{4}} (6.85,14mm);
	
	\foreach \x in {0.25,2.35,4.15,5.65,6.85}{
		\draw[gray] (\x,12mm) -- (\x,16mm);
	}
	\draw[gray] (7.2,-1mm) -- (7.2,1mm);
	\draw[-latex] (0.25,10mm) -- node[above=5mm]{} (0.25,0mm);
	\draw[-latex] (6.85,10mm) -- node[above=5mm]{} (6.85,0mm);
	\draw[-latex] (4.15,0mm) -- node[above=5mm]{} (4.15,10mm);
	\draw[-latex,color=blue] (2.35,10mm) -- node[above=5mm]{} (2.35,0mm);
	\draw[-latex,color=blue] (5.65,0mm) -- node[above=5mm]{} (5.65,10mm);
	
	\draw[<->,color=black] (0.3,-2mm) -- node[below=0mm,pos=0.5]{Type-$(1-\alpha)$ frame} (6.9,-2mm);
	
	\draw[black] (0.3,-1mm) -- (0.3,-3mm);
	\draw[black] (6.9,-1mm) -- (6.9,-3mm);
	
	\node[] at (-1,-0.85) {$\alpha$-OPT};
	
	\filldraw[fill=black!] (0.25,-1) rectangle (4.15,-0.7);
	\filldraw[fill=white!40] (4.15,-1) rectangle (6.85,-0.7);
	\node[] at (-1,-1.4) {$\alpha$-RR};
	\filldraw[fill=white!40] (0.25,-1.5) rectangle (2.4,-1.2);
	\filldraw[fill=blue!] (2.4,-1.5) rectangle (5.7,-1.2);
	\filldraw[fill=white!40] (5.7,-1.5) rectangle (6.85,-1.2);

	\end{tikzpicture}
	\caption{Illustration showing 
	a fetch and eviction by $\alpha$-OPT and  $\alpha$-RR in a Type-$(1-\alpha)$ frame. 
	Downward arrows represent fetches, upward arrows indicate evictions. 
	Black and blue arrows correspond to the $\alpha$-OPT and $\alpha$-RR policies respectively. 
	The two bars below the time-line indicate the state of the system under 
	$\alpha$-OPT and $\alpha$-RR. The solid black and solid blue portions represent the 
	intervals during with $\alpha$-OPT and $\alpha$-RR host the service respectively.}
	\vspace{15pt}
\label{fig:a_OPT_RR_frame_proofoutline}
\end{figure}
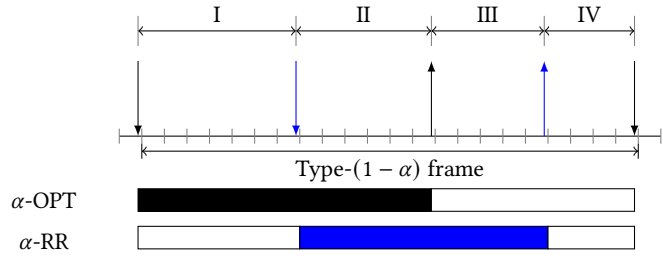
Now we consider any  Type-$\alpha$ frame that is when a  frame starts with $\alpha$ download of service by $\alpha$-OPT. This frame may contains many sub-frames of  Type-$(1-\alpha)$.
We refer to them as Type-$(1-\alpha)$ sub-frames. Let the number of such sub-frames  be $N_{1-\alpha}$. 
We divide the cost incurred $\alpha$-OPT in this frame into two parts
\begin{itemize}
 \item [--] $SC^{\text{$\alpha$-OPT}}_{\alpha}(i)$- Cost incurred by $\alpha$-OPT in the frame $i$ except in the region during Type-$(1-\alpha)$ sub-frames in it.
 \item [--] $SC^{\text{$\alpha$-OPT}_j}_{1-\alpha}(i)$- Cost incurred by $\alpha$-OPT in the $j^\text{th}$ Type-$(1-\alpha)$ sub-frame in  Type-$\alpha$ frame $i$.
\end{itemize}
Therefore, $$C^{\text{$\alpha$-OPT}}(i)=SC^{\text{$\alpha$-OPT}}_{\alpha}(i)+\displaystyle\sum_{j=1}^{N_{1-\alpha}} SC^{\text{$\alpha$-OPT}_j}_{1-\alpha}(i).$$
Similarly, We divide the cost incurred $\alpha\text{-RR}$ in this frame into two parts
\begin{itemize}
 \item [--] $SC^{\alpha\text{-RR}}_{\alpha}(i)$- Cost incurred by $\alpha\text{-RR}$ in frame $i$ except in the region during  Type-$(1-\alpha)$ sub-frames in it.
 \item [--] $SC^{\alpha\text{-RR}_j}_{1-\alpha}(i)$- Cost incurred by $\alpha\text{-RR}$ in the $j^\text{th}$ Type-$(1-\alpha)$ sub-frame in frame $i$.
\end{itemize}
Therefore, $$C^{\alpha\text{-RR}}(i)=SC^{\alpha\text{-RR}}_{\alpha}(i)+\displaystyle \sum_{j=1}^{N_{1-\alpha}} SC^{\alpha\text{-RR}_j}_{1-\alpha}(i).$$ 

By Lemmas \ref{lem:RR_slots}, \ref{lem:RR_goback_time} and \ref{lem:download_alpha}, $\alpha\text{-RR}$ downloads only $\alpha$ fraction of service at least after $\frac{\alpha M}{1-g(\alpha)-\alpha c_{\text{min}}}$ time-slots from the beginning of the frame. So we can write 
\begin{align}
&SC^{\text{$\alpha$-OPT}}_{\alpha}(i)\geq \alpha M+\alpha c_{\text{min}}\frac{\alpha M}{1-g(\alpha)-\alpha c_{\text{min}}}\nonumber\\
\label{ineq:bound1_OPTalpha}
&\implies M\alpha\leq \frac{1-g(\alpha)-\alpha c_{\text{min}}}{1-g(\alpha)}SC^{\text{$\alpha$-OPT}}_{\alpha}(i).
\end{align}
By using Lemmas \ref{lem:max_requests} and \ref{lem:min_requests}, we get 
\begin{align}
SC^{\alpha\text{-RR}}_{\alpha}(i)-&SC^{\text{$\alpha$-OPT}}_{\alpha}(i)\nonumber\\
&\leq  2M\alpha+1-g(\alpha)+\alpha (c_{\text{max}}-c_{\text{min}})\nonumber\\
&\leq  3M\alpha+1-g(\alpha)+\alpha.
\label{ineq:bound1_RRsubalpha}
\end{align}

Using inequalities \ref{ineq:bound2_RRgeneral}, \ref{ineq:bound1_OPTalpha} and \ref{ineq:bound1_RRsubalpha} we write,
\begin{align}
 C^{\alpha\text{-RR}}(i)-&C^{\alpha\text{-OPT}}(i)\nonumber\\
 &\leq \left(3+\dfrac{1}{M}+\dfrac{1-g(\alpha)}{M\alpha}\right)SC^{\alpha-\text{OPT}}_{\alpha}(i)\nonumber\\
&\hspace{8em} +(3+\frac{2}{M}) \displaystyle\sum_{j=1}^{N_{1-\alpha}} SC^{\alpha-\text{OPT}_j}_{1-\alpha}(i)\nonumber\\
 C^{\alpha\text{-RR}}(i) &\leq \left(4+\frac{1}{M}+\max\left\{\dfrac{1}{M},\dfrac{1-g(\alpha)}{M\alpha}\right\}\right) C^{\alpha\text{-OPT}}(i).
 \label{ineq:bound1_RRalpha}
\end{align}
Hence we conclude that in every  Type-$\alpha$ frame $\alpha\text{-RR}$ is $\left(4+\frac{1}{M}+\max\left\{\dfrac{1}{M},\dfrac{1-g(\alpha)}{M\alpha}\right\}\right)$-optimal.

Now we consider any Type-$1$ frame that is when a  frame starts with full download of service by $\alpha$-OPT.  This frame may contains  an eviction of  $(1-\alpha)$ fraction of service by  $\alpha$-OPT followed by many sub-frames which start with the download of $(1-\alpha)$ fraction of service by  $\alpha$-OPT. Let the number of such sub-frames be $N_{1-\alpha}$. 
We divide the cost incurred $\alpha$-OPT in this frame into two parts
\begin{itemize}

 \item [--] $SC^{\text{$\alpha$-OPT}_j}_{1-\alpha}(i)$- Cost incurred by $\alpha$-OPT in the $j^\text{th}$ Type-$(1-\alpha)$ sub-frame in frame $i$.
 \item [--] $SC^{\text{$\alpha$-OPT}}_{1}(i)$- Cost incurred by $\alpha$-OPT in frame $i$ except in the region during $N_{1-\alpha}$ Type-$(1-\alpha)$ sub-frames.
\end{itemize}
Therefore, 
\begin{align*}
C^{\text{$\alpha$-OPT}}(i)&=SC^{\text{$\alpha$-OPT}}_{1}(i)+\displaystyle\sum_{j=1}^{N_{1-\alpha}} SC^{\text{$\alpha$-OPT}_j}_{1-\alpha}(i).
\end{align*}

Similarly, We divide the cost incurred $\alpha\text{-RR}$ in this frame into two parts
\begin{itemize}
 
 \item [--] $SC^{\alpha\text{-RR}_j}_{1-\alpha}(i)$- Cost incurred by $\alpha\text{-RR}$ in the $j^\text{th}$ {\it $(1-\alpha)$- sub-frame} in frame $i$.
 \item [--] $SC^{\alpha\text{-RR}}_{1}(i)$- Cost incurred by $\alpha\text{-RR}$ in frame $i$ except in the region during $N_{1-\alpha}$ sub-frames.
\end{itemize}
Therefore, $$C^{\alpha\text{-RR}}(i)=SC^{\alpha\text{-RR}}_{1}(i)+\displaystyle \sum_{j=1}^{N_{1-\alpha}} SC^{\alpha\text{-RR}_j}_{1-\alpha}(i).$$ 
By Lemmas \ref{lem:RR_slots}, \ref{lem:RR_goback_time} and \ref{lem:download_alpha}, $\alpha\text{-RR}$ downloads full service at least after $\frac{M}{1-c_{\text{min}}}$ time-slots from the beginning of the frame $i$. So we can write 
\begin{align}
&SC^{\text{$\alpha$-OPT}}_{1}(i)\geq  M+ c_{\text{min}}\frac{ M}{1- c_{\text{min}}}\nonumber\\
&\implies M\leq (1- c_{\text{min}})SC^{\text{$\alpha$-OPT}}_{1}(i).
\label{ineq:bound1_OPTfull}
\end{align}
Note that the last eviction by $\alpha$-OPT in frame $i$ is either eviction of $\alpha$ fraction of service or eviction of full service.
In either case the difference $SC^{\alpha\text{-RR}}_{1}(i)-SC^{\text{$\alpha$-OPT}}_{1}(i)$ is at most $2M+c_{\text{max}}-c_{\text{min}}+1$.
This can be verified  by applying Lemmas \ref{lem:max_requests} and \ref{lem:min_requests} to the above two cases separately.

By Lemma \ref{lem:min_requests} and inequality \ref{ineq:bound1_OPTfull} we write,
\begin{align}
 C^{\alpha\text{-RR}}(i)&-C^{\text{$\alpha$-OPT}}(i)\leq 2M+c_{\text{max}}-c_{\text{min}}+1\nonumber\\
 &\hspace{2cm}+(3+\frac{2}{M}) \displaystyle\sum_{j=1}^{N_{1-\alpha}} SC^{\text{$\alpha$-OPT}_j}_{1-\alpha}(i)\nonumber\\
 &<  (3+\frac{2}{M})SC^{\text{$\alpha$-OPT}}_{1}(i)+ (3+\frac{2}{M})\displaystyle \sum_{j=1}^{N_{1-\alpha}} SC^{\text{$\alpha$-OPT}_j}_{1-\alpha}(i)\nonumber\\
 C^{\alpha\text{-RR}}(i) &\leq (4+\frac{2}{M}) C^{\text{$\alpha$-OPT}}(i).
 \label{ineq:bound1_RRalpha}
\end{align}
Hence we conclude that in every Type-1 frame, $\alpha\text{-RR}$ is $(4+\frac{2}{M})$-optimal.

We then stitch results across frames to prove the result.


\subsection{Proof of Theorem \ref{thm:anyonline}}
We provide a sketch of the proof. 
Let $\mathcal{P}$ be any deterministic online policy and  $r_t \in \{0,\alpha$,1\} denote the hosting status under $\mathcal{P}$ in a time-slot $t$.
If under $\mathcal{P}$, $r_0 = 0$, consider the arrival sequence with an arrival in each time-slot until $\mathcal{P}$ fetches the service (entire or partial) and no arrivals thereafter. If $\mathcal{P}$ fetches entire service after $\tau$ time-slots then the cost under $\mathcal{P}$ for this request sequence is, $C^{\mathcal{P}}\geq \tau+M+c_{\text{min}}$.
By Lemma \ref{lem:OPT_slots}, the cost under $\alpha$-OPT for this request sequence is,
\begin{align*}
C^{\alpha\text{-OPT}}=
\begin{cases}
\tau &\text{ if } \tau\leq \frac{M}{1-c_{\text{min}}},\\
M+\tau c_{\text{min}} &\text{ otherwise }.
	\end{cases}
\end{align*}
For $M(1-c_{\text{min}})>c_{\text{min}}^2$, the ratio $\frac{C_{\mathcal{P}}}{C_{\text{OPT}}}$  is minimum at $\tau=\frac{M}{1-c_{\text{min}}}$. For $M(1-c_{\text{min}})\leq c_{\text{min}}^2$, the ratio $\dfrac{C^{\mathcal{P}}}{C^{\alpha-\text{OPT}}}$  is minimum when  $\tau$ is very large.

Let $\overline{\mathcal{P}}$ be any deterministic online policy and  $\overline{r}_t \in \{0,\alpha$,1\} denote the hosting status under $\overline{\mathcal{P}}$ in a time-slot $t$.
If under $\overline{\mathcal{P}}$, $\overline{r}_0 = 1$, consider the arrival sequence with no arrival in each time-slot until $\overline{\mathcal{P}}$ evicts the full service. Then the arrival sequence changes in such a way that there is an arrival in each time-slot. 
If $\overline{\mathcal{P}}$ evicts entire service after $\overline{\tau}$ time-slots then the cost under $\overline{\mathcal{P}}$ for this request sequence is, $C_{\overline{\mathcal{P}}}\geq \overline{\tau} c_{\text{min}}+M+C^{\mathcal{P}}$.
Then the ratio 
\begin{align*}
\dfrac{C^{\overline{\mathcal{P}}}}{C^{\alpha-\text{OPT}}}\geq \dfrac{\overline{\tau} c_{\text{min}}+M}{C^{\alpha-\text{OPT}}}+\dfrac{C^{\mathcal{P}}}{C^{\alpha-\text{OPT}}}.
\end{align*}

Other cases can be proved along similar lines.

%% file: appendix_proofs_stochastic.tex
In this section, we discuss the proofs of the results presented in Section \ref{sec:analyticalResults_stochastic}.
We use the following lemmas to prove Theorem \ref{thm:RR_stochastic_theorem}.
 Note that we use Model  \ref{model:random} to prove results in this section. 

\begin{lemma}
	\label{lemma:optimal_causal}
	Let $X_t$ is an indicator of request arriving in time-slot $t$, $p = \mathbb{E}[X_t]$,  
	Let $Z_t$ be the rent cost per time-slot, $\{Z_t\}_{t\geq 1}$ is the sequence of negatively associated random variables and $c= \mathbb{E}[Z_t]$.
	Under Assumption \ref{assum:arrivals_and_rent}, let $\mathbb{E}[ C_t^{\alpha-\text{OPT-ON}}]$ be the cost per time-slot incurred by the $\alpha-$OPT-ON policy. Then,
	\begin{align*}
	\mathbb{E}[ C_t^{\alpha-\text{OPT-ON}}] \geq \min\{c,\alpha c+g(\alpha) p, p\}.
	\end{align*}
\end{lemma}
\begin{proof}
	If full service is hosted  at the edge in time-slot $t$, the expected cost incurred is at least $\mathbb{E}[Z_t]$ = $c $.
	If $\alpha$ fraction of service is hosted  at the edge in time-slot $t$, the expected cost incurred is at least $\mathbb{E}[\alpha Z_t+g(\alpha)X_t]$ = $\alpha c+g(\alpha) p$. 
	If the service is not hosted  at the edge server, the expected cost incurred is at least $p$. This proves the result.
\end{proof}

\begin{lemma}\label{lem:Hoeffding}
	Let $X_t$ is an indicator of request arriving in time-slot $t$, $p = \mathbb{E}[X_t]$.
	Let $Z_t$ be the rent cost per time-slot, $\{Z_t\}_{t\geq 1}$ is the sequence of negatively associated random variables and $c= \mathbb{E}[Z_t]$.  For each of these combinations $(a,B)=(1,1)$, $(a,B)=(\alpha,1-S_l)$, $(a,B)=(1-\alpha,S_l)$, define $b=\mathbb{E}[B]$, $Y_l=BX_l-a Z_l$, and  $Y=\sum\limits_{l=t-\tau+1}^t Y_l$ then $Y$ satisfies,
	
	\noindent for $((a c-b p)\tau+a M) > 0$,,
	\begin{align*}
	\mathbb{P}\left(Y \geq a M\right)  &\leq \exp\left(-2\frac{((a c-b p)\tau+a M)^2}{\tau (b+a c_{\text{max}}-a c_{\text{min}})^2}\right), 
	\end{align*}
	and for $(b p-a c) \tau +a M > 0$,
	\begin{align*}
	\mathbb{P}\left(Y \leq a\tau c-a M\right)  &\leq \exp\left(-2\frac{((bp-ac)\tau+aM)^2}{\tau (b+ac_{\text{max}}-ac_{\text{min}})^2}\right). 
	\end{align*}
\end{lemma}
\begin{proof} Using total expectation rule we have,
\begin{align*}
    \mathbb{E}[X_l S_l]&=\mathbb{E}[X_l S_l|X_l=1]p+\mathbb{E}[X_l S_l|X_l=0](1-p)=pg(\alpha).
\end{align*}
Using i.i.d. condition of $\{X_t\}_{t\geq 1}$  and  negatively associativity of $\{Z_t\}_{t\geq 1}$, it follows that for $s>0$, $\mathbb{E}[\exp(sY)]\leq \prod\limits_{l=t-\tau+1}^t \mathbb{E}[\exp(sY_l)]$.  Moreover,	 $Y_l\in [ -ac_{\text{max}}, 1-ac_{\text{min}}]$. Then the result  follows by Hoeffding's inequality  \cite{hoeffding1994probability, wajc2017negative}. 
\end{proof}

\begin{lemma}
	\label{lemma:difference_RRstochastic}
		Let	$\Lambda_t^{\mathcal{P}} = \mathbb{E}[C_t^{\mathcal{P}} - C_t^{\alpha-\text{OPT-ON}}]$. Under Assumption \ref{assum:arrivals_and_rent}, 
	\begin{itemize}
		\item[--] Case $\frac{\alpha c}{1-g(\alpha)}<p<\frac{(1-\alpha)c}{g(\alpha)}$:
		\begin{align*}
			\Lambda_t^{\alpha-\text{RR}} &\leq \min_{\lambda: \lambda>1 \text{ and } t >\lambda\widetilde{M}_f}f(\lambda,M,p,c,\alpha,g(\alpha)).
		\end{align*}
		\item[--] Case $p>\max\{c,\frac{(1-\alpha) c}{g(\alpha)}\}$:
		\begin{align*}
			\Lambda_t^{\alpha-\text{RR}} &\leq \min_{\lambda: \lambda>1 \text{ and } t >\lambda\widetilde{M}_q}q(\lambda,M,p,c,\alpha,g(\alpha)).
		\end{align*}
		\item[--] Case $p<\min\{c,\frac{\alpha c}{1-g(\alpha)}\}$:
		\begin{align*}
			\Lambda_t^{\alpha-\text{RR}} &\leq \min_{\lambda: \lambda>1 \text{ and } t >\lambda\widetilde{M}_h}h(\lambda,M,p,c,\alpha,g(\alpha)).
		\end{align*}
	\end{itemize}
\end{lemma}

\subsection{Proof of Theorem \ref{thm:RR_stochastic_theorem}}
\begin{proof}
Case 1: When $\frac{\alpha c}{1-g(\alpha)}<p<\frac{(1-\alpha)c}{g(\alpha)}$.\\
Let $\widetilde{M_f}=\max\Bigg\{\left\lceil\frac{ M\alpha}{p(1-g(\alpha))-\alpha c}\right\rceil, \left\lceil\frac{ M(1-\alpha)}{(1-\alpha) c-pg(\alpha)}\right\rceil\Bigg\}$.
We define the events\\
$$A_{t_1,t_2}:\displaystyle\sum_{l=t_1}^{t2}X_l<\alpha \displaystyle\sum_{l=t_1}^{t2}Z_l+\displaystyle\sum_{l=t_1}^{t2}X_l S_l-\alpha M,$$ 
$$A_{\tau} =\displaystyle\bigcup_{t_1=1}^{\tau} A_{t_1,\tau}, A =\displaystyle\bigcup_{\tau=t-\lambda \widetilde{M} }^{t-1} A_{\tau},$$

 $$B_{t_1,t_2}:\displaystyle\sum_{l=t_1}^{t2}S_l X_l\geq (1-\alpha) \displaystyle\sum_{l=t_1}^{t2} Z_l+(1-\alpha) M,$$ 
 $$B_{\tau} =\displaystyle\bigcup_{t_1=1}^{\tau} B_{t_1,\tau}, B =\displaystyle\bigcup_{\tau=t-\lambda \widetilde{M}}^{t-1} B_{\tau},$$
 %
 
 
 $$E:\displaystyle\sum_{l=t-\lambda \widetilde{M}}^{t-1}(S_l X_l-(1-\alpha)Z_l)+(1-\alpha) M< 0,$$
 $$F:\displaystyle\sum_{l=t-\lambda \widetilde{M}}^{t-1}((1-S_l)X_l-\alpha Z_l)\geq  \alpha M.$$\\
 Using Hoeffding's inequality,
 \begin{align*}
\mathbb{P}(A_{t_1,t_2})&\leq \exp\left(\frac{-2  (p(1-g(\alpha))-\alpha c)(t_2-t_1+1)+\alpha M)^2}{(t_2-t_1+1)(1+\alpha c_{\text{max}}-\alpha c_{\text{min}})^2}\right)\leq \delta_A \exp\left(\frac{-2  (p(1-g(\alpha))-\alpha c))^2 (t_2-t_1+1)}{(1+\alpha c_{\text{max}}-\alpha c_{\text{min}})^2}\right), 
 \end{align*}
 where $\delta_A=\exp\left(\frac{-4  (p(1-g(\alpha))-\alpha c)) \alpha M}{(1+\alpha c_{\text{max}}-\alpha c_{\text{min}})^2}\right)$. Note that since $p(1-g(\alpha))>\alpha c$, the quantity $\delta_A<1$ and it decreases exponentially with increase in $M$. 
 
 Using Union bound,
 \begin{align*}
\mathbb{P}(A_{\tau})&\leq \displaystyle\sum_{t_1=1}^{\tau} \mathbb{P}(A_{t_1,\tau})\leq \displaystyle\sum_{t_1=1}^{\tau-\frac{M}{c_{\text{max}}}} \mathbb{P}(A_{t_1,\tau})\leq \delta_A \displaystyle\sum_{l=0}^{\infty} \exp\left(\frac{-2  (p(1-g(\alpha))-\alpha c))^2 (\frac{M}{c_{\text{max}}}+1+l)}{(1+\alpha c_{\text{max}}-\alpha c_{\text{min}})^2}\right)\\
&\leq \delta_A \frac{\exp\left(-2 \frac{M+c_{\text{max}}}{c_{\text{max}}} \frac{(p(1-g(\alpha))-\alpha c)^2}{(1+\alpha c_{\text{max}}-\alpha c_{\text{min}})^2}\right)}{1-\exp\left(-2 \frac{(p(1-g(\alpha))-\alpha c)^2}{(1+\alpha c_{\text{max}}-\alpha c_{\text{min}})^2}\right)}.
 \end{align*}
 \begin{align}
  \mathbb{P}(A)\leq \lambda \widetilde{M} \mathbb{P}(A_{\tau})\leq \lambda \widetilde{M}  \frac{\delta_A\exp\left(-2 \frac{M+c_{\text{max}}}{c_{\text{max}}} \frac{(p(1-g(\alpha))-\alpha c)^2}{(1+\alpha c_{\text{max}}-\alpha c_{\text{min}})^2}\right)}{1-\exp\left(-2 \frac{(p(1-g(\alpha))-\alpha c)^2}{(1+\alpha c_{\text{max}}-\alpha c_{\text{min}})^2}\right)}.
 \label{eq_eventA1}
 \end{align}
 
 Along similar lines we prove the following bounds.
 \begin{align}
  \mathbb{P}(B)&\leq \lambda \widetilde{M} \mathbb{P}(B_{\tau})\leq \lambda \widetilde{M}\frac{\delta_B\exp\left(-2 (\frac{(1-\alpha) M}{1-(1-\alpha) c_{\text{min}}}+1) \frac{((1-\alpha) c-pg(\alpha))^2}{(1+(1-\alpha) (c_{\text{max}}- c_{\text{min}}))^2}\right)}{1-\exp\left(-2 \frac{((1-\alpha) c-pg(\alpha))^2}{(1+(1-\alpha) (c_{\text{max}}- c_{\text{min}}))^2}\right)},
 \label{eq_eventB1}
 \end{align}
where $\delta_B=\exp\left(\frac{-4  ((1-\alpha) c-pg(\alpha))) (1-\alpha) M}{(1+(1-\alpha) (c_{\text{max}}- c_{\text{min}}))^2}\right)$ which is less than one and decreases exponentially with $M$.


\begin{align}
  \mathbb{P}(E^c)&\leq\exp\left(  \frac{-2 (\lambda \widetilde{M}((1-\alpha) c-pg(\alpha))-(1-\alpha)M)^2}{\lambda \widetilde{M}(1+(1-\alpha) (c_{\text{max}}-c_{\text{min}}))^2}\right)\leq\exp\left(  \frac{-2 (\lambda (1-\alpha)M-(1-\alpha)M)^2}{\lambda \widetilde{M}(1+(1-\alpha) (c_{\text{max}}-c_{\text{min}}))^2}\right)\nonumber\\
  &\leq\exp\left(  \frac{-2 (\lambda-1)^2 M^2(1-\alpha)^2}{\lambda \widetilde{M}(1+(1-\alpha) (c_{\text{max}}-c_{\text{min}}))^2}\right).
  \label{eq_eventE1}
 \end{align}
  Similarly,
 \begin{align}
  \mathbb{P}(F^c)\leq\exp\left(  \frac{-2 (\lambda-1)^2 M^2\alpha^2}{\lambda \widetilde{M}(1+\alpha (c_{\text{max}}-c_{\text{min}}))^2}\right).
  \label{eq_eventF1}
 \end{align}
Note that the right hand side of inequalities \eqref{eq_eventE1}, \eqref{eq_eventF1} diminishes exponentially  increase in $M$.

 \begin{itemize}
 	\item[--]If $\alpha$ fraction of service is  at the edge at time $t-\lambda \widetilde{M}$ then $A^c$ implies $\text{totalCost}(R_0^{n},I_m) > \text{totalCost}(R_{\alpha}^{n},I_m)$, for any $n\in [1, t-1]$ and $m\in [n, t-1]$. In the same case, $B^c$ implies $\text{totalCost}(R_1^{n},I_m) > \text{totalCost}(R_{\alpha}^{n},I_m)$, for any $n\in [1, t-1]$ and $m\in [n, t-1]$. Thus $A^c\cap B^c$ ensures $\alpha$ fraction of service will be at the edge at  time-slot $t$.
 	\item[--] If no fraction of service is at the edge at  time-slot $t-\lambda \widetilde{M}$, 
 and  $\alpha$ fraction of service is fetched during $t-\lambda \widetilde{M}+1\leq\tilde{\tau}\leq t-2$ then $A^c$ implies $\text{totalCost}(R_0^{n},I_m) > \text{totalCost}(R_{\alpha}^{n},I_m)$, for any $n\in [t-\lambda \widetilde{M}+1, t-1]$ and $m\in [n, t-1]$. In the same case, $B^c$ implies $\text{totalCost}(R_1^{n},I_m) > \text{totalCost}(R_{\alpha}^{n},I_m)$ ,for any $n\in [t-\lambda \widetilde{M}+1, t-1]$ and $m\in [n, t-1]$. Thus $A^c\cap B^c$ ensures $\alpha$ fraction of service will be at the edge at  time-slot $t$.
 
 \item[--] If full service is at the edge at time-slot $t-\lambda \widetilde{M}$ and only $1-\alpha$ fraction of service is evicted during $t-\lambda \widetilde{M}+1\leq\tilde{\tau}\leq t-2$ $A^c$ implies $\text{totalCost}(R_0^{n},I_m) > \text{totalCost}(R_{\alpha}^{n},I_m)$ for any $n\in [t-\lambda \widetilde{M}+1, t-1]$ and $m\in [n, t-1]$. In the same case, $B^c$ implies $\text{totalCost}(R_1^{n},I_m) > \text{totalCost}(R_{\alpha}^{n},I_m)$, for any $n\in [t-\lambda \widetilde{M}+1, t-1]$ and $m\in [n, t-1]$. Thus $A^c\cap B^c$ ensures $\alpha$ fraction of service will be at the edge at  time-slot $t$.
 
 	\item[--] If no service is at the edge at  time-slot $t-\lambda \widetilde{M}$ and no amount of service is fetched $t-\lambda \widetilde{M}+1\leq\tilde{\tau}\leq t-2$, then the event $F$ implies $\text{totalCost}(R_0^{t-\lambda \widetilde{M}},I_{t-1}) > \text{totalCost}(R_{\alpha}^{t-\lambda \widetilde{M}},I_{t-1})$, that is $\alpha$ fraction of service will be at the edge at  time-slot $t$.
 	\item[--] If service is not hosted at  time-slot $t-\lambda \widetilde{M}$ and full service is fetched during $t-\lambda \widetilde{M}+1\leq\tilde{\tau}\leq t-2$ then the event $E$ implies $\text{totalCost}(R_1^{t-\lambda \widetilde{M}},I_{t-1}) > \text{totalCost}(R_{\alpha}^{t-\lambda \widetilde{M}},I_{t-1})$, that is $\alpha$ fraction of service will be at the edge at  time-slot $t$.
 	
 	\item[--] If full service is at the edge at time-slot $t-\lambda \widetilde{M}$ and not evicted till  $t-2$ then the event $E$ implies $\text{totalCost}(R_1^{t-\lambda \widetilde{M}},I_{t-1}) > \text{totalCost}(R_{\alpha}^{t-\lambda \widetilde{M}},I_{t-1})$, that is $\alpha$ fraction of service will be at the edge at  time-slot $t$.
 \end{itemize}
 \begin{align*}
&\mathbb{P}(A^c\cap B^c  \cap E\cap F)\geq 1-\mathbb{P}(A)-\mathbb{P}(B)-\mathbb{P}(E^c)-\mathbb{P}(F^c).
 \end{align*}
From the inequalities \eqref{eq_eventA1}, \eqref{eq_eventB1}, \eqref{eq_eventE1} and \eqref{eq_eventF1}, we see that $\mathbb{P}(A^c\cap B^c  \cap E\cap F)$ approaches unity as the value of $M$ increases.\\

Conditioned on $G = A^c\cap B^c  \cap E\cap F$, the $\alpha$ fraction of service is hosted  at the edge during time-slot $t$. 
The expected cost incurred by the $\alpha-$RR policy is
	$
	\mathbb{E}[C_t^{\alpha-\text{RR}}] 
	=  \mathbb{E}[C_t^{\alpha-\text{RR}}|G] \mathbb{P}(G) +\mathbb{E}[C_t^{\alpha-\text{RR}}|G^c] \times \mathbb{P}(G^c).
	$
	Note that, 
	$
	\mathbb{E}[C_t^{\alpha-\text{RR}}|G] = \mathbb{E}[ C_t^{\alpha-\text{OPT-ON}}] , \ \mathbb{E}[C_t^{\alpha-\text{RR}}|G^c] \leq \max\{M+p, M+c\}.
	$
Therefore,
	\begin{align*}
	\mathbb{E}[C_t^{\alpha-\text{RR}}] -\mathbb{E}[ C_t^{\alpha-\text{OPT-ON}}]
	= & \mathbb{E}[C_t^{\alpha-\text{RR}}|G] (\mathbb{P}(G)-1)\max\{M+p, M+c\} \mathbb{P}(G^c) 
	\leq   \max\{M+p, M+c\} \mathbb{P}(G^c).
	\end{align*}
	\begin{align}
	\frac{\mathbb{E}[C_t^{\alpha-\text{RR}}]}{\mathbb{E}[ C_t^{\alpha-\text{OPT-ON}}]}
	\leq & 1+ \max\{M+p, M+c\} \frac{\mathbb{P}(G^c)}{\alpha c+g(\alpha) p} \nonumber \\
	\end{align}

Case 2: When $p>\max\{c,\frac{(1-\alpha) c}{g(\alpha)}\}$\\
Let $\widetilde{M}=\max\Bigg\{\frac{M}{p- c},  \left\lceil\frac{ M(1-\alpha)}{pg(\alpha)-(1-\alpha) c}\right\rceil\Bigg\}$.
We define the events $$A_{t_1,t_2}:\displaystyle\sum_{l=t_1}^{t2}X_l +M< \displaystyle\sum_{l=t_1}^{t2}Z_l,$$ 
$$A_{\tau} =\displaystyle\bigcup_{t_1=1}^{\tau} A_{t_1,\tau}, A =\displaystyle\bigcup_{\tau=t-\lambda \widetilde{M}}^{t-1} A_{\tau},$$
$$B_{t_1,t_2}:\displaystyle\sum_{l=t_1}^{t2}S_lX_l +(1-\alpha)M<  (1-\alpha)\displaystyle\sum_{l=t_1}^{t2}Z_l,$$
$$B_{\tau} =\displaystyle\bigcup_{t_1=1}^{\tau} B_{t_1,\tau}, B =\displaystyle\bigcup_{\tau=t-\lambda \widetilde{M}}^{t-1} B_{\tau},$$
$$D:\displaystyle\sum_{l=t-\lambda \widetilde{M}}^{t-1}X_l \geq  \displaystyle\sum_{l=t-\lambda \widetilde{M}}^{t-1}Z_l+M,$$ 
$$E:\displaystyle\sum_{l=t-\lambda \widetilde{M}}^{t-1}S_lX_l \geq  (1-\alpha)\displaystyle\sum_{l=t-\lambda \widetilde{M}}^{t-1}Z_l +(1-\alpha)M.$$
  
 Using Hoeffding's inequality,
 \begin{align*}
\mathbb{P}(A_{t_1,t_2})&\leq \exp\left(\frac{-2  ((p- c)(t_2-t_1+1)+ M)^2}{(t_2-t_1)(1+ c_{\text{max}}- c_{\text{min}})^2}\right)\leq \delta_A \exp\left(\frac{-2  (p- c)^2 (t_2-t_1+1)}{(1+ c_{\text{max}}- c_{\text{min}})^2}\right), 
 \end{align*}
where $\delta_A=\exp\left(\frac{-4  (p- c) \alpha M}{(1+ c_{\text{max}}-c_{\text{min}})^2}\right)$. Note that since $p> c$, the quantity $\delta_A<1$ and it decreases exponentially with increase in $M$.   
Using Union bound,
 \begin{align*}
\mathbb{P}(A_{\tau})&\leq \displaystyle\sum_{t_1=1}^{\tau} \mathbb{P}(A_{t_1,\tau})\leq \displaystyle\sum_{t_1=1}^{\tau-\frac{M}{c_{\text{max}}}} \mathbb{P}(A_{t_1,\tau})\leq \delta_A \displaystyle\sum_{l=0}^{\infty} \exp\left(\frac{-2  (p- c))^2 (\frac{M}{c_{\text{max}}}+1+l)}{(1+ c_{\text{max}}- c_{\text{min}})^2}\right)\\
&\leq \delta_A \frac{\exp\left(-2 \frac{M+c_{\text{max}}}{c_{\text{max}}} \frac{(p- c)^2}{(1+ c_{\text{max}}-\alpha c_{\text{min}})^2}\right)}{1-\exp\left(-2 \frac{(p- c)^2}{(1+ c_{\text{max}}- c_{\text{min}})^2}\right)}.
 \end{align*}
 \begin{align}
  \mathbb{P}(A)\leq \lambda \widetilde{M} \mathbb{P}(A_{\tau})\leq \lambda \widetilde{M}  \delta_A \frac{\exp\left(-2 \frac{M+c_{\text{max}}}{c_{\text{max}}} \frac{(p- c)^2}{(1+ c_{\text{max}}-\alpha c_{\text{min}})^2}\right)}{1-\exp\left(-2 \frac{(p- c)^2}{(1+ c_{\text{max}}- c_{\text{min}})^2}\right)}.
 \label{eq_eventA2}
 \end{align}

Along similar lines we prove the following bounds.
 \begin{align}
  \mathbb{P}(B)&\leq \lambda \widetilde{M} \mathbb{P}(B_{\tau})\leq \lambda \widetilde{M}\frac{\delta_B\exp\left(-2 \frac{M+c_{\text{max}}}{c_{\text{max}}} \frac{(pg(\alpha)-(1-\alpha) c)^2}{(1+(1-\alpha) (c_{\text{max}}- c_{\text{min}}))^2}\right)}{1-\exp\left(-2 \frac{(pg(\alpha)-(1-\alpha) c)^2}{(1+(1-\alpha) (c_{\text{max}}- c_{\text{min}}))^2}\right)},
 \label{eq_eventB2}
 \end{align}
where $\delta_B=\exp\left(\frac{-4  (pg(\alpha)-(1-\alpha) c)) (1-\alpha) M}{(1+(1-\alpha) (c_{\text{max}}- c_{\text{min}}))^2}\right)$ which is less than one and decreases exponentially with $M$.

\begin{align}
  \mathbb{P}(E^c)&\leq\exp\left(  \frac{-2 (\lambda \widetilde{M}(pg(\alpha)-(1-\alpha) c)-(1-\alpha)M)^2}{\lambda \widetilde{M}(1+(1-\alpha) (c_{\text{max}}-c_{\text{min}}))^2}\right)\nonumber
  \leq\exp\left(  \frac{-2 (\lambda (1-\alpha)M-(1-\alpha)M)^2}{\lambda \widetilde{M}(1+(1-\alpha) (c_{\text{max}}-c_{\text{min}}))^2}\right)\nonumber\\
  &\leq\exp\left(  \frac{-2 (\lambda-1)^2 M^2(1-\alpha)^2}{\lambda \widetilde{M}(1+(1-\alpha) (c_{\text{max}}-c_{\text{min}}))^2}\right).
  \label{eq_eventE2}
 \end{align}
  Similarly,
 \begin{align}
  \mathbb{P}(D^c)\leq\exp\left(  \frac{-2 (\lambda-1)^2 M^2\alpha^2}{\lambda \widetilde{M}(1+\alpha (c_{\text{max}}-c_{\text{min}}))^2}\right).
  \label{eq_eventD2}
 \end{align}
Note that the right hand side of inequalities \eqref{eq_eventE2}, \eqref{eq_eventD2} diminishes exponentially  increase in $M$.
  
 \begin{itemize}
 \item[--] 
 If $\alpha$ full service is at the edge at time $t-\lambda \widetilde{M}$ then $A^c$ implies $\text{totalCost}(R_0^{n},I_m) > \text{totalCost}(R_1^{n},I_m)$, for any $n\in [1, t-1]$ and $m\in [n, t-1]$. In the same case, $B^c$ implies $\text{totalCost}(R_{\alpha}^{n},I_m) > \text{totalCost}(R_1^{n},I_m)$, for any $n\in [1, t-1]$ and $m\in [n, t-1]$. Thus,$A^c\cap B^c$ ensures that full service will be at the edge at  time-slot $t$.

 	\item[--] If full service or $\alpha$ fraction of service is at the edge at  time-slot $t-\lambda \widetilde{M}$, 
 	and  is evicted during $t-\lambda \widetilde{M}+1\leq\tilde{\tau}\leq t-2$ then the event  $D$ implies $\text{totalCost}(R_0^{t-\lambda \widetilde{M}+1},I_{t-2}) > \text{totalCost}(R_1^{t-\lambda \widetilde{M}+1},I_{t-2})$, that is full service will be at the edge at  time-slot $t$.
 	
  \item[--] If  $\alpha$ fraction of service is at the edge at  time-slot $t-\lambda \widetilde{M}$, 
 	and  no service is fetched during $t-\lambda \widetilde{M}+1\leq\tilde{\tau}\leq t-2$ then $E$ ensures that
 	$\text{totalCost}(R_{\alpha}^{t-\lambda \widetilde{M}+1},I_{t-2}) > \text{totalCost}(R_1^{t-\lambda \widetilde{M}+1},I_{t-2})$, that is	full service will be at the edge at  time-slot $t$.
 	
 	\item[--] If  service is not at the edge at  time-slot $t-\lambda \widetilde{M}$, and $\alpha$ fraction of  service is fetched during $t-\lambda \widetilde{M}+1\leq\tilde{\tau}\leq t-2$,
 	then $D\cap E$ ensures that $\text{totalCost}(R_{\alpha}^{t-\lambda \widetilde{M}+1},I_{t-2}) > \text{totalCost}(R_1^{t-\lambda \widetilde{M}+1},I_{t-2})$, that is full service will be at the edge at  time-slot $t$.
 	
 	\item[--] If  service is not at the edge at  time-slot $t-\lambda \widetilde{M}$, and no fraction of  service is fetched during $t-\lambda \widetilde{M}+1\leq\tilde{\tau}\leq t-2$, 	then $D\cap E$ ensures that $\text{totalCost}(R_{\alpha}^{t-\lambda \widetilde{M}+1},I_{t-2}) > \text{totalCost}(R_1^{t-\lambda \widetilde{M}+1},I_{t-2})$, that is full service will be at the edge at  time-slot $t$.
 \end{itemize}

Using union bound and Hoeffding inequality,
\begin{align*}
\mathbb{P}(A^c\cap B^c \cap D\cap E)\geq 1-\mathbb{P}(A)-\mathbb{P}(B)-\mathbb{P}(D^c)-\mathbb{P}(E^c)
\end{align*}
From the inequalities \eqref{eq_eventA2}, \eqref{eq_eventB2}, \eqref{eq_eventD2} and \eqref{eq_eventE2}, we see that $\mathbb{P}(A^c\cap B^c \cap D\cap E)$ approaches unity as the value of $M$ increases.

Conditioned on $G = A^c\cap B^c \cap D\cap E$, the full service is hosted  at the edge during time-slot $t$. 
The expected cost incurred by the $\alpha-$RR policy is
	$
	\mathbb{E}[C_t^{\alpha-\text{RR}}] 
	=  \mathbb{E}[C_t^{\alpha-\text{RR}}|G] \mathbb{P}(G) +\mathbb{E}[C_t^{\alpha-\text{RR}}|G^c] \times \mathbb{P}(G^c).
	$
	Note that, 
	$
	\mathbb{E}[C_t^{\alpha-\text{RR}}|G] = \mathbb{E}[ C_t^{\alpha-\text{OPT-ON}}] , \ \mathbb{E}[C_t^{\alpha-\text{RR}}|G^c] \leq \max\{\alpha M+\alpha c+g(\alpha)p, M+p\}.
	$
Therefore,
	\begin{align*}
	\mathbb{E}[C_t^{\alpha-\text{RR}}]-\mathbb{E}[ C_t^{\alpha-\text{OPT-ON}}]
	= & \mathbb{E}[C_t^{\alpha-\text{RR}}|G] (\mathbb{P}(G)-1)	+ \max\{\alpha M+\alpha c+g(\alpha)p, M+p\} \mathbb{P}(G^c)  \\
	\leq &  \max\{\alpha M+\alpha c+g(\alpha)p, M+p\} \mathbb{P}(G^c).
	\end{align*}
	\begin{align}
	\frac{\mathbb{E}[C_t^{\alpha-\text{RR}}]}{\mathbb{E}[ C_t^{\alpha-\text{OPT-ON}}]}
	\leq & 1+ \max\{\alpha M+\alpha c+g(\alpha)p, M+p\} \frac{\mathbb{P}(G^c)}{c} \nonumber \\
	\end{align}

Case 3: When $p<\min\{c,\frac{\alpha c}{1-g(\alpha)}\}$\\
Let $\widetilde{M}=\max\Bigg\{\frac{M}{c-p},  \left\lceil\frac{ M\alpha}{\alpha c-p(1-g(\alpha))}\right\rceil\Bigg\}$.
We define the events
$$A_{t_1,t_2}: \displaystyle\sum_{l=t_1}^{t_2} X_l \geq  \displaystyle\sum_{l=t_1}^{t_2} Z_l+M,$$ 
$$A^{\tau}=\displaystyle\bigcup_{t_1=1}^{\tau} A_{t_1,\tau},  A_{t}=\displaystyle\bigcup_{\tau=t-\lambda \widetilde{M}}^{t} A^{\tau},$$
$$B_{t_1,t_2}: \displaystyle\sum_{l=t_1}^{t_2} (1-S_l)X_l \geq \alpha \displaystyle\sum_{l=t_1}^{t_2} Z_l+\alpha M,$$ 
$$B^{\tau}=\displaystyle\bigcup_{t_1=1}^{\tau} B_{t_1,\tau}, B_{t}=\displaystyle\bigcup_{\tau=t-\lambda \widetilde{M}}^{t} B^{\tau},$$
$$D: \displaystyle\sum_{l=t-\lambda \widetilde{M}}^{t-1} X_l+M <  \displaystyle\sum_{l=t-\lambda \widetilde{M}}^{t-1} Z_l,$$
$$E: \displaystyle\sum_{l=t-\lambda \widetilde{M}}^{t-1} (1-S_l)X_l+\alpha M < \alpha \displaystyle\sum_{l=t-\lambda \widetilde{M}}^{t-1} Z_l.$$

\begin{align*}
\mathbb{P}(A_{t_1,t_2})&\leq \exp\left(\frac{-2  ((c-p)(t_2-t_1+1)+ M)^2}{(t_2-t_1)(1+ c_{\text{max}}- c_{\text{min}})^2}\right)\leq \delta_A \exp\left(\frac{-2  (c-p)^2 (t_2-t_1+1)}{(1+ c_{\text{max}}- c_{\text{min}})^2}\right), 
 \end{align*}
where $\delta_A=\exp\left(\frac{-4  (c-p) \alpha M}{(1+ c_{\text{max}}-c_{\text{min}})^2}\right)$. Note that since $c>p$, the quantity $\delta_A<1$ and it decreases exponentially with increase in $M$.   
Using Union bound,
 \begin{align*}
\mathbb{P}(A^{\tau})&\leq \displaystyle\sum_{t_1=1}^{\tau} \mathbb{P}(A_{t_1,\tau})\leq \displaystyle\sum_{t_1=1}^{\tau-\frac{M}{1-c_{\text{min}}}} \mathbb{P}(A_{t_1,\tau})\leq \delta_A \displaystyle\sum_{l=0}^{\infty} \exp\left(\frac{-2  (c-p)^2 (\frac{M}{1-c_{\text{min}}}+1+l)}{(1+ c_{\text{max}}- c_{\text{min}})^2}\right)\\
&\leq \delta_A \frac{\exp\left(-2 (\frac{M}{1-c_{\text{min}}}+1) \frac{(c-p)^2}{(1+ c_{\text{max}}-\alpha c_{\text{min}})^2}\right)}{1-\exp\left(-2 \frac{(c-p)^2}{(1+ c_{\text{max}}- c_{\text{min}})^2}\right)}.
 \end{align*}
 \begin{align}
  \mathbb{P}(A_t)\leq \lambda \widetilde{M} \mathbb{P}(A_{\tau})\leq \lambda \widetilde{M}   \frac{\delta_A\exp\left(-2 (\frac{M}{1-c_{\text{min}}}+1) \frac{(c-p)^2}{(1+ c_{\text{max}}-\alpha c_{\text{min}})^2}\right)}{1-\exp\left(-2 \frac{(c-p)^2}{(1+ c_{\text{max}}- c_{\text{min}})^2}\right)}.
 \label{eq_eventA3}
 \end{align}
 
 Along similar lines we prove the following bounds.
 \begin{align}
  \mathbb{P}(B_t)&\leq \lambda \widetilde{M} \mathbb{P}(B_{\tau})\leq \lambda \widetilde{M}\frac{\delta_B\exp\left(-2 (\frac{\alpha M}{1-g(\alpha)-\alpha c_{\text{min}}}+1) \frac{(c\alpha -p(1-g(\alpha)))^2}{(1+\alpha (c_{\text{max}}- c_{\text{min}}))^2}\right)}
  {1-\exp\left(-2 \frac{(c\alpha -p(1-g(\alpha)))^2}{(1+\alpha (c_{\text{max}}- c_{\text{min}}))^2}\right)},
 \label{eq_eventB3}
 \end{align}
where $\delta_B=\exp\left(\frac{-4  (\alpha c-p(1-g(\alpha)) \alpha M}{(1+\alpha (c_{\text{max}}- c_{\text{min}}))^2}\right)$ which is less than one and decreases exponentially with $M$.
 
\begin{align}
  \mathbb{P}(E^c)&\leq\exp\left(  \frac{-2 (\lambda \widetilde{M}(\alpha c-p(1-g(\alpha))-\alpha M)^2}{\lambda \widetilde{M}(1+\alpha (c_{\text{max}}-c_{\text{min}}))^2}\right)\leq\exp\left(  \frac{-2 (\lambda \alpha M-\alpha M)^2}{\lambda \widetilde{M}(1+\alpha (c_{\text{max}}-c_{\text{min}}))^2}\right)\leq\exp\left(  \frac{-2 (\lambda-1)^2 M^2\alpha^2}{\lambda \widetilde{M}(1+\alpha (c_{\text{max}}-c_{\text{min}}))^2}\right).
  \label{eq_eventE3}
 \end{align}
  Similarly,
 \begin{align}
  \mathbb{P}(D^c)\leq\exp\left(  \frac{-2 (\lambda-1)^2 M^2}{\lambda \widetilde{M}(1+ (c_{\text{max}}-c_{\text{min}}))^2}\right).
  \label{eq_eventD3}
 \end{align}
Note that the right hand side of inequalities \eqref{eq_eventE3}, \eqref{eq_eventD3} diminishes exponentially  increase in $M$.
 
 \begin{itemize}

 	\item[--] If full service or $\alpha$ fraction of service is at the edge at  time-slot $t-\lambda \widetilde{M}$, and  is not evicted during $t-\lambda \widetilde{M}+1\leq\tilde{\tau}\leq t-2$ then $A^c_t \cap A^c_{t-1}\cap B^c_t \cap B^c_{t-1} \cap D \cap E$ implies $\text{totalCost}(R_0^{n},I_m) < \text{totalCost}(R_1^{n},I_m)$ and $\text{totalCost}(R_0^{n},I_m) < \text{totalCost}(R_{\alpha}^{n},I_m)$ , for any $n\in [1, t-1]$ and $m\in [n, t-1]$. Thus no service will be at the edge at  time-slot $t$.
 	\item[--] If full service or $\alpha$ fraction of service is at the edge at  time-slot $t-\lambda \widetilde{M}$, and  $1-\alpha$ fraction of service  is  evicted during $t-\lambda \widetilde{M}+1\leq\tilde{\tau}\leq t-2$ then the event $E$ implies that $\text{totalCost}(R_0^{t-\lambda \widetilde{M}+1},I_{t-1}) < \text{totalCost}(R_{\alpha}^{t-\lambda \widetilde{M}+1},I_{t-1})$,that is 	no service will be at the edge at  time-slot $t$.
 	\item[--]  If no service  is at the edge at  time-slot $t-\lambda \widetilde{M}$, and  full or $\alpha$ fraction of service  is  fetched during $t-\lambda \widetilde{M}+1\leq\tilde{\tau}\leq t-2$ then $D\cap E$ ensures that $\text{totalCost}(R_0^{t-\lambda\widetilde{M}+1},I_{t-1}) <\text{totalCost}(R_{\alpha}^{t-\lambda \widetilde{M}+1},I_{t-1})$ and\\ $\text{totalCost}(R_0^{t-\lambda\widetilde{M}+1},I_{t-1})< \text{totalCost}(R_1^{t-\lambda \widetilde{M}+1},I_{t-1})$. Thus	no service will be at the edge at  time-slot $t$.
 \end{itemize}
Using union bound and Hoeffding inequality,
\begin{align*}
&\mathbb{P}(A^c_t \cap A^c_{t-1}\cap B^c_t \cap B^c_{t-1} \cap D \cap E)\geq 1-\mathbb{P}(A_t)-\mathbb{P}(A_{t-1})-\mathbb{P}(B_t)-\mathbb{P}(B_{t-1})-\mathbb{P}(B)-\mathbb{P}(D).
\end{align*}

Conditioned on $G = A^c_t \cap A^c_{t-1}\cap B^c_t \cap B^c_{t-1} \cap D \cap E$, the full service is hosted  at the edge during time-slot $t$. 
The expected cost incurred by the $\alpha-$RR policy is
	$
	\mathbb{E}[C_t^{\alpha-\text{RR}}] 
	=  \mathbb{E}[C_t^{\alpha-\text{RR}}|G] \mathbb{P}(G) +\mathbb{E}[C_t^{\alpha-\text{RR}}|G^c] \times \mathbb{P}(G^c).
	$
	Note that, 
	$
	\mathbb{E}[C_t^{\alpha-\text{RR}}|G] = \mathbb{E}[ C_t^{\alpha-\text{OPT-ON}}] , \ \mathbb{E}[C_t^{\alpha-\text{RR}}|G^c] \leq \max\{\alpha M+\alpha c+g(\alpha)p, M+c\}.
	$
Therefore,
	\begin{align*}
	\mathbb{E}[C_t^{\alpha-\text{RR}}] -\mathbb{E}[ C_t^{\alpha-\text{OPT-ON}}]
	= & \mathbb{E}[C_t^{\alpha-\text{RR}}|G] (\mathbb{P}(G)-1)+ \max\{\alpha M+\alpha c+g(\alpha)p, M+c\} \mathbb{P}(G^c)  \\
	\leq &  \max\{\alpha M+\alpha c+g(\alpha)p, M+c\} \mathbb{P}(G^c).
	\end{align*}
	\begin{align}
	\frac{\mathbb{E}[C_t^{\alpha-\text{RR}}]}{\mathbb{E}[ C_t^{\alpha-\text{OPT-ON}}]}
	\leq & 1+ \max\{\alpha M+\alpha c+g(\alpha)p, M+c\} \frac{\mathbb{P}(G^c)}{p} \nonumber \\
	\end{align}
\end{proof}